\documentclass[12pt]{article}

\usepackage[margin=2cm]{geometry}

\usepackage{amsmath,amsthm,amsfonts,amscd,amssymb,bbm,mathrsfs,enumerate,url}
\usepackage{authblk}

\numberwithin{equation}{section}

\newcounter{2dCN}
\newcounter{1dCN}

\usepackage{braket}

\usepackage{forest}
\usepackage{comment}
\usepackage{mathtools}
\usepackage[all]{xy}

\usepackage{graphicx,tikz}
\usepackage{tikz-cd}
\usepackage[colorlinks = true,
            linkcolor = black,
            urlcolor  = blue,
            citecolor = black,
            anchorcolor = black,
            pdfborder={0 0 0}]{hyperref}
\makeatletter
\newcommand*{\textlabel}[2]{%
  \edef\@currentlabel{#1}
  \phantomsection
  #1\label{#2}
}

\usetikzlibrary{matrix,arrows}
\usetikzlibrary{decorations.pathreplacing}
\usetikzlibrary{decorations.pathmorphing}
\usetikzlibrary{patterns,fadings}

\def\semicolon{;}
\def\applytolist#1{
    \expandafter\def\csname multi#1\endcsname##1{
        \def\multiack{##1}\ifx\multiack\semicolon
            \def\next{\relax}
        \else
            \csname #1\endcsname{##1}
            \def\next{\csname multi#1\endcsname}
        \fi
        \next}
    \csname multi#1\endcsname}

\def\calc#1{\expandafter\def\csname c#1\endcsname{{\mathcal #1}}}
\applytolist{calc}QWERTYUIOPLKJHGFDSAZXCVBNM;
\def\bbc#1{\expandafter\def\csname bb#1\endcsname{{\mathbb #1}}}
\applytolist{bbc}QWERTYUIOPLKJHGFDSAZXCVBNM;
\def\bffc#1{\expandafter\def\csname bf#1\endcsname{{\mathbf #1}}}
\applytolist{bffc}QWERTYUIOPLKJHGFDSAZXCVBNMqwertyuiopasdfghjklzxcvbnm;
\def\sfc#1{\expandafter\def\csname s#1\endcsname{{\sf #1}}}
\applytolist{sfc}QWERTYUIOPLKJHGFDSAZXCVBNM;
\def\rmfc#1{\expandafter\def\csname rm#1\endcsname{{\mathrm #1}}}
\applytolist{rmfc}QWERTYUIOPLKJHGFDSAZXCVBNMqwertyuiopasdfghjklzxcvbnm;
\def\frfc#1{\expandafter\def\csname fr#1\endcsname{{\mathfrak #1}}}
\applytolist{frfc}QWERTYUIOPLKJHGFDSAZXCVBNMqwertyuiopasdfghjklzxcvbnm;
\def\scfc#1{\expandafter\def\csname sc#1\endcsname{{\mathscr #1}}}
\applytolist{scfc}QWERTYUIOPLKJHGFDSAZXCVBNMqwertyuiopasdfghjklzxcvbnm;

\DeclareMathOperator{\cAd}{Ad}
\DeclareMathOperator{\id}{id}
\DeclareMathOperator{\supp}{supp}

\DeclareMathOperator{\Dom}{Dom}

\DeclareMathOperator{\End}{End}
\DeclareMathOperator{\Rep}{Rep}

\DeclareMathOperator{\Mob}{M\ddot{o}b}

\DeclareMathOperator{\PSLtwoR}{\mathrm{PSL}_2(\bbR)}

\usepackage{xspace}
\DeclareRobustCommand{\eg}{e.g.\@\xspace}

\DeclareRobustCommand{\cf}{cf.\@\xspace}

\DeclareRobustCommand{\ie}{i.e.\@\xspace}

\def\bb1{\mathbbm{1}}

\def\DiffSone{\mathrm{Diff}_+(S^1)}
\def\uDiffSone{\overline{\mathrm{Diff}_+(S^1)}} 
\def\Mob{\mathrm{M\ddot{o}b}}
\def\uMob{\overline{\Mob}}

\def\<{\langle}
\def\>{\rangle}

\newtheorem{theorem}{Theorem}[section]
\newtheorem{definition}[theorem]{Definition}

\newtheorem{proposition}[theorem]{Proposition}
\newtheorem{lemma}[theorem]{Lemma}
\theoremstyle{remark}
\newtheorem{remark}[theorem]{Remark}
\newtheorem{example}[theorem]{Example}

\title{Rational and non-rational two-dimensional conformal field theories arising from lattices}
\author[1]{Maria Stella Adamo\thanks{{\tt maria.stella.adamo@fau.de}}}
\author[2]{Luca Giorgetti\thanks{{\tt giorgett@mat.uniroma2.it}}}
\author[2]{Yoh Tanimoto\thanks{{\tt hoyt@mat.uniroma2.it}}}

\affil[1]{Department Mathematik, Friedrich-Alexander-Universit\"at Erlangen-N\"urnberg \authorcr
Cauerstrasse 11, 91058 Erlangen, Germany}
\affil[2]{Dipartimento di Matematica, Universit\`a di Roma Tor Vergata,\authorcr
   Via della Ricerca Scientifica 1, I-00133 Roma, Italy}
\date{}

\begin{document}
\maketitle
\begin{center}
 \textit{Dedicated to the memory of Huzihiro Araki.}
\end{center}

\begin{abstract}
For a (finite-dimensional) real Hilbert space $\frh$ and an orthogonal projection $p$,
we consider the associated Heisenberg Lie algebra and the two-dimensional Heisenberg conformal net.
Given an even lattice $Q$ in $\frh$ with respect to the indefinite bilinear form on $\frh$ defined by $p$,
we construct a two-dimensional conformal net $\cA_Q$ extending the Heisenberg conformal net.
Moreover, with a certain discreteness assumption on the spectrum of the extension,
we show that any two-dimensional extension of the Heisenberg conformal net is of the form $\cA_Q$
up to unitary equivalence.

We consider explicit examples of even lattices where $\frh$ is two-dimensional and $p$ is one-dimensional, and we show that the extended net may have
completely rational or non-completely rational chiral (\ie one-dimensional lightray) components,
depending on the choice of lattice.
In the non-rational case, we exhibit the braided equivalence of a certain subcategory of the representation category
of the chiral Heisenberg net corresponding to the two-dimensional lattice extension.

Inspired by the charge and braiding structures of these nets, we construct two-dimensional conformal Wightman fields on
the same Hilbert spaces. We show that, in some cases, these Wightman fields generate the corresponding extended nets.
\end{abstract}

\section{Introduction}
Two-dimensional Conformal Field Theories (CFT) are nowadays studied in various frameworks.
Some of them are specifically developed to describe and encompass conformal covariance, such as Vertex Operator Algebras (VOA) \cite{FLM88VertexOperatorAlgebras, HK07FullFieldAlgebras},
or the Segal axioms \cite{Segal04}. Others have been introduced to deal with more general Quantum Field Theories
and they can be restricted to the conformal setting: such as conformal nets \cite{KL04-1, KL04-2} (fulfilling the Araki--Haag--Kastler axioms \cite{Haag96,Araki99}), or the G{\aa}rding--Wightman axioms \cite{LM75, SW00}.
There is a subclass of conformal field theories that is relatively well-understood, especially in two spacetime dimensions: they are called rational CFTs.
These theories have well-behaved chiral components (described by rational VOAs or chiral conformal nets),
and their representation categories have rigid structures (they are typically fusion, braided and modular).
In the rational setting, the full two-dimensional CFTs extending the tensor product of chiral components 
can be described purely algebraically and categorically using certain commutative algebra objects in the tensor product of their representation categories \cite{DMNO13},
or equivalently, using module categories \cite{Ost03}, or braided equivalences of the chiral components \cite{DNO13}, see \eg \cite{FRS02, SFR06, RFFS07, Kong07, BKL15} and references therein.
On the other hand, there is a vast majority of non-rational CFTs, see \eg \cite{Carpi04}, and their structure remains to be understood.

A representative family of non-rational (chiral or two-dimensional) CFTs is given by realizations of the canonical commutation relations: the so-called Heisenberg algebras. Recently, an algebraic-analytic framework for two-dimensional (full) CFT (called full VOA) has been proposed \cite{Moriwaki23} and the extensions of the Heisenberg algebras have been studied.
In \cite{AMT24OS}, it has been proved that the correlation functions in a full VOA (assuming unitarity and some technical conditions) define Euclidean correlation functions
satisfying the Osterwalder--Schrader axioms, and thus there are corresponding two-dimensional conformal G{\aa}rding--Wightman fields \cite{AGT23Pointed}, by the celebrated Wightman reconstruction theorem.
Therefore, it is natural to address a more direct construction of G{\aa}rding--Wightman fields and conformal nets.
In this work, we construct (and classify) two-dimensional extensions of the Heisenberg models,
both in the conformal net and G{\aa}rding--Wightman framework.

The Heisenberg algebra associated with a finite-dimensional real Hilbert space $\frh$
is the Lie algebra, denoted by $\widehat\frh$,
generated by $\{\pmb{\alpha}(m)\}_{\pmb{\alpha} \in \frh, m \in \bbZ}$ and a central element $K$
with commutation relations $[\pmb{\alpha}(m), \pmb{\beta}(n)] = \delta_{m,-n} m (\pmb{\alpha}, \pmb{\beta})_\frh K$.
By choosing an orthogonal projection $p$, we can declare that $\widehat{p\frh}$ corresponds to the (left) chiral component,
and $\widehat{\bar p\frh}$ to the (right) chiral (\ie antichiral) component.
The algebra $\widehat{p\frh}$ has irreducible modules parametrized by $\pmb{\lambda} \in p\frh$,
thus they are naturally endowed with an additive group structure.
In analogy with the rational case, in order to construct a two-dimensional CFT as an extension of the two-dimensional tensor product of $\widehat{p\frh}$ with $\widehat{\bar p\frh}$, 
one might consider the trivial braided equivalence from
the category of all $\widehat{p\frh}$-modules to that of $\widehat{\bar p\frh}$-modules (if $\dim p = \dim \bar p$ and $p\frh$ and $\bar p\frh$ are
identified).
It turns out that this is not the only option. In addition to these ``diagonal'' extensions,
one can choose a (not necessarily $N$-dimensional) even lattice $Q$ in the $N$-dimensional real space $\frh$
with respect to the indefinite bilinear form defined by $p$, $(\cdot|\cdot) := (p\cdot|p\cdot)_{\frh} - (\bar p\cdot|\bar p\cdot)_{\frh}$,
where only a subset of simple modules appear, and then construct ``non-diagonal" CFT extensions.
Given an arbitrary even lattice $Q$ in $\frh$, we perform this construction first in the two-dimensional conformal net setting, 
and then in the G{\aa}rding--Wightman setting.
We do this by considering the $2$-cocycle $\epsilon$ of the lattice $Q$ \cite{Kac98} and the twisted shift operators.
The interplay between $\epsilon$ and the braiding of the chiral representation category guarantees the locality of the extended CFT.

By looking at special cases, we observe that these families of models accommodate both rational and non-rational
cases. For example, when $\frh$ is two-dimensional ($\frh \cong \bbR^2$) and we divide it into one-dimensional ($\cong \bbR$) orthogonal subspaces
corresponding the chiral and antichiral parts, there are even lattices $Q\subset \frh$ generated by $\frac{R}{\sqrt2}\oplus \frac{R}{\sqrt2}$
and $\frac{R^{-1}}{\sqrt 2}\oplus -\frac{R^{-1}}{\sqrt 2}$ for any $R \in \bbR, R\neq 0$, and the corresponding two-dimensional conformal net, denoted by $\cA_Q$, has rational chiral components properly extending the chiral Heisenberg nets if and only if $R^2 \in \bbQ$. If $R^2 \in \bbR\setminus\bbQ$, the chiral components of $\cA_Q$ remain the chiral Heisenberg nets $\cA_\bbR$, which are not rational.

This non-rational case is of particular interest. Differently from the rational case \cite{DNO13, BKL15},
we are not adding bulk (\ie non-chiral) fields associated with all the irreducible objects in the representation categories of the chiral components. 
Yet, the resulting two-dimensional conformal net $\cA_Q$ appears to be maximally extended 
(\cf our classification result in Section \ref{classification} which holds under a certain technical condition, see below).
This suggests that, even considering proper subcategories of the chiral components, one may obtain maximally extended CFTs.
In the examples above with $R^2 \in \bbR \setminus \bbQ$, we exhibit the braided equivalence corresponding to the (infinite index) extension $\cA_\bbR \otimes \cA_\bbR \subset \cA_Q$. The equivalence is defined on the (strict) tensor subcategory generated by the irreducible representations with charge $\frac{R}{\sqrt 2}$ and $\frac{R^{-1}}{\sqrt 2}$, it is a non-strictly tensor functor, and its non-trivial tensorator is the $2$-cocycle $\epsilon$.
The analysis of these concrete examples opens the way to the study of more general two-dimensional non-rational conformal nets and extensions.

For an arbitrary $N$-dimensional $\frh$, under the condition that the additive subgroup of charges appearing in the extension is
discrete, we classify all possible two-dimensional conformal net extensions of the Heisenberg nets: they are all associated with an even lattice.
We follow the strategy of \cite{BMT88}. That is, we first extract charge-carrying operators
and we put them in a canonical form, then we examine the possible $2$-cocycles arising.
By the discreteness assumption, the set of possible charges forms a lattice.
Locality of the extension imposes that the lattice is even and that the $2$-cocycles are symmetric.
As the symmetric $2$-cocycles of an abelian group are necessarily $2$-coboundaries, we obtain the classification result.

The models constructed here are relatively simple, yet of interest from various points of view.
For example, in the physics literature, models corresponding to different $Q$ (in the same $\frh$)
are believed to be related by the ``current-current'' deformation:
assuming that they are associated with a Lagrangian, different models should be associated
with the Lagrangian perturbed by the current.
On the other hand, when $\frh$ is $2$-dimensional and $p$ is $1$-dimensional as in the examples mentioned earlier,
the models constructed here seem to be related with the massless free field
with the compactified target spaces. See \cite[Section 6]{Moriwaki23} and references therein for the physics literature.
It is an interesting problem to connect these models with the conformal field theories arising from certain Lagrangians,
see \cite{GKR24ProbabilisticLiouville} for recent developments.

Once we have clear charge structures, we proceed to construct Wightman fields.
The fields generating the two-dimensional extensions are labelled by the elements of the lattice,
and the smeared operator creates states with that charge.
In order to show that they generate a two-dimensional conformal net, we need
to verify that they satisfy strong enough bounds, \cf \cite[Lemma 5.2]{AGT23Pointed}.
Such bounds are proved in some cases.

Heisenberg algebras and their extensions are associated with the loop group of the groups $\bbR^N$ or $(S^1)^N$,
respectively (see \cite[Chapter V]{Toledano-LaredoThesis}), and from their vacuum representations one can
construct conformal nets. In this regard, it is important to consider their reflection positive representations (\cf \cite{ANS22ReflectionPositivity, ANS24ReflectionPositivity} for finite-dimensional groups), which should give rise to
their Euclidean counterparts. Such representations will be the subject of future investigations.

This paper is organized as follows.
In Section \ref{conformalnets}, we recall the definitions of two-dimensional conformal nets,
chiral conformal nets on $S^1$, and their representations.
In Section \ref{sec:U(1)}, we consider the basic case where $\frh \cong \bbR$, and we recall the definition of the associated Heisenberg algebra $\widehat \frh$ (the $\rmU(1)$-current algebra) and the corresponding conformal net on $S^1$. 
In Section \ref{sec:Heisenberg}, we generalize it to the finite-dimensional case and recall the (braided and tensor) structure of its modules (representations).
In Section \ref{sec:2d}, we construct the two-dimensional conformal net $\cA_Q$ associated with an even lattice $Q \subset \frh$ with respect to the bilinear form $(\cdot| \cdot)$, twisted by $p$.
In Section \ref{sec:BE}, we study the braided equivalence associated with the extension in a class of examples, where $Q \subset \frh \cong \bbR^2$.
In Section \ref{classification}, we classify two-dimensional conformal nets extending the Heisenberg nets, under the assumption that the representations of the Heisenberg net appearing in the Hilbert space of the extension form a discrete subset of $\frh$.
In Section \ref{wightman}, we construct the corresponding Wightman fields on the Hilbert space of $\cA_Q$.

\section{Preliminaries}\label{preliminaries}

\subsection{Conformal nets}\label{conformalnets}
\subsubsection{Two-dimensional conformal nets}\label{2dconformalnet}
We consider the Araki--Haag--Kastler framework \cite{Haag96, Araki99} for two-dimensional conformal field theory \cite{KL04-2}.
In general, a Araki-Haag-Kastler net is defined on the Minkowski space $\bbR^{1+1}$,
but in the conformal case, one can take more convenient coordinates:
as we did in \cite[Section 2.1]{AGT23Pointed}, we map the two-dimensional Minkowski space
to $(-\pi, \pi)\times (-\pi, \pi)$ by $(a_+, a_-) \mapsto (2\arctan a_+, 2\arctan a_-)$,
where $(a_+, a_-)$ are the lightcone coordinates on $\bbR^{1+1}$ (\cf \cite{KL04-2}).
With this identification, $\uDiffSone\times\uDiffSone$ and $\uMob\times\uMob$ act locally on $\bbR^{1+1}$,
where $\Mob = \PSLtwoR$.
Recall that $\Mob \subset \DiffSone$ and if a unitary projective representation $U$ of $\uDiffSone \times \uDiffSone$ is
restricted to $\uMob\times\uMob$, we may assume that it is a true (non-projective) representation.

A \textbf{conformal net on $\bbR^{1+1}$}, or \textbf{two-dimensional conformal net}, on a Hilbert space $\cH$ is a triple $(\cA, U, \Omega)$
where $\cA$ is an assignment to each open bounded region $O \subset \bbR^{1+1}$ of a von Neumann algebra $\cA(O)$ on $\cH$,
$U$ is a unitary projective representation of $\uDiffSone\times\uDiffSone$ on $\cH$, $\Omega \in \cH$
such that (we assume a stronger condition that assures that the net extends to
the Einstein cylinder \cite[Section 2.3, Proposition 2.5]{AGT23Pointed}, so this is the definition of \cite{KL04-2})
\begin{enumerate}[{(2dCN}1{)}]
 \item \textbf{Isotony: }if $O_1\subset O_2$, then $\cA(O_1)\subset\cA(O_2)$. \label{2dcn:isotony}
 \item \textbf{Locality:} if $O_1$ and $O_2$ are spacelike separated, then $\cA(O_1)\subset\cA(O_2)'$. \label{2dcn:locality}
 \item\textbf{Diffeomorphism covariance:}
 for a bounded region $O \subset \bbR^{1+1}$, there is a neighborhood $\cU$ of the unit element of $\uDiffSone\times \uDiffSone$
 such that if $\gamma \in \cU$ then $\gamma O \subset \bbR^{1+1}$ and
 \[
  U(\gamma)\cA(O)U(\gamma)^*=\cA(\gamma O).
 \]
 Furthermore, if $\supp \gamma$ is spacelike from $O$, then $\cAd U(\gamma)(x) = x$ for $\cA(O)$.
 \label{2dcn:diff}
 \item\textbf{Positivity of energy:}
 the restriction of $U$ to the translation subgroup $\bbR^2 \subset \uMob \times \uMob$
 has the joint spectrum contained in the closed forward light cone
 $\overline {V_+}=\{(a_0,a_1)\in\bbR^{1+1}: a_0^2-a_1^2\geq0, a_0\geq0\}$ (in the original coordinates)
 $=\{(a_+,a_-)\in\bbR^{1+1}: a_\pm \geq0\}$ (in the lightcone coordinates). \label{2dcn:positiveenergy}
 \item \textbf{Vacuum and the Reeh-Schlieder property: }there exists a unique (up to a phase)
 vector $\Omega \in\cH$ such that $U(g)\Omega=\Omega$ for $g\in \uMob\times\uMob$
 and is cyclic for any local algebra, namely $\overline{\cA(O)\Omega}=\cH$. \label{2dcn:vacuum}
 \item \textbf{Extension to the cylinder: } $U(R_{2\pi}\times R_{-2\pi}) = \bb1$. \label{2dcn:cylinder}
 \setcounter{2dCN}{\value{enumi}}
\end{enumerate}

One of the main results of this paper is to construct a family of examples of two-dimensional conformal nets on $\bbR^{1+1}$.
For this purpose, we note that it is enough to construct $\cA(O)$ first for \textbf{double cones}, which correspond
to a set of the form $I_+ \times I_- \subset (-\pi, \pi)\times (-\pi, \pi)$,
where $I_+, I_-$ are relatively compact intervals in $(-\pi, \pi)$,
For such $\cA$, we require the same axioms except that locality means that
$\cA(I_{+,1} \times I_{-,1})$ commutes with $\cA(I_{+,2} \times I_{-,2})$
whenever one of the following conditions holds:
\begin{itemize}
 \item $I_{+,1} < I_{+,2}$ and $I_{-,2} < I_{-,1}$
 \item $I_{+,2} < I_{+,1}$ and $I_{-,1} < I_{-,2}$
\end{itemize}
Note that any of these conditions correspond to spacelike-separated double cones in $\bbR^{1+1}$,
and $\uDiffSone\times\uDiffSone$ map (locally) any double cone to a double cone.

From such $\cA$ for double cones, for any other open set $O \subset (-\pi, \pi)\times (-\pi, \pi)$
with non-trivial spacelike complement, we can extend $\cA$ by
$\cA(O) = \bigvee_{I_+ \times I_- \subset O} \cA(I_+ \times I_-)$,
where $\bigvee_\bullet \cM_\bullet$ denotes the von Neumann algebra generated by $\{\cM_\bullet\}$.
With same $U, \Omega$, the triple $(\cA, U, \Omega)$ with the extended $\cA$ satisfies the axioms of two-dimensional conformal net on $\bbR^{1+1}$.

When $(\cA, U, \Omega)$ is a conformal net on $\bbR^{1+1}$, the net $\cA$ contains nontrivial chiral components \cite{Rehren00}:
they are conformal nets $\cA_+, \cA_-$ on $S^1$ (see Section \ref{S1conformalnet}) such that $\cA$ is an extension of $\cA_+ \otimes \cA_-$
(see Section \ref{1d2dconformalnet}).
In this work, we do not assume that they are the maximal ones given in \cite{Rehren00}. Instead, we start with concrete examples of $\cA_+, \cA_-$
(the Heisenberg nets, Section \ref{sec:Heisenberg-net}) and construct (and classify) the extensions $\cA$.

\subsubsection{Conformal nets on \texorpdfstring{$S^1$}{S1}}\label{S1conformalnet}
Let $\cI$ be the set of non-empty, open, non-dense, proper intervals of $S^1$.
For $I \in \cI$, we denote by $I'$ the interior of $S^1 \setminus I$.
A \textbf{conformal net on $S^1$} on a Hilbert space $\cH_\kappa$ is a triple $(\cA_\kappa, U_\kappa, \Omega_\kappa)$,
where $\cA_\kappa$ (the subscript $\kappa$ is to remind the reader that the net is chiral, \ie defined on $S^1$) 
is an assignment to each $I \in \cI$ of a von Neumann algebra $\cA_\kappa(I)$ on $\cH_\kappa$,
$U_\kappa$ is a unitary projective representation of $\DiffSone$ on $\cH_\kappa$, $\Omega_\kappa \in \cH_\kappa$
such that
\begin{enumerate}[{(1dCN}1{)}]
 \item \textbf{Isotony: }if $I_1\subset I_2$, then $\cA_\kappa(I_1)\subset\cA_\kappa(I_2)$. \label{1dcn:isotony}
 \item \textbf{Locality:} if $I_1$ and $I_2$ are disjoint, then $\cA_\kappa(I_1)\subset\cA_\kappa(I_2)'$. \label{1dcn:locality}
 \item\textbf{Diffeomorphism covariance:} \label{1dcn:diff}
 For $I \in \cI, \gamma \in \DiffSone$,
 \[
  U_\kappa(\gamma)\cA_\kappa(I)U_\kappa(\gamma)^*=\cA_\kappa(\gamma  I).
 \]
 Furthermore, if $\supp \gamma$ is disjoint from $I$, then $\cAd U_\kappa(\gamma)(x) = x$ for $x \in \cA_\kappa(I)$.
 \item\textbf{Positivity of energy:}
 the restriction of $U_\kappa$ to the rotation subgroup has the spectrum contained in $\bbN$.
 \label{1dcn:positiveenergy}
 \item \textbf{Vacuum and the Reeh-Schlieder property:} there exists a unique (up to a phase)
 vector $\Omega_\kappa \in\cH_\kappa$ such that $U_\kappa(g)\Omega_\kappa=\Omega_\kappa$ for $g\in \Mob$
 and $\overline{\cA_\kappa(I)\Omega_\kappa}=\cH_\kappa$. \label{1dcn:vacuum}
 \setcounter{1dCN}{\value{enumi}}
\end{enumerate}
As we assume (1dCN\ref{1dcn:positiveenergy}),
$U_\kappa$ is a projective representation of $\DiffSone$ rather than of $\uDiffSone$, see \cite[Proposition 2.1]{AGT23Pointed}).

For a conformal net $\cA_\kappa$ on $S^1$,
the following properties are automatic:
\begin{itemize}
 \item $\cA_\kappa(I)$ are type III$_1$ factors \cite[Lemma 2.9]{GF93}
 \item Haag duality $\cA_\kappa(I)' = \cA_\kappa(I')$ \cite[Theorem 2.19]{GF93}
 \item Additivity $\cA_\kappa(I) = \overline{\bigcup_{I_\alpha \Subset I} \cA_\kappa(I_\alpha)}$ \cite[Section 3]{FJ96}
\end{itemize}

Let us recall some useful notions regarding representations of conformal nets (see \cite[Section 2.2.2]{AGT23Pointed}).
A \textbf{representation} of a conformal net $(\cA_\kappa, U_\kappa, \Omega_\kappa)$
on $S^1$ is a family of representations $\rho = \{\rho_I\}$ of $\{\cA_\kappa(I)\}_{I \in \cI}$ on $\cH_\rho$
such that if $I_1 \subset I_2$, then $\rho_{I_2}|_{\cA_\kappa(I_1)} = \rho_{I_1}$ (\textbf{compatibility}).
For any such irreducible representation, there is a multiplier representation $U_\rho$ of $\uDiffSone$,
defined by $U_\rho(\gamma) = \rho(U_\kappa(\gamma))$ if $\supp \gamma \subset I_1$ for some $I_1 \in \cI$,
such that for any $I \in \cI$ and $x \in \cA_\kappa(I)$ it holds that $\rho_{\gamma I}(\cAd U(\gamma)(x)) = \cAd U_\rho(\gamma)(\rho_I(x))$
\cite[Section 5.2]{AGT25Multiplier}.
In particular, if $\cH_\rho = \cH_\kappa$ (the vacuum Hilbert space), and $\rho_I(\cA_\kappa(I)) = \cA_\kappa(I)$ for each $I$,
we call $\rho$ an \textbf{automorphism}. Note that an automorphism is irreducible.
A representation on the same Hilbert space $\cH_\rho = \cH_\kappa$ is said to be \textbf{localized in $I$} if $\rho_{I'} = \id$.
where $I'$ is the interior of $S^1 \setminus I$. In this case, $\rho_{I}$ maps $\cA_\kappa(I)$ to itself by Haag duality.
Given $\rho_1$ localized in $I_1$, one can always find a unitarily equivalent automorphism $\rho_2$ localized in another interval $I_2$,
as local algebras are of type III.
If $I_1 \cup I_2 \subset I \in \cI$, there is a unitary operator $V \in \cA_\kappa(I)$ implementing the equivalence
between $\rho_1$ and $\rho_2$ by Haag duality. Such a $V$ is called a \textbf{charge transporter}.

For an automorphism $\rho$ and $\gamma \in \uDiffSone$,
$\rho^\gamma := \cAd U_\kappa(\gamma) \circ \rho \circ \cAd U_\kappa(\gamma^{-1})$ is localized in $\gamma I$
and is equivalent to $\rho$. We can take
$z_{\kappa,\rho}(\gamma) := U_\kappa(\gamma)U_\kappa^\rho(\gamma)^*$ as a charge transporter such that $\cAd z_{\kappa, \rho}\circ \rho = \rho^\gamma$.
We call $z_{\kappa,\rho}(\gamma)$ the \textbf{covariance cocycle} of $\rho$.

Let us remove $-1$ from $S^1$ and identify $S^1 \setminus \{-1\}$ with $(-\pi, \pi)$.
Let $\rho_1, \rho_2$ be two automorphisms localized in $I$ such that $\overline I \subset (-\pi, \pi)$.
Take $\tilde\rho_1, \tilde\rho_2$ localized in $I_1, I_2 \in \cI$, respectively,
such that $I_1 \cap I_2 = \emptyset$, $I_1, I_2 \subset (-\pi, \pi)$ and $\tilde\rho_j$ is equivalent to $\rho_j, j=1,2$.
For charge transporters $V_1, V_2$ between $\rho_1$ and $\tilde\rho_1$, $\rho_2$ and $\tilde \rho_2$, respectively,
$\varepsilon^{\pm}_{\rho_1, \rho_2} := \rho_2(V_1^*) V_2^* V_1 \rho_1(V_2)$,
where $\pm$ depends on whether $I_1$ is on the future (right, $+$) or the past (left, $-$) of $I_2$,
is called the \textbf{(DHR) braiding} between $\rho_1, \rho_2$, after \cite{DHR71, FRS89, GF93}.
The braiding does not depend on the choice of $\tilde \rho_j$ or the charge transporters under these conditions.

If we take $\gamma$ such that $\gamma I$ is on the future or the past of $I$, $I\cap \gamma I = \emptyset$
and $V_1 = \bb1, V_2 = z_{\kappa,\rho_2}(\gamma)$, then the braiding is ($+$ corresponds to $\gamma I$ being in the past of $I$,
$-$ corresponds to the future)
\begin{align}\label{eq:braiding-z}
 \varepsilon^\pm_{\rho_1, \rho_2} = z_{\kappa,\rho_2}(\gamma)^* \rho_1(z_{\kappa,\rho_2}(\gamma)).
\end{align}

Let us consider a pair of conformal nets $(\cA_j, U_j, \Omega_j)$ on $S^1$ on the Hilbert spaces $\cH_j$, $j=1,2$. We can consider its tensor product
$\cA_1\otimes \cA_2(I) := \cA_1(I)\otimes \cA_2(I)$ and it is again a conformal net\footnote{This should be distinguished from
the two-dimensional conformal net with $\cA_1$ as the chiral component, $\cA_2$ as the antichiral component, see Section \ref{1d2dconformalnet}.}
on $S^1$ with respect to $U_1\otimes U_2$
and $\Omega_1\otimes \Omega_2$ on $\cH_1\otimes \cH_2$.
For two representations $\sigma_j$ of $\cA_j$, respectively $j=1,2$,
the tensor product $\sigma_1\otimes \sigma_2$ is a representation of $\cA_1\otimes \cA_2$.
Let $\rho_j, \sigma_j, j=1,2$ be localized representations of $\cA_j$, respectively.
It is easy to show, if $V_j, W_j, j=1,2$ are charge transporters of $\rho_j, \sigma_j$ as before,
then $V_j\otimes W_j$ are a charge transporter for $\rho_j\otimes \sigma_j$,
thus $\varepsilon^\pm_{\rho_1\otimes\rho_2, \sigma_1\otimes\sigma_2} = \varepsilon^\pm_{\rho_1, \sigma_1}\otimes\varepsilon^\pm_{\rho_2, \sigma_2}$.

\subsubsection{Two-dimensional tensor product nets and extensions}\label{1d2dconformalnet}
Let $(\cA_\pm, U_\pm, \Omega_\pm)$ be two conformal nets on $S^1$ on the Hilbert space $\cH_\pm$, respectively.
From them, we can construct a two-dimensional conformal net $(\cA_+\otimes\cA_-, U_+\otimes U_-, \Omega_+\otimes \Omega_-)$
on the Hilbert space $\cH_+\otimes \cH_-$ by setting (first for double cones, then extending to arbitrary bounded regions
as explained in Section \ref{2dconformalnet})
\begin{itemize}
 \item $\cA_+\otimes\cA_-(I_+\times I_-) := \cA_+(I_+) \otimes \cA_-(I_-)$
 \item $U_+\otimes U_-(\gamma_+\times \gamma_-) := U_+(\gamma_+)\otimes U_-(\gamma_-)$
 \item $\Omega_+\otimes \Omega_- \in \cH_+\otimes\cH_-$
\end{itemize}
The $\pm$-components are referred to as the \textbf{chiral} and \textbf{antichiral} components, respectively.
Such a net \textit{a priori} extends to the Einstein cylinder (see \cite{AGT23Pointed}), but by
(1dCN\ref{1dcn:positiveenergy}), it reduces to the torus.

One can also consider representations of such nets, requiring the compatibility condition as we did in Section \ref{S1conformalnet}.
We say that $\tau = \{\tau_O\}$, where $\tau_O$ is a representation of $\cA_+\otimes \cA_-(O)$,
where $O$ is a region on the Einstein cylinder, if $\tau_{\tilde O}|_{\cA(O)} = \tau_{O}$ for any $O \subset \tilde O$,
where $\tilde O$ is contained in a copy of the Minkowski space.
Such $\tau$ can be restricted to $\cA_+\otimes \bb1$, $\bb1\otimes \cA_-$
and gives representations of $\cA_+, \cA_-$ as conformal nets on $S^1$, by considering
the time-zero circle of the cylinder.

A generic two-dimensional conformal net $(\cA, U, \Omega)$ is an extension of such net in the following sense:
There are (two-dimensional) conformal nets of von Neumann algebras $\underline\cA_+$, $\underline\cA_-$
acting on the Hilbert space $\cH$ of $\cA$, which satisfy all the axioms (2dCN\ref{2dcn:isotony})--(2dCN\ref{2dcn:positiveenergy}) except that
$\Omega$ is not required to be cyclic for $\underline{\cA}_\pm$, and in addition,
if $x \in \underline\cA_+(I_+\times I_-)$, then $\cAd U(\iota\times \gamma_-)(x) = x$ for all $\gamma_- \in \uDiffSone$ and $\iota$ the identity diffeomorphism,
and similarly, if $x \in \underline\cA_-(I_+\times I_-)$, then $\cAd U(\gamma_+\times \iota)(x) = x$ for all $\gamma_+ \in \uDiffSone$.
In this case, we can define (by a slight abuse of notations) $\underline\cA_\pm(I_\pm) := \underline\cA_\pm(I_+\times I_-)$,
where $\cA_\pm(I_\pm)$ does not depend on $I_\mp$, and
$\cH_\pm := \overline{\bigcup_{I \in \cI} \cA_\pm(I)\Omega}$
(which is equal to $\overline{\cA_\pm(I)\Omega}$ for any $I \in \cI$ by the Reeh-Schlieder argument, see \cite[Lemma 5.1]{BLM11} and \cite[Theorem 1.4]{GLW98})
and define the restrictions $\cA_\pm(I) := \underline{\cA}_\pm(I)|_{\cH_\pm}$.
They are conformal nets on $S^1$, as we assume (2dCN\ref{2dcn:cylinder}).
In this case, we say that $\cA$ is an \textbf{extension} $\cA_+\otimes \cA_-$
and we denote it as $\cA_+\otimes \cA_- \subset \cA$.

If $\cA$ is an extension of $\cA_+\otimes \cA_-$, then the map
$\underline \cA_\pm(I) \ni x \mapsto x|_{\cH_\pm} \in \cA_\pm(I)$ is an injective homomorphism
(as $\Omega$ is separating). The inverse map is a representation of $\cA_\pm(I)$.
Using this inverse map, we can consider the embedding $\cA_+\otimes \cA_- \to \cA$,
and it can be seen as a representation of $\cA_+\otimes \cA_-$ on $\cH$.

\subsection{The \texorpdfstring{$\rmU(1)$}{U(1)}-current algebra}\label{sec:U(1)}
\subsubsection{The \texorpdfstring{$\rmU(1)$}{U(1)}-current algebra and its modules}\label{sec:U(1)-algebra}

We will consider the infinite-dimensional Lie algebras known as the Heisenberg algebras.
As a special case with rank $1$, the Lie algebra is also known as the $\rmU(1)$-current algebra, especially in the context of
algebraic QFT \cite{BMT88}. For the convenience of the reader familiar with the latter, we will briefly summarize the structures.
See also \cite[Section 6]{AGT23Pointed}.

The \textbf{$\rmU(1)$-current algebra} is the infinite-dimensional Lie algebra generated by the symbols
$\{J_m\}_{m \in \bbZ}$ and $K$, with the following commutation relations
\begin{align*}
 [J_m, J_n] = \delta_{m,-n}m K
\end{align*}
and $K$ is a central element.

For each pair of values $c, \lambda \in \bbC$, we can consider the so-called Verma module of the $\rmU(1)$-current algebra:
it is spanned by a distinguished vector $\Omega_\lambda$ and vectors of the form
\begin{align*}
 J_{-m_1}\cdots J_{-m_k}\Omega_\lambda, \qquad m_j \in\bbZ, m_j > 0
\end{align*}
with the properties that $K=c\bb1, J_0 = \lambda\bb1$ and
\begin{align*}
 J_m \Omega_\lambda = 0, \qquad m > 0.
\end{align*}
This representation (module) of the $\rmU(1)$-current algebra is denoted by $M_{\frh}(c,\lambda)$ 
(where $\frh$ is a 1-dimensional $\bbR$-vector space, \cf Section \ref{sec:Heisenberg-algebra}).
This can be seen as the algebraic symmetric Fock space based on $\{J_m\Omega_\lambda, m < 0\}$.
On $M_\frh(c, \lambda)$ with $c \neq 0$, $\{c^{-\frac12} J_m\}_{m \in \bbZ}$ satisfy the same commutation relations 
as the $\{J_m\}_{m \in \bbZ}$ with $c=1$, so we may and do assume that $c=1$.

If $\lambda \in \bbR$, the representation is \textbf{unitary}, in the sense that there is a (positive definite)
invariant
scalar product $\<\cdot, \cdot\>$ with respect to which $(J_m)^* = J_{-m}$
and $\langle \Omega_\lambda, \Omega_\lambda \rangle = 1$.
Note that, by the invariance property, the value of $\<\cdot, \cdot\>$ is uniquely determined by $\lambda$ and $c=1$.
Moreover, $M_{\frh}(1,\lambda)$'s are isomorphic to each other as vector spaces with a scalar product,
by identifying the vectors $J_{-m_1}\cdots J_{-m_k}\Omega_\lambda$ (the values $\lambda$ do not appear here, nor in the scalar product).
The module $M_\frh(1, 0)$ with $\lambda = 0$ is called the \textbf{vacuum module} and its completion
accommodates the chiral conformal net.

\subsubsection{The \texorpdfstring{$\rmU(1)$}{U(1)}-current net and its representations}\label{sec:U(1)-net}

Let us take the vacuum module $M_\frh(1, 0)$ of the $\rmU(1)$-current algebra
and its Hilbert space completion $\cH_{\frh, 0}$ with respect to the invariant scalar product $\<\cdot, \cdot\>$.
We introduce the operators $L_m$ on $M_\frh(1, 0)$ (or on $M_\frh(1, \lambda)$) by the Sugawara formula
\begin{align}\label{eq:sugawara}
 L_m := \frac12\sum_{k \in \bbZ} :J_{m-k}J_k:, \quad m \in \bbZ.
\end{align}
They satisfy the commutation relations $[L_m, J_n] = -nJ_{m+n}$, and
the Virasoro algebra relations $[L_m, L_n] = (m-n)L_{m+n} + \frac{c}{12}m(m^2-1) \delta_{m,-n}$ with central charge $c=1$.
For $m_j > 0$
\begin{align*}
 L_0 \cdot J_{-m_1}\cdots J_{-m_k}\Omega_0 = \left(\sum_{j=1}^k m_j\right)J_{-m_1}\cdots J_{-m_k}\Omega_0.
\end{align*}
The operator $L_0$ is already diagonalized, essentially self-adjoint and positive on $M_\frh(1, 0) \subset \cH_{\frh, 0}$.
We denote its closure by the same symbol $L_0$, and refer to it as the conformal Hamiltonian. 
Let $C^\infty(L_0) := \bigcap_{j=1}^\infty \Dom(L_0^j)$.
It holds that $\|J_m\Psi\| \le (m+1)\|(L_0 + \bb1)\Psi\|$ for any $\Psi \in C^\infty(L_0)$ \cite[(2.23)]{BS90},
which we call the \textbf{linear energy bounds} for the $\rmU(1)$-current $\{J_m\}_{m\in\bbZ}$ in the vacuum Hilbert space $\cH_{\frh, 0}$.

Let $f \in C^\infty(S^1)$. We regard $f$ as a $2\pi$-periodic smooth function on $\bbR$.
Let us write its Fourier components by $f_n = \frac1{2\pi}\int_{-\pi}^{\pi} \rme^{-\rmi nt}f(t)\rmd t$. On the domain $C^\infty(L_0)$, we define
\begin{align*}
 J(f)\Psi := \sum_{m \in \bbZ} f_m J_m \Psi,
\end{align*}
which is convergent in $\cH_{\frh, 0}$ by the linear energy bounds, and $J(f)\Psi \in C^\infty(L_0)$. It defines an essentially self-adjoint operator whose closure we denote again by $J(f)$, and $C^\infty(L_0)$ is an invariant core for $J(f)$ for every $f \in C^\infty(S^1)$, \cf \cite[Appendix A]{AGT23Pointed}. We denote also $W(f) := \rme^{\rmi J(f)}$.

For $I \in \cI$ (see Section \ref{2dconformalnet}), we define
\begin{align*}
 \cA_\frh(I) := \{W(f) : \supp f \subset I\}''.
\end{align*}
The representation of the Virasoro algebra $\{L_m\}_{m\in\bbZ}$ integrates to a positive-energy projective unitary representation $U_0$
of $\DiffSone$ such that
\begin{align*}
 \cAd U_0(\gamma)(J(f)) = J(f\circ \gamma^{-1}), \quad \cAd U_0(\gamma)(W(f)) = W(f\circ \gamma^{-1}),
\end{align*}
where the first equality is as unbounded operators (including the domains) \cite[Proposition 6.4]{CKLW18}.
The triple $(\cA_\frh, U_0, \Omega_0)$ satisfies the axioms of conformal net on $S^1$,
and we call it the \textbf{$\rmU(1)$-current net}.

Let $h \in C^\infty(S^1)$. For $I_1 \in \cI$, we define
(note that we changed convention from \cite{AGT23Pointed}, and the integral is normalized by $\frac1{2\pi}$)
\begin{align*}
 \sigma_{h, I_1}(W(f)) = \rme^{\frac{\rmi}{2\pi} \int f(t)h(t)dt}W(f).
\end{align*}
Then this is an automorphism of $\cA_\frh$.
Any two such automorphisms $\sigma_{h_1, I_1}, \sigma_{h_2, I_1}$ are unitarily equivalent
if and only if $\frac{1}{2\pi} \int_{-\pi}^\pi h_1(t)dt = \frac{1}{2\pi} \int_{-\pi}^\pi h_2(t)dt$.
We know by \cite[Section 5.2]{AGT25Multiplier}, \cite[Theorem 2.2]{Gui21Categorical} that any such automorphism (representation) is diffeomorphism covariant,
in the sense that there is a positive-energy unitary multiplier representation $U_h$ of $\uDiffSone$
such that
\begin{align*}
 \cAd U_h(\gamma)(\sigma_h(W(f))) = \sigma_h(\cAd U_0(\gamma)(W(f))) = \sigma_h(W(f \circ \gamma^{-1})).
\end{align*}
Let $I \in \cI$. If we take $h \in C^\infty(S^1)$, $\supp h \subset I$,
then $\sigma_h$ is an automorphism of $\cA_\frh$ localized in $I$.

Let us fix $h$ such that $\frac1{2\pi}\int_{-\pi}^{\pi} h(t)dt = 1$.
For $\alpha \in \bbR$, we say that the automorphism $\sigma_{\alpha h}$ has charge $\alpha$.
If $h$ is localized in $I$, then by the calculations of \cite{AGT23Pointed}\footnote{In \cite[Section 6.3]{AGT23Pointed}, we took the convention
$\int h(t)dt = 1$ and the automorphism was given by $W(f)\mapsto \rme^{\rmi\int f(t)h(t)dt} W(f)$.
In this paper we have $\frac1{2\pi}\int h(t)dt = 1$, thus $h$ gets multiplied by $2\pi$. On the other hand, in the calculations of braiding,
we had integral such as $\int h_j(t)H_{j'}(t)dt$, where $h_j = \alpha_j h$ and $H_{j'}$ is a primitive of $h_{j'}$.
We replace the integral by $\frac1{2\pi}\int h_j(t)H_{j'}(t)dt$, thus altogether the value has a factor of $2\pi$
compared with \cite{AGT23Pointed}. In \cite{AGT23Pointed} we had $\varepsilon_{\alpha h, \beta h}^\pm = \rme^{\pm\frac{\alpha\beta}2}$,
thus it translates to the formula below.}, for $\alpha, \beta \in \bbR$, the braiding is
\begin{align*}
 \varepsilon_{\alpha h, \beta h}^\pm = \rme^{\pm \rmi \pi\alpha\beta}.
\end{align*}

\subsection{Heisenberg algebras}\label{sec:Heisenberg}

\subsubsection{The Heisenberg Lie algebras and their modules}\label{sec:Heisenberg-algebra}
Let $\frh$ be a finite-dimensional $\bbR$-vector space, $(\cdot, \cdot)_\frh$ be a \textit{positive-definite} symmetric bilinear form,
that is $(\pmb \alpha, \pmb \alpha)_\frh \geq 0$, and $=0$ if and only if $\pmb \alpha = \pmb 0$, and $(\pmb \alpha, \pmb \beta)_\frh = (\pmb \beta, \pmb \alpha)_\frh$, for every $\pmb \alpha, \pmb \beta \in \frh$, and linear on both variables.
If $\frh$ is 1-dimensional with the standard scalar product, the following construction gives the $\rmU(1)$-current algebra of Section \ref{sec:U(1)-algebra}.
See also \cite[Section 3.5]{Kac98}.

We consider the Heisenberg Lie algebra $\widehat\frh$, which is spanned by elements of the form
$\pmb{\alpha}(m)$, where $\pmb{\alpha} \in \frh, m \in \bbZ$
(with the relations $(\pmb{\alpha} + \pmb{\beta})(m) = \pmb{\alpha}(m) + \pmb{\beta}(m)$ and $(\lambda \pmb{\alpha})(m) = \lambda \pmb{\alpha}(m)$)
and a central element $K$ (that is, as a linear space, $\widehat\frh \cong (\bigoplus_\bbZ \frh) \oplus \bbC K$
and $\pmb{\alpha}(m)$ has the element $\pmb{\alpha} \in \frh$ in the $m$-th summand),
with the commutation relations
\begin{align*}
 [\pmb{\alpha}(m), \pmb{\beta}(n)] &= \delta_{m,-n}m\cdot (\pmb{\alpha}, \pmb{\beta})_\frh K.
\end{align*}
This is related with the $\rmU(1)$-current algebra as follows.
\begin{itemize}
 \item If $\frh$ is $1$-dimensional, then by taking a vector $\pmb{\upsilon} \in \frh$ with $(\pmb{\upsilon}, \pmb{\upsilon})_\frh = 1$,
 we obtain $[\pmb{\upsilon}(m), \pmb{\upsilon}(n)] = \delta_{m, -n}m K$, which is exactly the commutation relation of the $\rmU(1)$-current algebra,
 by identifying $\pmb{\upsilon}(m) = J_m$.
 \item If $\frh$ is $N$-dimensional, we can take an orthonormal basis.
 For orthogonal vectors $\pmb{\alpha}, \pmb{\beta} \in \frh$, one has $[\pmb{\alpha}(m), \pmb{\beta}(n)] = 0$ also for $m=-n$.
 For each $\pmb{\alpha} \in \frh$ with $(\pmb{\alpha}, \pmb{\alpha})_\frh = 1$,
 the elements $\{\pmb{\alpha}(m)\}_{m \in \bbZ}, K$ span a subalgebra isomorphic to the one-dimensional case.
 Note that, 
 the (commutative) subalgebra generated by $\pmb{\alpha}(m), m < 0$ is isomorphic to the \textit{direct sum} of $N$ copies of the subalgebra
 generated by $J_m, m < 0$ of the $\rmU(1)$-current algebra,
 that is, if $\pmb{\alpha} = (\alpha_1, \cdots, \alpha_n) \in \frh$,
 each $\alpha_j(m)$ can be seen as a copy of $\alpha_j J_m$.
\end{itemize}
As in the 1-dimensional case, there are irreducible modules of this algebra parametrized by $\pmb{\lambda} \in \frh$ as follows.
For each $\pmb{\lambda} \in \frh$, we denote $M_{\frh}(1, \pmb{\lambda}) = U(\widehat\frh)/J_{\pmb{\lambda}}$,
where $U(\widehat\frh)$ is the universal enveloping algebra of $\widehat\frh$, $J_{\pmb{\lambda}}$ is the ideal generated by $\pmb{\alpha}(m)$ for $m > 0$, $\pmb{\alpha}(0)-(\pmb{\lambda}, \pmb\alpha)_\frh$, and by $K-1$.
More explicitly, let $\{\pmb{\upsilon}_j\}_{j=1, \cdots, N}$ be a basis of $\frh$,
If $\pmb{\lambda} = \sum_j \lambda_j \pmb{\upsilon}_j, \lambda_j \in \bbR$, then this module is isomorphic to
the tensor product $M_{\bbR \pmb{\upsilon}_1}(1, \lambda_1)\otimes \cdots\otimes M_{\bbR \pmb{\upsilon}_N}(1, \lambda_N)$
of modules in Section \ref{sec:U(1)-algebra}, where $\bbR \pmb{\upsilon}_j$ are the $1$-dimensional vector spaces spanned by $\pmb{\upsilon}_j$
(see \cite[Theorem 3.43]{DG13MathematicsOfQuantization})
and each $\pmb{v}_j(m)$ acts on the $j$-th component as $J_m$.
If $\pmb{\lambda} = \pmb 0$, the module $M_\frh(1, \pmb 0)$ is called the \textbf{vacuum module} here as well.

\subsubsection{The Heisenberg conformal net and its representations}\label{sec:Heisenberg-net}
Let $\frh$ be $N$-dimensional with $(\cdot, \cdot)_\frh$ as in Section \ref{sec:Heisenberg-algebra}.
The following construction gives the $N$-fold tensor product of the $\rmU(1)$-current net of Section \ref{sec:U(1)-net}.

We consider the Heisenberg Lie algebra $\widehat\frh$ as above.
Let $M_\frh(1, \pmb{\lambda})$ be one of the representations above.
We denote the Hilbert space completion with respect to the invariant scalar product $\<\cdot, \cdot\>$
by $\cH_{\pmb \lambda}$, or when we need to stress $\frh$, by $\cH_{\frh, \pmb{\lambda}}$.

Fix an orthonormal basis $\{\pmb{\upsilon}_j\}_{j=1, \cdots, N}$ of $\frh$. On each module $M_\frh(1, \pmb{\lambda})$ (we will be mainly interested in $\pmb{\lambda} = \pmb{0}$ for the following construction)
we introduce the representation of the Virasoro algebra given by the $N$-dimensional analogue of the Sugawara formula
\begin{align*}
 L_m := \frac12 \sum_{j=1}^N \sum_{k\in \bbZ} :\pmb{v}_j(m-k)\pmb{v}_j(k):
\end{align*}
Actually, it is easy to check that this does not depend on the choice of the orthonormal basis $\{\pmb{\upsilon}_j\}$.
The operators $\{L_m\}$ satisfy the Virasoro relations with central charge $c=N$.
Indeed, each $\pmb{v}_j(m)$ acts on the $j$-th tensor component as $J_m$.
In this way, $L_m$ can be seen as the $N$-fold tensor product
$\bigoplus_{j=1}^N \bb1\otimes\cdots \otimes\underset{j\text{-th}}{L^{\rmU(1)}_m}\otimes\cdots \otimes\bb1$
of the representation of the Virasoro algebra for the $\rmU(1)$-current algebra.

As in the $1$-dimensional case, for each $f \in C^\infty(S^1)$ with $f_m = \frac1{2\pi}\int_{-\pi}^{\pi} \rme^{-\rmi mt}f(t)dt$
and $\pmb{\alpha} \in \frh$,
we introduce the densely defined operator on $\cH_{\pmb 0}$:
\begin{align*}
 \pmb{\alpha}(f) := \sum_{m\in\bbZ} f_m\pmb{\alpha}(m),
\end{align*}
(with a slight abuse of notation ``$\pmb{\alpha}(\cdot)$''): $\pmb{\alpha} \in \frh$, $\pmb{\alpha}(m) \in \End(M_\frh(1, \pmb{0}))$ for $m\in \bbZ$,
but here $\pmb{\alpha}(\cdot)$ will define a Wightman field on $S^1$.
For $\pmb{\alpha}, \pmb{\beta} \in \frh$, $f,g \in C^\infty(S^1)$, it holds that
\begin{align*}
 [\pmb{\alpha}(f), \pmb{\beta}(g)] = (\pmb{\alpha}, \pmb{\beta})_\frh \cdot \frac1{2\pi}\int_{-\pi}^\pi f'(t)g(t)dt.
\end{align*}
In particular, $\pmb{\alpha}(f)$ and $\pmb{\beta}(g)$ commute when $\supp f$ and $\supp g$ are disjoint.

Let us consider the vacuum module $M_\frh(1, \pmb{0})$.
For each $\pmb{\alpha} \in \frh$, the components $\{\pmb{\alpha}(m)\}_{m \in \bbZ}$ satisfy the linear energy bounds with respect to $L_0$.
To see this, one can take an orthonormal basis of $\frh$ whose first vector is a scalar multiple of $\pmb{\alpha}$, then the $\pmb{\alpha}(m)$
are bounded by the first component of $L_0$,
thus by the whole operator $L_0$:
$\|\pmb{\alpha}(m) \Psi\| \le \|\pmb{\alpha}\|\cdot (m+1)\|(L_0 + \bb1)\Psi\|$ for every $\Psi \in C^\infty(L_0) = \bigcap_{j=1}^\infty \Dom(L_0^j)$.

As in the $1$-dimensional case, we obtain a positive-energy projective unitary representation $U_{\pmb{0}}$ of $\DiffSone$, which satisfies
for each $\pmb{\alpha} \in \frh$
\begin{align*}
 \cAd U_{\pmb{0}}(\gamma)(\pmb{\alpha}(f)) = \pmb{\alpha}(f\circ \gamma^{-1}),
\end{align*}
including the domain $C^\infty(L_0)$. Let us define
\begin{align*}
 \cA_\frh(I) := \{\rme^{\rmi \pmb{\alpha}(f)}: \pmb{\alpha} \in \frh, f \in C^\infty(S^1), \supp f \subset I\}.
\end{align*}
Together with the vacuum vector $\Omega_{\pmb{0}}$, the triple $(\cA_\frh, U_{\pmb{0}}, \Omega_{\pmb{0}})$
constitutes a conformal net on $S^1$. There is a natural unitary equivalence
$\cA_\frh \cong \cA_{\bbR\pmb{\upsilon}_1}\otimes \cdots \otimes \cA_{\bbR\pmb{\upsilon}_N}$, where each $\cA_{\bbR\pmb{\upsilon}_1}$
is a copy of the $\rmU(1)$-current net from Section \ref{sec:U(1)-net}.

The sectors (equivalence classes of irreducible representations) of the net can be parametrized by $\frh$.
Indeed, let us take a basis $\{\pmb{\upsilon}_j\}$ of $\frh$.
By \cite[Proposition 3.11]{LT18}, which applies to the $\rmU(1)$-current net by \cite[Example 3.9]{LT18},
and by the classification in \cite[Section 2B, 3B]{BMT88} together with \cite{CW-localenergy} (\cf \cite{Zellner17})
any sector $\sigma$ of $\cA_\frh \cong \cA_{\pmb{\upsilon}_1\bbR}\otimes \cdots \otimes \cA_{\pmb{\upsilon}_N\bbR}$
is of the form $\sigma_{\lambda_1 h}\otimes\cdots \otimes\sigma_{\lambda_N h}$,
where $h \in C^\infty(S^1)$ such that $\frac1{2\pi}\int_{-\pi}^\pi h(t)dt = 1$,
giving the parametrization $(\lambda_1, \cdots, \lambda_N) \in \bbR^N \cong \frh$,
whose Hilbert space we denote by $\cH_{\pmb{\lambda}h}$ (or $\cH_{\frh, \pmb{\lambda}h}$ when necessary).
These sectors correspond to the module $M_\frh(1, \pmb{\lambda}), \pmb{\lambda} = \sum_j \lambda_j \pmb{\upsilon}_j$
(up to a unitary equivalence).

Let us fix $h\in C^\infty(S^1)$ such that $\frac1{2\pi}\int_{-\pi}^{\pi} h(t)dt = 1$.
We know that, for $\alpha, \beta \in \bbR$, the braiding of two sectors $\sigma_{\alpha h}, \sigma_{\beta h}$ of the $\rmU(1)$-current net is
$\varepsilon_{\alpha h, \beta h}^\pm = \rme^{\pm\rmi\pi\alpha\beta}$.
If $\frh$ is $N$-dimensional, the braiding of the tensor product is the tensor product of the braidings, see Section \ref{S1conformalnet}.
Let us write in the basis $\pmb{\alpha} = \sum_j \alpha_j \pmb{\upsilon}_j, \pmb{\beta} = \sum_j \beta_j \pmb{\upsilon}_j \in \frh$,
and $\sigma_{\pmb{\alpha}h} := \bigotimes_j \sigma_{\alpha_j h}, \sigma_{\pmb{\beta}h} := \bigotimes_j \sigma_{\beta_j h}$.
Then the braiding of the two sectors $\sigma_{\pmb{\alpha}h}, \sigma_{\pmb{\beta}h}$ of $\cA_\frh$ is 
\begin{align}\label{eq:braidingNdim}
\varepsilon_{\pmb{\alpha}h, \pmb{\beta}h}^\pm = \bigotimes_{j=1}^N \varepsilon^\pm_{\alpha_j\pmb{v}_jh, \beta_j\pmb{v}_jh} 
=\rme^{\pm\rmi\pi \sum_j \alpha_j \beta_j} = \rme^{\pm\rmi\pi(\pmb{\alpha}h, \pmb{\beta})_{\frh}},
\end{align}
where we used that the braidings $\varepsilon^\pm_{\alpha_j\pmb{v}_j h, \beta_j\pmb{v}_j h}$ are scalar in this case.
Furthermore, by the concrete construction, we have $\sigma_{\pmb{\alpha}h}\sigma_{\pmb{\beta}h} = \sigma_{(\pmb{\alpha} + \pmb{\beta})h}$.
Although the Hilbert space of the representation $\sigma_{\pmb{\alpha}h}$ remains the same $\cH_{\pmb 0}$,
when we stress that the net $\cA_\frh$ is represented by $\sigma_{\pmb{\alpha}h}$, we denote the Hilbert space by
$\cH_{\pmb{\alpha}h}$.

Each representation $\sigma_{\pmb{\alpha}h}$ is diffeomorphism covariant \cite[Section 5.2]{AGT25Multiplier}:
there is a unitary multiplier representation $U_{\pmb{\alpha}h}$ of $\uDiffSone$ 
such that
\begin{align}\label{eq:covariance-chiral-sector}
 \cAd U_{\pmb{\alpha}h}(\gamma)(\sigma_{\pmb{\alpha}h}(x)) = \sigma_{\pmb{\alpha}h}(\cAd U_{\pmb{0}}(\gamma)(x)),
 \qquad x \in \cA_\frh(I_1), I_1 \in \cI, \gamma \in \uDiffSone.
\end{align}

By taking the constant function $\frac{\pmb{\alpha}}{2\pi}$ (we put $\frac1{2\pi}$ to keep the normalization),
we have $\sigma_{\pmb{\alpha} h} \cong \sigma_{\frac{\pmb{\alpha}}{2\pi}}$.
While $\sigma_{\frac{\pmb{\alpha}}{2\pi}}$ is not localized in any interval,
it can be seen as the representation corresponding to the module $M_\frh(1, \pmb{\lambda})$
\textit{without} unitary equivalence.

\section{Two-dimensional extensions}\label{sec:2d}
From now on, we fix an $N$-dimensional $\bbR$-vector space $\frh$ together with a \textit{positive-definite} symmetric bilinear form $(\cdot, \cdot)_\frh$
as in Section \ref{sec:Heisenberg-algebra}. In this Section, we construct a two-dimensional conformal net $\cA_Q$
that whose chiral components extend the Heisenberg nets, associated with a certain even lattice $Q$.

\subsection{Chiral and antichiral components}\label{sec:2d-chiral}
Let $p$ be an orthogonal projection on $\frh$ with respect to $(\cdot, \cdot)_\frh$. We consider $p\frh$ to correspond to the chiral component
and $\bar p \frh$ to correspond to the antichiral component,
where $\bar p = 1-p$ is the orthogonal complement of $p$. Let us denote the pair $(\frh, p)$ as $\frh_p$
(see Section \ref{sec:vertex} for the current algebra $\widehat{\frh_p}$ with the distinction between the chiral and antichiral parts).
As each $p\frh, \bar p\frh$ are finite-dimensional $\bbR$-vector spaces,
we have the associated Heisenberg conformal nets $\cA_{p\frh}, \cA_{\bar p\frh}$ on $S^1$
on $\cH_{p\frh, \pmb 0}, \cH_{\bar p\frh, \pmb 0}$, respectively.
Then, on the Hilbert space $\cH_{p\frh, \pmb 0}\otimes \cH_{\bar p\frh, \pmb 0} \cong \cH_{\frh, \pmb 0}$, which we also denote by $\cH_{\frh_p, \pmb 0}$,
we have the two-dimensional conformal net $\cA_{p\frh}\otimes \cA_{\bar p\frh}$ as in Section \ref{1d2dconformalnet},
which we shall also denote by $\cA_{\frh_p}$.

We construct irreducible representations of the two-dimensional net $\cA_{p\frh}\otimes \cA_{\bar p\frh}$
in the sense of Section \ref{1d2dconformalnet}:
Let us fix $I \in \cI$ and $h \in C^\infty(S^1)$ with $\frac1{2\pi}\int_{-\pi}^{\pi} h(t)dt = 1$
and $\supp h \subset I$.
Then there is a one-to-one correspondence between equivalence classes of
irreducible representations of $\cA_{p\frh}\otimes \cA_{\bar p\frh}$ and $\pmb{\lambda} \in \frh$
(see Lemma \ref{pr:tensorproduct2dsector} for classification),
where the representation of the algebra $\cA_{p\frh}(I_+)\otimes \cA_{\bar p\frh}(I_-)$
is $\sigma_{p\pmb{\lambda}h, I_+}\otimes \sigma_{\bar p\pmb{\lambda}h, I_-}$.
Again, the Hilbert spaces of these representations remain the same $\cH_{\pmb{0}} (=\cH_{\frh_p, \pmb 0})$, but in order to
indicate the representation, we denote it by
$\cH_{\frh_p, \pmb{\lambda}h} \cong \cH_{p\pmb{\lambda}h}\otimes \cH_{\bar p\pmb{\lambda}h}$.

\subsection{Twisted lattice algebra}\label{sec:2d-twist}
For $\pmb{\alpha}, \pmb{\beta} \in \frh$, define $(\pmb{\alpha}|\pmb{\beta}) := (p\pmb{\alpha}, p\pmb{\beta})_{\frh} - (\bar p\pmb{\alpha}, \bar p \pmb{\beta})_{\frh}$.
Let $Q$ be an (not necessarily $N$-dimensional\footnote{In the usual convention, a lattice in an $N$-dimensional vector space $\frh$
is a $\bbZ$-span of a basis of $\frh$. Here we allow $Q$ to be generated by linearly independent elements,
not necessarily a basis. To stress this difference, we write ``not necessarily $N$-dimensional''.})
\emph{even} lattice in $\frh$ with respect to $(\cdot|\cdot)$, that is, an additive subgroup $Q$ of $(\frh,+)$ spanned by
a set of linearly independent elements of $\frh$, such that $(\pmb{\alpha}|\pmb{\beta}) \in \bbZ$ for any pair $\pmb{\alpha}, \pmb{\beta} \in Q$, and $(\pmb{\alpha}|\pmb{\alpha}) \in 2\bbZ$ for any $\pmb{\alpha} \in Q$.

As in \cite[Section 5.5]{Kac98}, 
for an ordered basis $\pmb{\upsilon}_1, \cdots, \pmb{\upsilon}_{N_0}$ of $Q$, $N_0 \le N$, let
\begin{align}\label{eq:twococycle}
 \epsilon(\pmb{\upsilon}_i, \pmb{\upsilon}_j)
 &:= \begin{cases}
     (-1)^{(\pmb{\upsilon}_i | \pmb{\upsilon}_j) + (\pmb{\upsilon}_i | \pmb{\upsilon}_i)(\pmb{\upsilon}_j | \pmb{\upsilon}_j)} = (-1)^{(\pmb{\upsilon}_i | \pmb{\upsilon}_j)} 
     & \text{ if }i < j \\
     (-1)^{((\pmb{\upsilon}_i|\pmb{\upsilon}_i)+(\pmb{\upsilon}_i|\pmb{\upsilon}_i)^2)/2} = (-1)^{(\pmb{\upsilon}_i|\pmb{\upsilon}_i)/2}& \text{ if }i = j \\
     1 & \text{ if }i > j
    \end{cases},
\end{align}
where the equalities hold as a special case of \cite[Equations (5.5.11), (5.5.12)]{Kac98} because the lattice $Q$ is even,
and extend $\epsilon$ to $Q\times Q$ biadditively
(that is, $\epsilon(\pmb{\alpha} + \pmb{\beta}, \pmb{\gamma}) = \epsilon(\pmb{\alpha}, \pmb{\gamma})\epsilon(\pmb{\beta}, \pmb{\gamma})$
and $\epsilon(\pmb{\alpha}, \pmb{\beta} + \pmb{\gamma}) = \epsilon(\pmb{\alpha}, \pmb{\beta})\epsilon(\pmb{\alpha}, \pmb{\gamma})$).
Then $\epsilon$ is a $\{\pm 1\}$-valued $2$-cocycle of $Q$, see \cite[Remark 5.5a]{Kac98}, \ie $\epsilon(\pmb{\beta} + \pmb{\gamma},\pmb{\alpha}) \epsilon(\pmb{\beta}, \pmb{\gamma}) = \epsilon(\pmb{\beta}, \pmb{\gamma} + \pmb{\alpha}) \epsilon(\pmb{\gamma},\pmb{\alpha})$, $\epsilon(\pmb{\alpha}, \pmb{0}) = \epsilon(\pmb{0}, \pmb{\alpha}) = 1$, $\epsilon(\pmb{\alpha}, \pmb{\beta}) = (-1)^{(\pmb{\alpha}|\pmb{\beta})} \epsilon(\pmb{\beta}, \pmb{\alpha})$ for every $\pmb{\alpha}, \pmb{\beta}, \pmb{\gamma} \in Q$, and it is such that $\epsilon(\pmb{\alpha}, \pmb{\alpha}) = (-1)^{(\pmb{\alpha}|\pmb{\alpha})/2}$ for every $\pmb{\alpha}\in Q$.

We define the twisted group algebra $\bbC_\epsilon[Q]$ (twisted by the 2-cocycle $\epsilon$), as the $\bbC$-algebra with basis
$\{e_{\pmb{\alpha}}\}_{\pmb{\alpha} \in Q}$, with the product given by
\begin{align}\label{eq:shift-def}
 e_{\pmb{\alpha}} e_{\pmb{\beta}} = \epsilon(\pmb{\alpha}, \pmb{\beta})e_{\pmb{\alpha} + \pmb{\beta}}.
\end{align}
For vectors $\pmb{\alpha}, \pmb{\beta} \in Q$, if $\pmb{\alpha} \neq \pmb{\beta}$, again because the lattice is even,
\begin{align}\label{eq:shift-commut}
 e_{\pmb{\alpha}} e_{\pmb{\beta}} = \epsilon(\pmb{\alpha}, \pmb{\beta}) e_{\pmb{\alpha} + \pmb{\beta}} = (-1)^{(\pmb{\alpha}|\pmb{\beta})} \epsilon(\pmb{\beta}, \pmb{\alpha}) e_{\pmb{\beta} + \pmb{\alpha}} = (-1)^{(\pmb{\alpha}|\pmb{\beta})}e_{\pmb{\beta}} e_{\pmb{\alpha}}
\end{align}
and if $\pmb{\alpha} = \pmb{\beta}$, we trivially have $e_{\pmb{\alpha}} e_{\pmb{\alpha}} = (-1)^{(\pmb{\alpha}|\pmb{\alpha})}e_{\pmb{\alpha}} e_{\pmb{\alpha}}$, as $(\pmb{\alpha}|\pmb{\alpha}) \in 2\bbZ$.

We define a (sesquilinear, linear in the second component) scalar product on $\bbC_\epsilon[Q]$ by
$\<e_{\pmb{\alpha}}, e_{\pmb{\beta}}\> := \delta_{\pmb{\alpha}, \pmb{\beta}}$.
With this scalar product, the left multiplication operators by $e_{\pmb{\alpha}}$ on $\bbC_\epsilon[Q]$ are unitary:
\begin{align*}
 \<e_{\pmb{\alpha}}e_{\pmb{\gamma}}, e_{\pmb{\alpha}}e_{\pmb{\beta}}\>
 &= \epsilon(\pmb{\alpha}, \pmb{\beta})\overline{\epsilon(\pmb{\alpha}, \pmb{\gamma})} \<e_{\pmb{\alpha} + \pmb{\gamma}}, e_{\pmb{\alpha} + \pmb{\beta}}\> \\
 &= \epsilon(\pmb{\alpha}, \pmb{\beta})\epsilon(\pmb{\alpha}, \pmb{\gamma}) \delta_{\pmb{\alpha} + \pmb{\gamma}, \pmb{\alpha} + \pmb{\beta}} \\
 &= \epsilon(\pmb{\alpha}, \pmb{\beta})\epsilon(\pmb{\alpha}, \pmb{\gamma}) \delta_{\pmb{\gamma}, \pmb{\beta}} \\
 &= \epsilon(\pmb{\alpha}, \pmb{\beta})^2 \delta_{\pmb{\gamma}, \pmb{\beta}} \\
 &= \<e_{\pmb{\gamma}}, e_{\pmb{\beta}}\>,
\end{align*}
because $\epsilon(\pmb{\alpha}, \pmb{\beta}), \epsilon(\pmb{\alpha}, \pmb{\gamma}) \in \{\pm 1\}$ for any $\pmb{\alpha}, \pmb{\beta}, \pmb{\gamma} \in Q$.

\begin{example}\label{ex:RR2example}
 Let $\frh = \bbR^2 = \bbR \oplus \bbR$ with the Euclidean scalar product $(\cdot,\cdot)_\frh$ .
 With respect to Section \ref{sec:2d-chiral}, we can take $p$ as the orthogonal projection to the first component,
 and $\bar p$ as the orthogonal projection to the second component, and define $(\cdot|\cdot)$ accordingly.
 Fix arbitrarily $R \in \bbR, R \neq 0$. We take the lattice $Q \subset \frh$ generated by
 $\frac{1}{\sqrt 2}(R \oplus R)$ and $\frac{1}{\sqrt 2}(R^{-1} \oplus (-R^{-1}))$. This is indeed an even lattice:
 for any $a, b, a_j, b_j \in \bbZ$, $j=1,2$, we have
 \begin{align}\label{eq:semidefscalarprodinQ2d-1}
  &\left(\frac{a}{\sqrt 2}(R \oplus R) + \frac{b}{\sqrt 2}(R^{-1} \oplus (-R^{-1}))
  \middle|\frac{a}{\sqrt 2}(R \oplus R) + \frac{b}{\sqrt 2}(R^{-1} \oplus (-R^{-1}))\right) = 2ab \in 2\bbZ, \\ \label{eq:semidefscalarprodinQ2d-2}
  &\left(\frac{a_1}{\sqrt 2}(R \oplus R) + \frac{b_1}{\sqrt 2}(R^{-1} \oplus (-R^{-1}))
  \middle|\frac{a_2}{\sqrt 2}(R \oplus R) + \frac{b_2}{\sqrt 2}(R^{-1} \oplus (-R^{-1}))\right) = a_1 b_2 + a_2 b_1 \in \bbZ. 
 \end{align}
 Note that, if $(x\oplus y)\in \frh$ satisfies
 $\left(x\oplus y\middle|\frac1{\sqrt 2}(R\oplus R)\right), \left(x \oplus y\middle|\frac1{\sqrt 2}(R^{-1}\oplus (-R^{-1}))\right) \in \bbZ$,
 then $(x\oplus y) \in \bbZ \frac1{\sqrt 2}(R\oplus R) + \bbZ \frac1{\sqrt 2}(R^{-1}\oplus (-R^{-1})) = Q$.
 Indeed, let us put
 \begin{align*}
  \left(x\oplus y\middle|\frac1{\sqrt 2}(R\oplus R)\right) =: m, \left(x \oplus y\middle|\frac1{\sqrt 2}(R^{-1}\oplus (-R^{-1}))\right) =: n,
 \end{align*}
 then $x - y = \sqrt 2 R^{-1} m, x + y = \sqrt 2 R n$,
 thus $x = \frac1{\sqrt 2}mR^{-1} + \frac1{\sqrt 2}nR, y = -\frac1{\sqrt 2}mR^{-1} + \frac1{\sqrt 2}nR$.
 This shows that $Q$ is maximal, \ie there is no integral lattice that includes $Q$ as a proper sublattice.
 
\end{example}

\subsection{Construction of two-dimensional conformal nets}\label{2dlatticenet}
\subsubsection{The Hilbert space and basic operators}
Let $\frh$ and $p$ be as in Section \ref{sec:2d-chiral}.
For $\pmb{\alpha} \in \frh$, we write $\pmb{\alpha} = p\pmb{\alpha} \oplus \bar p\pmb{\alpha}$ the orthogonal decomposition.
Let $Q \subset \frh$ be an (not necessarily $N$-dimensional) even lattice with respect to
$(\pmb{\alpha}|\pmb{\beta}) = (p\pmb{\alpha}, p\pmb{\beta})_\frh - (\bar p\pmb{\alpha}, \bar p\pmb{\beta})_\frh$ for $\pmb{\alpha},\pmb{\beta} \in Q$ as before.
Let us fix the 2-cocycle $\epsilon$ of $Q$ as in Section \ref{sec:2d-twist}, $I \in \cI$
and $h \in C^\infty(S^1)$ such that $\supp h \subset I$ and $\frac1{2\pi}\int_{-\pi}^{\pi} h(t)dt = 1$ as in Section \ref{sec:Heisenberg-net}.

We consider the Hilbert space
\begin{align}\label{eq:H_Q}
 \cH_Q &:= \bigoplus_{\pmb{\lambda} \in Q} \cH_{\frh_p, \pmb{\lambda}h}
 \cong \bigoplus_{\pmb{\lambda} \in Q} \cH_{p\pmb{\lambda}h}\otimes \cH_{\bar p\pmb{\lambda}h}.
\end{align}
Using the unitary equivalence (only as Hilbert spaces) $\cH_{\frh_p, \pmb{\lambda}h} \cong \cH_{\frh_p, \pmb{0}}$,
we also have the following unitary equivalence
\begin{align*}
 \cH_{Q} \cong \cH_{\frh_p, \pmb{0}} \otimes \overline{\bbC_\epsilon[Q]},
\end{align*}
where $\bbC_\epsilon[Q]$ is endowed with the scalar product $\<e_{\pmb{\alpha}}, e_{\pmb{\beta}}\> = \delta_{\pmb{\alpha}, \pmb{\beta}}$ as in the previous section and completed to get a Hilbert space.
The unitary left multiplication operators $e_{\pmb{\alpha}}$ act only on the second component, while
the operators of the net $\cA_{p\frh}\otimes \cA_{\bar p\frh}$ act on
$\cH_{\frh_p, \pmb{0}}\otimes e_{\pmb{\lambda}} \cong \cH_{p\pmb{\lambda}h}\otimes \cH_{\bar p\pmb{\lambda}h}$
by $\sigma_{p\pmb{\lambda}h}\otimes \sigma_{\bar p\pmb{\lambda}h}$ (see Section \ref{1d2dconformalnet}).
Denote this representation of $\cA_{p\frh}\otimes \cA_{\bar p\frh}$ on $\cH_Q$ by $\tau_Q$.
By \eqref{eq:covariance-chiral-sector}, $\tau_Q$ is diffeomorphism covariant with the representation $U_Q$
of $\uDiffSone\times \uDiffSone$ defined as follows
\begin{align*}
 U_Q(\gamma_+\times \gamma_-) := \bigoplus_{\pmb{\lambda} \in Q} U_{p\pmb{\lambda}h}(\gamma_+)\otimes U_{\bar p\pmb{\lambda}h}(\gamma_-).
\end{align*}

We denote the twisted shift operator on $\cH_Q$ corresponding to $\pmb\alpha \in Q$,
induced by the left multiplication by $e_{\pmb{\alpha}}$ on $\bbC_\epsilon[Q]$, as $\psi^{\pmb{\alpha}}$.
Let $\Psi = (\Psi_{\pmb{\lambda}})_{\pmb{\lambda} \in Q} \in \cH_Q$
(where for each $\pmb{\lambda}, \Psi_{\pmb{\lambda}} \in \cH_{\frh_p, \pmb{0}}\otimes e_{\pmb{\lambda}}
\cong \cH_{p\pmb{\lambda}h}\otimes \cH_{\bar p\pmb{\lambda}h}$ but we parametrize the components by $\pmb{\lambda} \in Q$,
instead of the representation $\sigma_{p\pmb{\lambda}h}\otimes \sigma_{\bar p\pmb{\lambda}h}$).
Then, definition \eqref{eq:shift-def} translates into
\begin{align}\label{eq:shift-field}
 (\psi^{\pmb{\alpha}} \Psi)_{\pmb{\alpha} + \pmb{\lambda}} = \epsilon(\pmb{\alpha}, \pmb{\lambda})\Psi_{\pmb{\lambda}},
\end{align}
or equivalently, $(\psi^{\pmb{\alpha}} \Psi)_{\pmb{\lambda}} = \epsilon(\pmb{\alpha}, \pmb{\lambda} - \pmb{\alpha})\Psi_{\pmb{\lambda} - \pmb{\alpha}}$.
The commutation relations \eqref{eq:shift-commut} translate into
\begin{align}\label{eq:field-comm}
 \psi^{\pmb{\alpha}}\psi^{\pmb{\beta}} = (-1)^{(\pmb{\alpha}|\pmb{\beta})}\psi^{\pmb{\beta}}\psi^{\pmb{\alpha}},
\end{align}
where recall that $(\pmb{\alpha}|\pmb{\beta}) \in \bbZ$ and $(\pmb{\alpha}|\pmb{\alpha}) \in 2\bbZ$ for any $\pmb\alpha, \pmb\beta \in Q$.

For a later use (Section \ref{classification}), we also introduce simple shift operators without $\epsilon$, that can be defined for any additive subgroup $Q$:
\begin{align}\label{eq:simple-shift}
 (\underline \psi^{\pmb{\alpha}} \Psi)_{\pmb{\alpha} + \pmb{\lambda}} = \Psi_{\pmb{\lambda}}.
\end{align}

Moreover, we have
\begin{align}\label{eq:field-cov}
 \cAd U_Q(\gamma_+\times \gamma_-)(\psi^{\pmb{\alpha}}) =
 \tau_Q\left(z_{p\pmb{\alpha}h}(\gamma_+)\otimes z_{\bar p\pmb{\alpha}h}(\gamma_-)\right) \psi^{\pmb{\alpha}} 
\end{align}
by calculations analogous to \cite[(2.2), Theorem 3.1]{AGT23Pointed}: differently from \cite{AGT23Pointed}, $\psi^{\pmb{\alpha}}$
is not a simple shift operator, but it is followed only by $\epsilon(\pmb{\alpha}, \pmb{\lambda})$.
As it is a scalar and $U_Q(\gamma_+\times \gamma_-)$ preserves the spaces $\cH_{\frh_p, \pmb{\lambda}h}$,
the results hold here as well.

Note that $\epsilon(\pmb{\alpha}, \pmb{\lambda}) = \epsilon(-\pmb{\alpha}, \pmb{\lambda})$
by bimultiplicatvity and $\epsilon(\pmb{\alpha}, \pmb{\alpha}) \in \{1, -1\}$,
\begin{align*}
 \epsilon(-\pmb{\alpha}, \pmb{\lambda} + \pmb{\alpha})
 &= \epsilon(-\pmb{\alpha}, \pmb{\lambda})\epsilon(-\pmb{\alpha}, \pmb{\alpha})
 = \epsilon(\pmb{\alpha}, \pmb{\lambda})\epsilon(\pmb{\alpha}, \pmb{\alpha}).
\end{align*}
Using this, for $\Psi, \Phi \in \cH_Q$, we have
\begin{align}
 \<\Psi, \psi^{\pmb{\alpha}}\Phi\> &= \sum_{\pmb{\lambda} \in Q} \<\Psi_{\pmb{\lambda}}, (\psi^{\pmb{\alpha}}\Phi)_{\pmb{\lambda}}\> \nonumber \\
 &= \sum_{\pmb{\lambda} \in Q} \epsilon(\pmb{\alpha}, \pmb{\lambda}-\pmb{\alpha})\<\Psi_{\pmb{\lambda}}, \Phi_{\pmb{\lambda} - \pmb{\alpha}}\> \nonumber \\
 &= \sum_{\pmb{\lambda} \in Q} \epsilon(\pmb{\alpha}, \pmb{\lambda})\<\Psi_{\pmb{\lambda} + \pmb{\alpha}}, \Phi_{\pmb{\lambda}}\> \nonumber \\
 &= \epsilon(\pmb{\alpha}, \pmb{\alpha})\sum_{\pmb{\lambda} \in Q} \epsilon(-\pmb{\alpha}, \pmb{\lambda}+\pmb{\alpha})\<\Psi_{\pmb{\lambda} + \pmb{\alpha}}, \Phi_{\pmb{\lambda}}\> \nonumber \\
 &= \epsilon(\pmb{\alpha}, \pmb{\alpha})\<\psi^{-\pmb{\alpha}}\Psi, \Phi\>, \label{eq:fieldstar}
\end{align}
which shows that $(\psi^{\pmb{\alpha}})^* = \epsilon(\pmb{\alpha}, \pmb{\alpha})\psi^{-\pmb{\alpha}}$.

\subsubsection{The two-dimensional conformal net}
In the notation of the previous section, we construct a conformal net $(\cA_Q, U_Q, \Omega_Q)$ on $\bbR^{1+1}$ on the Hilbert space $\cH_Q$,
defined above, as follows.
\begin{itemize}
 \item The net $\cA_Q$ is given as follows: For the diamond $I \times I$, we set
 \begin{align*}
  \cA_Q(I\times I) := \tau_Q(\cA_{p\frh}(I)\otimes \cA_{\bar p\frh}(I)) \vee \{\psi^{\pmb{\alpha}}, (\psi^{\pmb{\alpha}})^*\}_{\pmb{\alpha} \in Q}.
 \end{align*}
 For any other diamond $\gamma_+ I \times \gamma_- I$,
 where $\gamma_+ \times \gamma_- \in \uDiffSone\times\uDiffSone$
 such that $\gamma_+ I \times \gamma_- I \subset \bbR^{1+1}$,
 we set
 \begin{align*}
  \cA_Q(\gamma_+ I \times \gamma_- I) := \cAd U_Q(\gamma_+ \times \gamma_-)(\cA_Q(I\times I)).
 \end{align*}
 We show below that this definition does not depend on the choice of $\gamma_+, \gamma_-$ fulfilling the above constraints.
 \item The $\uDiffSone\times\uDiffSone$-covariance is implemented by $U_Q$ as in the previous section.
 \item The vacuum vector is $\Omega_Q := \Omega_{\pmb 0} \otimes \Omega_{\pmb 0} \in \cH_{\pmb 0}\otimes \cH_{\pmb 0} \subset \cH_Q$.
\end{itemize}

\begin{theorem}\label{th:2dnetQ}
 The triple $(\cA_Q, U_Q, \Omega_Q)$ is a conformal net on $\bbR^{1+1}$.
\end{theorem}
\begin{proof}
 Let us show that the net $\cA_Q$ is well defined.
 We first take $\gamma_\pm$ such that $\gamma_+ I = I, \gamma_- I = I$.
 In this case, it is clear that
 $\cAd U_Q(\gamma_+ I \times \gamma_- I)(\tau_Q(\cA_\frh(I)\otimes \cA_\frh(I))) = \tau_Q(\cA_\frh(I)\otimes \cA_\frh(I))$.
 By \eqref{eq:field-cov}, we see that $\cAd U_Q(\gamma_+\times \gamma_-)(\psi^{\pmb{\alpha}})$
 is contained in $\cA_Q(I\times I)$, because in this case
 $z_{p\pmb{\alpha}h} \in \cA_{p\frh}(I), z_{\bar p\pmb{\alpha}h} \in \cA_{\bar p\frh}(I)$. Similar for $(\psi^{\pmb{\alpha}})^*$.
 Therefore, $\cAd U_Q(\gamma_+\times \gamma_-)(\cA_Q(I\times I)) \subset \cA_Q(I\times I)$.
 By applying the same argument to $\gamma_+^{-1}\times \gamma_-^{-1}$,
 we see that $\cAd U_Q(\gamma_+\times \gamma_-)(\cA_Q(I\times I)) = \cA_Q(I\times I)$.

 Next, let us show that $\cA_Q(\gamma_+ I \times \gamma_- I)$ is well-defined.
 For this, let $\gamma_{+,j}, \gamma_{-,j} \in \uDiffSone, j=1,2$ and that
 $\gamma_{+,1} I = \gamma_{+, 2} I$ and $\gamma_{-, 1} I = \gamma_{-, 2} I$.
 Then $\gamma_{+, 2}^{-1}\gamma_{+,1} I = I$ and $\gamma_{-, 2}^{-1}\gamma_{-, 1} I =  I$,
 and by the previous paragraph,
 \begin{align*}
  \cAd U_Q(\gamma_{+, 1} \times \gamma_{-, 1})(\cA_Q(I \times I))
  &= \cAd U_Q(\gamma_{+, 2} \times \gamma_{-, 2})\circ
  \cAd U_Q(\gamma_{+, 2}^{-1}\gamma_{+, 1} \times \gamma_{-, 2}^{-1}\gamma_{-, 1})(\cA_Q(I \times I)) \\
  &= \cAd U_Q(\gamma_{+, 2} \times \gamma_{-, 2})(\cA_Q(I \times I)),
 \end{align*}
 that is, $\cA_Q(\gamma_+ I \times \gamma_- I)$ does not depend on
 the choice of $\gamma_+, \gamma_-$.

 Let us check the axioms (2dCN\ref{2dcn:isotony})--(2dCN\ref{2dcn:cylinder}).
 \begin{itemize}
  \item (2dCN\ref{2dcn:diff}) follows by definition.
  \item To check (2dCN\ref{2dcn:isotony}), by covariance we may assume $I \subset I_+, I_-$.
  There are $\gamma_+, \gamma_-$ such that $\gamma_+ I = I_+, \gamma_- I = I_-$.
  Then, as before,
  $\cAd U_Q(\gamma_+ \times \gamma_-)(\tau_Q(\cA_{p\frh}(I)\otimes\cA_{\bar p\frh}(I))) = \tau_Q(\cA_{p\frh}(I_+)\otimes\cA_{\bar p\frh}(I_-))$
  and $\cAd U_Q(\gamma_+\times \gamma_-)(\psi^{\pmb{\alpha}}) =
  \tau_Q\left(z_{p\pmb{\alpha}}(\gamma_+)\otimes z_{\bar p\pmb{\alpha}}(\gamma_-)\right) \psi^{\pmb{\alpha}}$.
  As we have $\tau_Q\left(z_{p\pmb{\alpha}}(\gamma_+)\otimes z_{\bar p\pmb{\alpha}}(\gamma_-)\right) \in
  \tau_Q(\cA_{p\frh}(I_+)\otimes\cA_{\bar p\frh}(I_-)) \subset \cA_Q(I_+ \times I_-)$,
  it follows that $\psi^{\pmb{\alpha}} \in \cA_Q(I_+ \times I_-)$.
  To show that $(\psi^{\pmb{\alpha}})^* \in \cA_Q(I_+ \times I_-)$, it is enough to use \eqref{eq:fieldstar}.
  \item To check (2dCN\ref{2dcn:locality}), again by covariance,
  we may assume that there are
  $\gamma_+, \gamma_- \in \uDiffSone$ such that $\gamma_+ I$ is in the past of $I$,
  while $\gamma_-(I)$ is in the future of $I$ (see Section \ref{2dconformalnet}, the other case is similar).
  We prove that $\cA_Q(I\times I)$ and $\cA_Q(\gamma_+ I \times \gamma_- I)$
  commute, where
  \begin{align*}
   \cA_Q(I \times I) &= \tau_Q(\cA_{p\frh}(I)\otimes\cA_{\bar p\frh}(I))
   \vee \{\psi^{\pmb{\alpha}}, (\psi^{\pmb{\alpha}})^*\}, \\
   \cA_Q(\gamma_+ I \times \gamma_- I) &= \tau_Q(\cA_{p\frh}(I_+)\otimes\cA_{\bar p\frh}(I_-))
   \vee \{\cAd U_Q(\gamma_+\times \gamma_-)(\psi^{\pmb{\alpha}}), \cAd U_Q(\gamma_+\times \gamma_-)(\psi^{\pmb{\alpha}})^*\}.
  \end{align*}
  \begin{itemize}
  \item It is clear that
   $\tau_Q(\cA_{p\frh}(I)\otimes\cA_{\bar p\frh}(I))$ and $\tau_Q(\cA_{p\frh}(I_+)\otimes\cA_{\bar p\frh}(I_-))$ commute.
  \item As $\tau_Q$ is localized in $I \times I$, $\psi^{\pmb{\alpha}}$ commutes with
   $\tau_Q(\cA_{p\frh}(I_+)\otimes\cA_{\bar p\frh}(I_-))$.
   Similar for $(\psi^{\pmb{\alpha}})^*$ by using \eqref{eq:fieldstar}.
  \item Showing that
   $\tau_Q(\cA_{p\frh}(I)\otimes\cA_{\bar p\frh}(I))$ and $\cAd U_Q(\gamma_+\times \gamma_-)(\psi^{\pmb{\alpha}})$ commute
   is equivalent to showing that
   $\cAd U_Q(\gamma_+^{-1}\times \gamma_-^{-1})(\tau_Q(\cA_{p\frh}(I)\otimes\cA_{\bar p\frh}(I)))
   = \tau_Q(\cA_{p\frh}(\gamma_+^{-1} I)\otimes\cA_{\bar p\frh}(\gamma_-^{-1} I)))$ and $\psi^{\pmb{\alpha}}$ commute,
   and this holds because again $\tau_Q$ is localized in $I \times I$ and $\cAd \psi^{\pmb{\alpha}}$ implements
   $\sigma_{\pmb{\alpha}h} = \sigma_{p\pmb{\alpha}h}\otimes \sigma_{\bar p\pmb{\alpha}h}$, while
   $\gamma_+^{-1} I \times \gamma_-^{-1} I$ is spacelike from $I \times I$ and
   thus $\sigma_{\pmb{\alpha}h}$ is trivial on $\cA_\frh(\gamma_+^{-1} I)\otimes\cA_\frh(\gamma_-^{-1} I)$.
   Similar for $\cAd U_Q(\gamma_+\times \gamma_-)(\psi^{\pmb{\alpha}})^*$.
  \item Let us show that $\psi^{\pmb{\alpha}}$ and $\cAd U_Q(\gamma_+\times \gamma_-)(\psi^{\pmb{\beta}})$
   commute (and their adjoints).
   By using $\psi^{\pmb{\alpha}} \psi^{\pmb{\beta}} = (-1)^{(\pmb{\alpha}|\pmb{\beta})}\psi^{\pmb{\beta}}\psi^{\pmb{\alpha}}$
   \eqref{eq:field-comm}, compute
   \begin{align*}
    \psi^{\pmb{\alpha}} \cAd U_Q(\gamma_+\times \gamma_-)(\psi^{\pmb{\beta}})
    &= \psi^{\pmb{\alpha}} \tau_Q\left(z_{p\pmb{\beta}h}(\gamma_+)\otimes z_{\bar p\pmb{\beta}h}(\gamma_-)\right) \psi^{\pmb{\beta}} \\
    &= \tau_Q\left(\sigma_{p\pmb{\alpha}h}((z_{p\pmb{\beta}h}(\gamma_+))\otimes \sigma_{\bar p\pmb{\alpha}h}(z_{\bar p\pmb{\beta}h}(\gamma_-))\right) \psi^{\pmb{\alpha}}\psi^{\pmb{\beta}} \\
    &= \tau_Q\left(\sigma_{p\pmb{\alpha}h}((z_{p\pmb{\beta}h}(\gamma_+))\otimes \sigma_{\bar p\pmb{\alpha}h}(z_{\bar p\pmb{\beta}h}(\gamma_-))\right)
    (-1)^{(\pmb{\alpha}|\pmb{\beta})}\psi^{\pmb{\beta}}\psi^{\pmb{\alpha}} \\
    \cAd U_Q(\gamma_+\times \gamma_-)(\psi^{\pmb{\beta}}) \psi^{\pmb{\alpha}}
    &= \tau_Q\left(z_{p\pmb{\beta}h}(\gamma_+)\otimes z_{\bar p\pmb{\beta}h}(\gamma_-)\right) \psi^{\pmb{\beta}} \psi^{\pmb{\alpha}}.
    \end{align*}
    To conclude that $\psi^{\pmb{\alpha}} \cAd U_Q(\gamma_+\times \gamma_-)(\psi^{\pmb{\beta}}) = \cAd U_Q(\gamma_+\times \gamma_-)(\psi^{\pmb{\beta}}) \psi^{\pmb{\alpha}}$, it is enough to use \eqref{eq:braiding-z} and to observe that
    \begin{align*}
   (-1)^{(\pmb{\alpha}|\pmb{\beta})} &= 
    \rme^{\rmi\pi(\pmb{\alpha}|\pmb{\beta})} \\
    &= \rme^{\rmi\pi((p\pmb{\alpha},p\pmb{\beta})_\frh - (\bar p\pmb{\alpha},\bar p\pmb{\beta})_\frh) }\\
    &= \varepsilon^+_{p\pmb{\alpha}h, p\pmb{\beta}h} \varepsilon^-_{\bar p\pmb{\alpha}h, \bar p\pmb{\beta}h} \\
    &= \tau_Q\left(z_{p\pmb{\beta}h}(\gamma_+)\otimes z_{\bar p\pmb{\beta}h}(\gamma_-)\right)^* 
    \tau_Q\left(\sigma_{p\pmb{\alpha}h}((z_{p\pmb{\beta}h}(\gamma_+))\otimes \sigma_{\bar p\pmb{\alpha}h}(z_{\bar p\pmb{\beta}h}(\gamma_-))\right),
   \end{align*}
   using the definition of braiding between sectors in the last equality.
   As $Q$ is even, in particular it is an integral lattice,
   we have $((-1)^{(\pmb{\alpha}|\pmb{\beta})})^2 = 1$ for any pair $\pmb{\alpha}, \pmb{\beta}$.
   Therefore,
   \begin{align*}
   \psi^{\pmb{\alpha}} \cAd U_Q(\gamma_+\times \gamma_-)(\psi^{\pmb{\beta}})
   &=(-1)^{(\pmb{\alpha}|\pmb{\beta})}(-1)^{(\pmb{\alpha}|\pmb{\beta})} \cAd U_Q(\gamma_+\times \gamma_-)(\psi^{\pmb{\beta}}) \psi^{\pmb{\alpha}} \\
   &= \cAd U_Q(\gamma_+\times \gamma_-)(\psi^{\pmb{\beta}}) \psi^{\pmb{\alpha}}    
   \end{align*}
   as desired.
  \end{itemize}
  As we have $(\psi^{\pmb{\alpha}})^* = \epsilon(\pmb{\alpha}, \pmb{\alpha})\psi^{-\pmb{\alpha}}$
  and $\epsilon(\pmb{\alpha}, \pmb{\alpha}) \in \{1, -1\}$,
  the commutation of all other combinations follow from this.
  \item (2dCN\ref{2dcn:positiveenergy}) holds because each $U_{p\pmb{\lambda}h}\otimes U_{\bar p\pmb{\lambda}h}$ has positive energy.
  \item For (2dCN\ref{2dcn:vacuum}), $\Omega_{\pmb 0}\otimes\Omega_{\pmb 0} \in \cH_{\pmb 0}\otimes\cH_{\pmb 0} \subset \cH_Q$
  is cyclic for $\cA_Q$ because $\tau_Q(\cA_{p\frh}(I)\otimes \cA_{\bar p\frh}(I))\Omega_{\pmb 0}\otimes\Omega_{\pmb 0}$
  spans $\cH_{\pmb 0}\otimes\cH_{\pmb 0}$,
  while $\psi^{\pmb{\alpha}}\cdot \Omega_{\pmb 0} \otimes\Omega_{\pmb 0} \in \cH_{p\pmb{\alpha}h}\otimes\cH_{\bar p\pmb{\alpha}h}$
  and $\cA_{p\frh}\otimes \cA_{\bar p\frh}$ is irreducibly represented there.
  As for uniqueness, among $U_{p\pmb{\lambda}h}$
  (respectively $U_{\bar p\pmb{\lambda}h}$), only $p\pmb{\lambda} = \pmb 0$ (respectively $\bar p\pmb{\lambda} = \pmb 0$)
  has an invariant vector. They occur at the same time only for $\pmb{\lambda} = \pmb 0$,
  and $\Omega_{\pmb 0}\otimes\Omega_{\pmb 0}$ is the unique (up to a scalar) invariant vector for $U_{p\pmb{\lambda}h}\otimes U_{\bar p\pmb{\lambda}h}$.
  \item To see (2dCN\ref{2dcn:cylinder}), it is enough to show that the spectrum of $U_Q$ with respect to
  the spacelike rotations is contained in $\bbN$.
  On each sector $\sigma_{p\pmb{\lambda}h}\otimes \sigma_{\bar p\pmb{\lambda}h}$,
  they are generated by $L_0\otimes \bb1 - \bb1\otimes L_0$, where $L_0$ is given by \eqref{eq:sugawara},
  and their lowest weights are $\frac{\<p\pmb{\lambda}, p\pmb{\lambda}\>}2, \frac{\<\bar p\pmb{\lambda}, \bar p\pmb{\lambda}\>}2$.
  As $\pmb{\lambda} \in Q$, one has $\frac{\<p\pmb{\lambda}, p\pmb{\lambda}\>}2 - \frac{\<\bar p\pmb{\lambda}, \bar p\pmb{\lambda}\>}2
  = \frac{(\pmb{\lambda}|\pmb{\lambda})}2 \in \bbZ$.
 \end{itemize}
\end{proof}

\begin{example}\label{ex:RR2net}
The two-dimensional net $\cA_Q$ constructed above is an extension of $\cA_{p\frh}\otimes \cA_{\bar p\frh}$.
Let now $Q\subset \frh \cong \bbR^2$ as in Example \ref{ex:RR2example}, generated by $\frac{1}{\sqrt 2}(R \oplus R)$ and $\frac{1}{\sqrt 2}(R^{-1} \oplus (-R^{-1}))$, where $R\in\bbR$, $R\neq 0$.
The chiral conformal nets $\cA_{p\frh}$ and $\cA_{\bar p\frh}$ are not completely rational in the sense of \cite{KLM01},
in fact they are both isomorphic to the $\rmU(1)$-current net $\cA_\bbR$ \cite{BMT88},
where we denote $\cA_{p\frh}\cong \cA_\bbR \cong \cA_{\bar p\frh}$, see Section \ref{sec:U(1)-net}.

For any $R\in\bbR$, $R\neq 0$ the (infinite index) extension $\cA_\bbR \otimes \cA_\bbR \subset \cA_Q$ is expected to be maximal among the local irreducible extensions of $\cA_\bbR \otimes \cA_\bbR$, \cf \cite[Section 6.3]{Moriwaki23}.
Indeed, supporting this idea, if there is a conformal extension $\cA_Q \subset \cB$ such that the decomposition of $\cA_\bbR \otimes \cA_\bbR \subset \cB$
has discrete spectrum, then by Theorem \ref{th:classification} in Section \ref{classification}, it must correspond to an even lattice extending $Q$,
but we saw already that $Q$ is maximal in Example \ref{ex:RR2example}. We do not know whether there is an extension $\cB$ of $\cA_Q$ with
non-discrete spectrum. In general, extensions of $\cA_\bbR\otimes \cA_\bbR$ may have non-discrete spectrum, see Example \ref{ex:nondiscrete}.

Assume now that $R^2 \in \bbQ$, \ie $R^2 = \frac p q$ with $p, q \in \bbN$ coprime.
Then in $Q$ there is an element $(\frac{q}{\sqrt 2}R + \frac{p}{\sqrt2}R^{-1}) \oplus (\frac{q}{\sqrt 2}R - \frac{p}{\sqrt2}R^{-1}) = \sqrt{2pq} \oplus 0$,
and similarly $0 \oplus \sqrt{2pq}$.
Therefore, there is an intermediate two-dimensional extension $\cA_\bbR\otimes \cA_\bbR \subset \cA_{\bbR, \sqrt{2pq}}\otimes \cA_{\bbR, \sqrt{2pq}} \subset \cA_Q$,
where $\cA_{\bbR, \sqrt{2pq}}$ is the local extension of the conformal net $\cA_\bbR$ on $S^1$ as in \cite[Section IV]{BMT88}.
In particular, if $R = 1$, then $\cA_{\bbR, \sqrt{2}}$ is unitarily equivalent to the $\rmS\rmU(2)$-loop group net
at level $1$ \cite{Wassermann98}.
Corresponding full Vertex Operator Algebras have been obtained in \cite[Section 6.3]{Moriwaki23}.
\end{example}

\begin{example}\label{ex:heterotic}
Our construction depends only on the even lattice $Q \subset \bbR^N = \frh$ with respect to
the bilinear form $(\cdot|\cdot) = (p\cdot, p\cdot)_\frh - (\bar p\cdot, \bar p\cdot)_\frh$.
There are examples of such lattices where $\dim p \neq \dim \bar p$, see \cite[Section 7.5]{Moriwaki23}
and references therein. In particular, the central charges of the chiral and antichiral components may be different.
\end{example}

\begin{example}\label{ex:nondiscrete}
We fix an additive subgroup $G \subset \bbR$
and take $Q = \{(\alpha \oplus \alpha) \in \bbR^2: \alpha \in G\}$, with $p$ the projection onto the first component.
Note that for any $\pmb{\alpha}, \pmb{\beta} \in Q$, we have $(\pmb{\alpha}|\pmb{\beta}) = 0$.
Therefore, by the same construction (here $Q$ is not a lattice, but by the property $(\pmb{\alpha}|\pmb{\beta}) = 0$,
we do not need the cocycle $\epsilon$), we have a two-dimensional conformal net $\cA_Q$.
Clearly there are additive subgroups $G$ such that $Q \subset \bbR^2$ is not discrete (in the relative topology), \eg $G = \bbQ, \bbR$.
These examples have been studied in \cite[Section 6]{MTW18}.
\end{example}

\section{Braided equivalences of the chiral components}\label{sec:BE}
\subsection{Braided equivalences and extensions}

We compare our construction of Example \ref{ex:RR2net} with the results of \cite[Section 6.3]{BKL15}.
Below, we construct the \emph{braided equivalence} $\phi_{Q}$ corresponding to the extension $\cA_\bbR \otimes \cA_\bbR \subset \cA_Q$, \cf \cite[Proposition 6.7(6)]{BKL15}, \cite[Theorem 3.6, Corollary 3.8]{DNO13}, with due variations from the completely rational and finite index context. First, we recall the definition of braided tensor functor and equivalence, see \eg \cite[Chapter XI]{MacLane98}, \cite[Chapter 4]{EGNO15}, between braided tensor categories, and of natural transformation between braided tensor functors. The categories will be strictly tensor, unitary (i.e. $C^*$, see \eg \cite{LR97}), with simple tensor unit object and not necessarily fusion.

\begin{definition}\label{def:tensorfunctor}
Let $(\cC,\otimes_\cC,\id_\cC)$ and $(\cD,\otimes_\cD,\id_\cD)$ be two strict tensor categories, where we drop the subscripts in the tensor products $\otimes$ and tensor unit objects $\id$ if no confusion arises. A \textbf{tensor functor} is a triple $(\phi, \mu, \eta)$, where $\phi:\cC \to \cD$ is a functor, $\mu$ is an invertible (unitary if $\cC, \cD$ are $C^*$) natural transformation (see Definition \ref{def:tensornattrafo})
$\mu : \otimes \circ (\phi,\phi) \Rightarrow \phi \circ \otimes$ (\ie a collection of invertible or unitary arrows $\mu_{x,y} : \phi(x)\otimes \phi(y) \to \phi(x\otimes y)$ for each pair of objects $x,y$ in $\cC$), $\eta$ is an invertible (unitary in the $C^*$ case) arrow $\eta : \id \to \phi(\id)$ (the identity arrow $\eta = 1$ if $\phi(\id) = \id$), such that the following diagrams commute
\begin{center}
\begin{tikzcd}
	(\phi(x)\otimes \phi(y))\otimes \phi(z) \arrow [r, rightarrow, " "] \arrow[d, rightarrow, "="]  &  \phi(x\otimes y)\otimes \phi(z) \arrow[r, rightarrow, " "] & 
	\phi((x\otimes y)\otimes z) \arrow[d, rightarrow, "="]  \\
	\phi(x)\otimes (\phi(y)\otimes \phi(z)) \arrow[r, rightarrow, " "]  & \phi(x)\otimes \phi(y\otimes z) \arrow[r, rightarrow, " "] &  \phi(x\otimes (y\otimes z)),
\end{tikzcd}
\end{center}
\begin{center}
\begin{tikzcd}
	\id \otimes \phi(x) \arrow [r, rightarrow, " "] \arrow[d, rightarrow, "="]  &  \phi(\id) \otimes \phi(x) \arrow [d, rightarrow, " "] \\ 
	\phi(x) \arrow[r, rightarrow, "="]  &  \phi(\id \otimes x),
\end{tikzcd}
\quad
\begin{tikzcd}
	\phi(x) \otimes \id \arrow [r, rightarrow, " "] \arrow[d, rightarrow, "="]  &  \phi(x) \otimes \phi(\id) \arrow [d, rightarrow, " "] \\ 
	\phi(x) \arrow[r, rightarrow, "="]  &  \phi(x \otimes \id)
\end{tikzcd}

\end{center}
for every $x,y,z$ objects in $\cC$, where the arrows labelled by $=$ are identity arrows because $\cC$ and $\cD$ are assumed to be strict (\ie $\id \otimes x = x = x \otimes \id$ and $(x \otimes y) \otimes z = x \otimes (y \otimes z)$ both in $\cC$ and $\cD$). Note that the commutativity of the first diagram is a categorical 2-cocycle condition on $\mu = \{\mu_{x,y}\}_{x,y\in\cC}$. $\mu$ and $\eta$ are respectively called the tensorator and unitor of the tensor functor.

A tensor functor is a \textbf{tensor equivalence} \cite[Definition 2.4.1]{EGNO15} if it is, in addition, an equivalence of categories \cite[Section IV.4]{MacLane98}. 
A tensor functor is called strict if all $\mu_{x,y}$ and $\eta$ (hence all arrows in the diagrams) are identity arrows.
\end{definition}

\begin{definition}\label{def:braidedfunctor}
Let $(\cC,\otimes_\cC,\id_\cC,\varepsilon_\cC)$ and $(\cD,\otimes_\cD,\id_\cD,\varepsilon_\cD)$ be in addition braided with both braidings denoted by $\varepsilon$ (hence by $\varepsilon_{x,y} : x \otimes y \to y \otimes x$ for every $x,y$ in $\cC$ or $\cD$) if no confusion arises (and assumed to be unitary in the $C^*$ case). A braided tensor functor, or \textbf{braided functor} for short, is just a tensor functor such that, in addition, the following diagram commutes
$$
\begin{tikzcd}
	\phi(x) \otimes \phi(y) \arrow [r, rightarrow, " "] \arrow[d, rightarrow, " "]  &  \phi(x\otimes y) \arrow[d, rightarrow, " "] \\ 
	\phi(y) \otimes \phi(x) \arrow [r, rightarrow, " "]  &  \phi(y \otimes x)
\end{tikzcd}
$$
for every $x,y$ objects in $\cC$, \ie such that $\phi(\varepsilon_{x,y})\, \mu_{x,y} = \mu_{y,x}\, \varepsilon_{\phi(x),\phi(y)}$.

A braided functor is a \textbf{braided equivalence} \cite[Definition 8.1.7]{EGNO15} if it is, in addition, an equivalence of categories.
If $\cC = \cD$, braided/tensor equivalences are also called autoequivalences. 
\end{definition}

Note that being tensor for a functor is additional structure with constraints, while being braided for a tensor functor is just a constraint. 
Recall also the following

\begin{definition}\label{def:tensornattrafo}
Let $\phi_1,\phi_2:\cC\rightarrow\cD$ be two functors. A \textbf{natural transformation} $\nu$ from $\phi_1$ to $\phi_2$, denoted by $\nu : \phi_1 \Rightarrow \phi_2$, is a collection of arrows $\nu = \{\nu_x\}_{x\in\cC}$ in $\cD$, one for every object of $\cC$, such that $\nu_x : \phi_1(x) \to \phi_2(x)$ in $\cD$, which is natural in $x$ in the sense that the following diagram commutes
\begin{center}
\begin{tikzcd}
	\phi_1(x) \arrow [r, rightarrow, "\phi_1(t)"] \arrow[d, rightarrow, "\nu_x"]  &  \phi_1(y) \arrow[d, rightarrow, "\nu_y"] \\ 
	\phi_2(x) \arrow [r, rightarrow, "\phi_2(t)"]                                                      &  \phi_2(y)
\end{tikzcd}
\end{center}
for every arrow $t:x\rightarrow y$ in $\cC$, \ie such that $\nu_y\,\phi_1(t) = \phi_2(t)\, \nu_x$ (naturality). In words, $\nu$ intertwines $\phi_1$ and $\phi_2$ both on objects and on morphisms.

A natural transformation $\nu: \phi_1\Rightarrow \phi_2$ is called a \textbf{natural isomorphism} if all the arrows $\nu_x$ are isomorphisms (\ie invertible arrows, or unitary in the $C^*$ case) in $\cD$. 

Let now $(\phi_1, \mu_1,\eta_1) ,(\phi_2,\mu_2,\eta_2):\cC\rightarrow\cD$ be tensor functors as before. A natural transformation $\nu:\phi_1 \Rightarrow \phi_2$ is called a \textbf{tensor} natural transformation if, in addition, it makes the following diagrams commute
\begin{center}
\begin{tikzcd}
	\phi_1(x) \otimes \phi_1(y) \arrow [r, rightarrow, " "] \arrow[d, rightarrow, " "]  &  \phi_1(x\otimes y) \arrow[d, rightarrow, " "] \\ 
	\phi_2(x) \otimes \phi_2(y) \arrow [r, rightarrow, " "]                                            &  \phi_2(x \otimes y),
\end{tikzcd}
\begin{tikzcd}
	\id \arrow [r, rightarrow, " "] \arrow [dr, rightarrow, " "]                  &         \phi_1(\id) \arrow[d, rightarrow, " "]                                               \\
	                                                                                                                &  \phi_2(\id) 
\end{tikzcd}
\end{center}
for every $x,y$ objects in $\cC$, \ie such that $\nu_{x \otimes y}\, {\mu_1}_{x,y} = {\mu_2}_{x,y} \, \nu_{x} \otimes \nu_{y}$ and $\nu_{\id} \eta_1 = \eta_2$.

A tensor natural isomorphism is a tensor natural transformation which is also an isomorphism. If the tensor functors are also braided, adding the word \textbf{braided} to tensor natural transformations amounts to no additional constraint.
\end{definition}

To construct the braided equivalence $\phi_{Q}$ associated with Example \ref{ex:RR2net}, we assume that $R^2 \notin \bbQ$.
Let, as before, $h \in C^\infty(S^1)$ be chosen such that $\frac1{2\pi}\int_{-\pi}^{\pi} h(t)dt = 1$ and localized in some $I \in \cI$.

In this case,
\begin{align*}
 \cH_Q \cong \bigoplus_{a,b \in \bbZ} \cH_{\bbR, (\frac{a}{\sqrt 2}R + \frac{b}{\sqrt 2}R^{-1})h} \otimes \cH_{\bbR, (\frac{a}{\sqrt 2}R - \frac{b}{\sqrt 2}R^{-1})h}.
\end{align*}
The representation of $\cA_\bbR \otimes \cA_\bbR$ on $\cH_Q$ is $\tau_Q = \bigoplus_{a,b \in \bbZ} \sigma_{(\frac{a}{\sqrt 2}R + \frac{b}{\sqrt 2}R^{-1})h}\otimes \sigma_{(\frac{a}{\sqrt 2}R - \frac{b}{\sqrt 2}R^{-1})h}$.
We let $\cC$ be the full subcategory of $\Rep(\cA_\bbR)$ generated by the automorphisms (with categorical dimension one, thus irreducible) $\sigma_{\frac{1}{\sqrt 2}Rh}$ and $\sigma_{\frac{1}{\sqrt 2}R^{-1}h}$, closed under tensor products, conjugates, directs sums and subobjects in $\Rep(\cA_\bbR)$.
Recall that fullness means that the hom-spaces between objects in $\cC$ coincide with those in $\Rep(\cA_\bbR)$.
The unitary equivalence classes of irreducible objects in $\cC$ and $\Rep(\cA_\bbR)$, respectively, are in one-to-one correspondence with $\bbZ \frac{R}{\sqrt 2} + \bbZ \frac{R^{-1}}{\sqrt 2}$ and with $\bbR$. Moreover, by definition with $h$ fixed, the composition (the categorical tensor product in $\cC$, that we either denote by $\circ$ or by no symbol) of tensor powers of irreducibles in $\cC$ is given by $\sigma_{\frac{1}{\sqrt 2}Rh}^n \sigma_{\frac{1}{\sqrt 2}R^{-1}h}^m = \sigma_{(\frac{n}{\sqrt 2}R + \frac{m}{\sqrt 2}R^{-1})h}$ for every $n,m\in\bbZ$ (in particular the composition is commutative), and the inverse (categorical conjugate in the case of automorphisms) is $\bar\sigma_{\frac{1}{\sqrt 2}Rh} = \sigma_{\frac{-1}{\sqrt 2}Rh}$, $\bar\sigma_{\frac{1}{\sqrt 2}R^{-1}h} = \sigma_{\frac{-1}{\sqrt 2}R^{-1}h}$. On the irreducible objects of $\cC$, which we denote below by $\sigma_{n,m} := \sigma_{(\frac{n}{\sqrt 2}R + \frac{m}{\sqrt 2}R^{-1})h}$, we define $\phi_Q$ by
\begin{align*}
\phi_{Q}(\sigma_{n,m}) := \sigma_{-n,m}.
\end{align*}
On the tensor unit of $\cC$ (the identity automorphism $\id = \sigma_{0,0}$) it holds $\phi_Q(\id) = \id$.
On arrows between the irreducibles $\sigma_{n,m}$ and $\sigma_{n',m'}$ for $n,m,n',m'\in\bbZ$, which are either $\lambda 1_{\sigma_{n,m}}$ for some $\lambda \in \bbC$ if $n=n'$ and $m=m'$, or $0$ otherwise (indeed, $\frac{n}{\sqrt 2}R + \frac{m}{\sqrt 2}R^{-1} = 0$ with $n,m\neq 0$ if and only if $R^2 = -\frac{m}{n}\in\bbQ$), we define $\phi_Q$ by linearity, namely by setting $\phi_Q(\lambda 1_{\sigma_{n,m}}) := \lambda 1_{\sigma_{-n,m}}$. 

In this way, the representation $\tau_Q$ can be written as
\begin{align*}
\tau_Q = \bigoplus_{n,m \in \bbZ} \sigma_{n,m}\otimes \phi_Q(\bar\sigma_{n,m}),
\end{align*}
\cf \cite[Section 3.2]{DNO13}, \cite[Definition 4.1]{BKL15}. 

Note that $\sigma_{n,m}\sigma_{n',m'} = \sigma_{n',m'}\sigma_{n,m} = \sigma_{n+n',m+m'}$. We define $\phi_Q$ to be a non-strict tensor functor with tensorator $\mu$ given by the 2-cocycle $\epsilon : Q \times Q \to \{\pm 1\}$ from \eqref{eq:twococycle}, namely
\begin{align*}
\phi_Q(\sigma_{n,m}) \phi_Q(\sigma_{n',m'}) \xlongrightarrow{\mu_{(n,m),(n',m')} 1\quad} \phi_Q(\sigma_{n,m}\sigma_{n',m'}),
\end{align*}
where $1$ above is a short-hand notation for $1_{\sigma_{-n-n',m+m'}}$, the identity arrow between the common source and target automorphism $\sigma_{-n-n',m+m'}$, and the scalar with values in $\{\pm 1\}$ defined by
\begin{align*}
\mu_{(n,m),(n',m')} := \epsilon(\pmb\alpha_{n,m}, \pmb\alpha_{n',m'}),
\end{align*}
where $\pmb\alpha_{n,m} := \frac{n}{\sqrt 2} (R\oplus R) + \frac{m}{\sqrt 2} (R^{-1}\oplus (-R^{-1})) \in Q$ for $n,m\in\bbZ$. In this notation, note that $\pmb\alpha_{n,m} + \pmb\alpha_{n',m'} = \pmb\alpha_{n+n',m+m'}$. In order to exploit the definition of $\epsilon$, we choose $\pmb v_1 := \pmb \alpha_{1,0} = \frac{1}{\sqrt 2} (R\oplus R)$, $\pmb v_2 := \pmb \alpha_{0,1} = \frac{1}{\sqrt 2} (R^{-1}\oplus (-R^{-1}))$ as an ordered basis of $Q$, and compute, using the bimultiplicative extension of \eqref{eq:twococycle} together with \eqref{eq:semidefscalarprodinQ2d-1}, \eqref{eq:semidefscalarprodinQ2d-2}, its value $\epsilon(\pmb\alpha_{n,m}, \pmb\alpha_{n',m'}) = \epsilon(\pmb\alpha_{1,0}, \pmb\alpha_{0,1})^{nm'} = (-1)^{nm'}$. The scalars $\mu_{(n,m),(n',m')}$ clearly define a tensorator (see Definition \ref{def:tensorfunctor} and \cite[Section XI.2, Equation (3)]{MacLane98}) for the (non-strict) tensor functor $\phi_Q : \cC \to \cC$ 
between strict tensor categories (after extending $\phi_Q$ to direct sums of irreducible objects in $\cC$), 
because $\epsilon$ is a 2-cocycle on $Q$. 
Note that $\phi_Q$ is a tensor autoequivalence $\cC \to \cC$ with inverse $\phi_Q^{-1} = \phi_Q$. 
Furthermore, $\phi_Q$ is also braided (see Definition \ref{def:braidedfunctor} and \cite[Section XI.2, Equation (10)]{MacLane98}), namely the following diagram commutes for every $n,m,n',m'\in\bbZ$
\begin{center}
\begin{tikzcd}
	\phi_Q(\sigma_{n,m}) \phi_Q(\sigma_{n',m'}) \arrow [r, rightarrow, " "] \arrow[d, rightarrow, " "]  &  
	\phi_Q(\sigma_{n,m}\sigma_{n',m'}) \arrow[d, rightarrow, " "] \\ 
	\phi_Q(\sigma_{n',m'})\phi_Q(\sigma_{n,m}) \arrow [r, rightarrow, " "]  &  
	\phi_Q(\sigma_{n',m'}\sigma_{n,m}),
\end{tikzcd}
\end{center}
where the horizontal arrows are $\mu_{(n,m),(n',m')} 1 = (-1)^{nm'} 1$ above and $\mu_{(n',m'),(n,m)} 1= (-1)^{n'm} 1$ below, and the vertical arrows are given by the braidings \eqref{eq:braidingNdim}, namely on the left
\begin{align*}
\varepsilon^+_{\phi_Q(\sigma_{n,m}),\phi_Q(\sigma_{n',m'})} = \varepsilon^+_{\sigma_{-n,m},\sigma_{-n',m'}} = \rme^{\rmi\pi(\frac{nn'}{2}R^2 + \frac{mm'}{2}R^{-2} - \frac{nm'}{2} - \frac{mn'}{2})} 1,
\end{align*}
and on the right
\begin{align*}
\phi_Q(\varepsilon^+_{\sigma_{n,m},\sigma_{n',m'}}) = \rme^{\rmi\pi(\frac{nn'}{2}R^2 + \frac{mm'}{2}R^{-2} + \frac{nm'}{2} + \frac{mn'}{2})} 1.
\end{align*}

Thus, the commutativity of the previous diagram, namely the following equality of scalars
\begin{align*}
(-1)^{nm'} \rme^{\rmi\pi(\frac{nn'}{2}R^2 + \frac{mm'}{2}R^{-2} + \frac{nm'}{2} + \frac{mn'}{2})} = (-1)^{n'm} \rme^{\rmi\pi(\frac{nn'}{2}R^2 + \frac{mm'}{2}R^{-2} - \frac{nm'}{2} - \frac{mn'}{2})}
\end{align*}
is readily checked by using \eg $(-1) = \rme^{\rmi \pi}$ and $\rme^{\rmi\frac{3\pi}{2}} = \rme^{-\rmi\frac{\pi}{2}}$, indeed
\begin{align*}
\rme^{\rmi\pi(\frac{nn'}{2}R^2 + \frac{mm'}{2}R^{-2} + \frac{3nm'}{2} + \frac{mn'}{2})} = \rme^{\rmi\pi(\frac{nn'}{2}R^2 + \frac{mm'}{2}R^{-2} - \frac{nm'}{2} + \frac{mn'}{2})},
\end{align*}
Thus $\phi_Q$ is a braided functor.

\begin{remark}
The braided autoequivalence $\phi_Q$ extends to a braided autoequivalence of the replete completion of $\cC$ in $\Rep(\cA_\bbR)$, denoted by $\tilde\cC$, obtained by including all objects in $\Rep(\cA_\bbR)$ that are (unitarily) equivalent to some object in $\cC$. If $\sigma$ is irreducible in $\tilde\cC$, choose a unitary arrow (intertwining operator, unique up to scalar factors) $u : \sigma \to \sigma_{n,m}$ for a unique pair $n,m\in\bbZ$. We set $\tilde\phi_Q(\sigma) := \sigma_{-n,m}$. If $\sigma_1$ and $\sigma_2$ are two irreducibles in $\tilde\cC$, then they are either mutually inequivalent and the only intertwiner is $0$, or every arrow $t:\sigma_1 \to \sigma_2$ equals $\lambda u_2^* u_1$ where $\lambda\in\bbC$ is some scalar depending on $t$ and $u_1 : \sigma_1 \to \sigma_{n,m}$, $u_2 : \sigma_2 \to \sigma_{n,m}$ are the unitaries chosen as before. In the latter case, we set $\tilde\phi_Q(t) := \lambda 1_{\sigma_{-n,m}}$. One can check that functoriality holds, \ie $\tilde\phi_Q(st) = \tilde\phi_Q(s)\tilde\phi_Q(t)$ if $t:\sigma_1 \to \sigma_2$ and $s:\sigma_2 \to \sigma_3$.
Extending to direct sums of irreducible objects gives a braided autoequivalence $\tilde\phi_Q:\tilde \cC \to \tilde \cC$, whose  
isomorphism class among the braided autoequivalences of $\tilde\cC$ does not depend on the choice of function $h\in\cC^\infty(S^1)$ made at the beginning.

Indeed, let $g\in\cC^\infty(S^1)$ be another function with $\frac1{2\pi}\int_{-\pi}^{\pi} g(t)dt = 1$ and localized in some other proper interval of $S^1$. There is a unitary tensor (\ie braided tensor) natural transformation (see Definition \ref{def:tensornattrafo}) from $\tilde\phi_Q$ to the functor $\tilde\phi_{Q,g} : \tilde\cC \to \tilde\cC$ defined in the same way using $g$ instead of $h$. Let $M\in\cC^\infty(S^1)$ be a primitive of $g-h$, namely $M(t) := \frac1{2\pi}\int_{-\pi}^{t} (g(s)-h(s)) ds$ so that $M(-\pi) = M(\pi) = 0$.
Let $u:= W(\frac{-1}{\sqrt 2}RM)$ be the Weyl unitary that intertwines $\sigma_{\frac{-1}{\sqrt 2}Rh} (= \sigma_{-1,0}$ in our previous notation) with $\sigma_{\frac{-1}{\sqrt 2}Rg}$, and $v:= W(\frac{1}{\sqrt 2}R^{-1}M)$ between $\sigma_{\frac{1}{\sqrt 2}R^{-1}h} (= \sigma_{0,1}$) with $\sigma_{\frac{1}{\sqrt 2}R^{-1}g}$. Observe that $u$ and $v$ commute because $\int_{-\pi}^{\pi} M(t)M'(t)dt = \frac{1}{2} [M^2]_{-\pi}^{\pi} = 0$, hence $uv = vu = W((\frac{-1}{\sqrt 2}R + \frac{1}{\sqrt 2}R^{-1})M)$. 
For $\sigma$ irreducible in $\tilde\cC$, let $\nu_{\sigma} := u \otimes \stackrel{(n)}{\ldots} \otimes u \otimes v \stackrel{(m)}{\ldots} \otimes v$ be the unitary intertwiner between $\tilde\phi_Q(\sigma) = \sigma_{-n,m}$ and $\tilde\phi_{Q,g}(\sigma) = \sigma_{-n,m,g} (:= \sigma_{(\frac{-n}{\sqrt 2} R + \frac{m}{\sqrt 2}R^{-1})g})$, where the symbol $\otimes$ denotes the tensor product of intertwiners in the category $\Rep(\cA_\bbR)$. Naturality of the collection of arrows $\{\nu_\sigma\}$ can be checked as for functoriality, using that all arrows in the naturality diagram are either zero or scalars and that $\sigma_{n,m}$ and $\sigma_{n,m,g}$ are unitarily equivalent. 
To check tensoriality of the natural transformation $\nu$, observe that the tensorators $\mu_{(n,m),(n',m')}$ only depend on $n,m,n',m'$ (they are the 2-cocycle $\varepsilon$) and not on $h,g$, hence they are the same for the two functors. Thus, we only have to check that 
$\nu_{\sigma} \otimes \nu_{\sigma'} = \nu_{\sigma \otimes \sigma'}$, 
where recall that the tensor product of automorphisms in $\Rep(\cA_\bbR)$ is just composition.
Observe that $v \otimes u = v \sigma_{0,1}(u) = \rme^{\frac{\rmi}{2\pi} \int \frac{1}{\sqrt 2}R^{-1} h(t) \frac{-1}{\sqrt 2}R M(t)dt} vu$ 
and $u \otimes v = u \sigma_{-1,0}(v) = \rme^{\frac{\rmi}{2\pi} \int \frac{-1}{\sqrt 2}R h(t) \frac{1}{\sqrt 2}R^{-1} M(t)dt} uv$, thus the phase factor is the same and we conclude $v \otimes u = u \otimes v$. Consequently, $\nu_{\sigma} \otimes \nu_{\sigma'} = \nu_{\sigma \otimes \sigma'}$ as unitary intertwiners between $\sigma_{-n-n',m+m'}$ and $\sigma_{-n-n',m+m',g}$, showing tensoriality of $\nu : \tilde\phi_Q \Rightarrow \tilde\phi_{Q,g}$.
\end{remark}

\subsection{Braided equivalences and unitary (non-)equivalence of extensions}\label{unitaryequiv}
When one considers the problem of classifying extensions $\cA_+\otimes\cA_- \subset \cB$ up to unitary equivalence,
there are two possibilities: either one requires that the ``unitary equivalence'' intertwines
also the representation $\tau_{\cB}$ of $\cA_+\otimes\cA_-$ on the Hilbert space of $\cB$ or not.
The stronger notion of equivalence for extensions is implicitly used in \cite[Proposition 6.7(1)]{BKL15}, \cf \cite[Definition 3.4]{DR90}.
In \cite[Section 6]{BKL15}, two extensions $\cA_+\otimes \cA_- \subset \cB_1, \cB_2$ are called unitarily equivalent if the nets $\cB_1, \cB_2$ are unitarily equivalent via some unitary operator $U$, and the same operator intertwines the representations $\tau_1, \tau_2$
of $\cA_+\otimes \cA_-$ obtained by restricting the vacuum representations of $\cB_1, \cB_2$, respectively, \ie $\cAd U \circ \tau_1 = \tau_2$.
It may happen that, even if two extensions $\cA_+\otimes\cA_- \subset \cB_j, j = 1,2$ are not unitarily
equivalent, two nets $\cB_i$ are unitarily equivalent. We exhibit this in examples.
In order to stress that we consider the stronger notion, we say that the \textit{extensions} $\cA_+\otimes\cA_- \subset \cB_j$
\textit{are unitarily equivalent}.

Let us consider the extensions of Example \ref{ex:nondiscrete} with $G = \bbZ$:
$\frh = \bbR^2$, $p$ is the projection onto the first component, $Q = \{(\alpha \oplus \alpha) \in \bbR^2: \alpha \in \bbZ\}$, and we obtain a two-dimensional conformal net $\cA_Q$ by Theorem \ref{th:2dnetQ}.
On the other hand, we can also take $Q' = \{(\alpha \oplus -\alpha) \in \bbR^2: \alpha \in \bbZ\}$, and we obtain another two-dimensional conformal net $\cA_{Q'}$ (the same net constructed in \cite[Section 6]{AGT23Pointed}). Both nets extend the tensor product of $\rmU(1)$-current nets $\cA_\bbR \otimes \cA_\bbR$. We discuss the equivalence of these two-dimensional conformal extensions at the end of the section.

Note that on $\cA_\bbR(I)$ there is an automorphism $\theta(W(f)) = W(-f)$,
This is implemented by the unitary operator $V$ which acts as $1$ on the even particle number space
and $-1$ on the odd particle number space, in particular, $V^2 = 1$.
Fix $h \in C^\infty(S^1)$ localized in some $I \in \cI$ as before.

We have $\cAd V \circ \sigma_{\alpha h} = \sigma_{-\alpha h} \circ \cAd V$ elementwise on $\cA_\bbR$, indeed
\begin{align*}
 &\cAd V\circ \sigma_{\alpha h}(W(f)) = \rme^{\frac{\rmi\alpha}{2\pi}\int f(t)h(t)dt}\cAd V(W(f)) = \rme^{\frac{\rmi\alpha}{2\pi}\int f(t)h(t)dt}(W(-f)) \\
 &= \rme^{-\frac{\rmi\alpha}{2\pi}\int (-f(t))h(t)dt}(W(-f)) = \sigma_{-\alpha h}(W(-f)) = \sigma_{-\alpha h}\circ \cAd V(W(f)).
\end{align*}

As $\cAd V$ is a vacuum-preserving automorphism of the net $\cA_\bbR$, it commutes with its implementation of the diffeomorphisms $U(\gamma)$
\cite{CW05}, \cite[Proposition 4.2]{CLTW12-1}, \cite[Theorem 6.10]{CKLW18}
(one can also see this from $\cAd (J_n) = -J_n$ and the fact that $L_n$ in \eqref{eq:sugawara} are quadratic in $J$).
As the representations $U_{\alpha h}$ on $\cH_{\alpha h}$ is (locally) given by $U_{\alpha h}(\gamma) = \sigma_{\alpha h}(U(\gamma))$,
we have
\begin{align*}
 \cAd V(U_{\alpha h}(\gamma)) = \cAd V\circ \sigma_{\alpha h}(U(\gamma)) = \sigma_{-\alpha h}\circ \cAd V(U(\gamma)) = \sigma_{-\alpha h}(U(\gamma))
 = U_{-\alpha h}(\gamma).
\end{align*}

Let $\tilde V := \bigoplus_{\alpha\in \bbZ} \bb1\otimes V$, we clearly have
$\tilde V (\bigoplus_{\alpha\in \bbZ} \cH_{\alpha h}\otimes \cH_{\alpha h}) = \bigoplus_{\alpha\in \bbZ} \cH_{\alpha h}\otimes \cH_{-\alpha h}$
\ie $\tilde V: \cH_Q \to \cH_{Q'}$, and recall the notation for the representations of $\cA_\bbR\otimes\cA_\bbR$ respectively
$\tau_Q = \bigoplus_{\alpha \in \bbZ} \sigma_{\alpha h}\otimes \sigma_{\alpha h}$ on $\cH_Q$ and 
$\tau_{Q'} = \bigoplus_{\alpha \in \bbZ} \sigma_{\alpha h}\otimes \sigma_{-\alpha h}$ on $\cH_{Q'}$.
Furthermore, let us denote by $\psi^{\pmb{\alpha}}_Q$, $\psi^{\pmb{\alpha}}_{Q'}$, where $\pmb \alpha \in Q$ or $Q'$, the corresponding shift operators (twisted by the 2-cocycle of $Q$ or $Q'$) as in Section \ref{sec:2d-chiral}.
Then we have
\begin{itemize}
 \item $\cAd \tilde V \circ \tau_Q = \tau_{Q'} \circ \cAd \tilde V$ elementwise on $\cA_\bbR\otimes\cA_\bbR$
 (but \emph{not} $\cAd \tilde V \circ \tau_Q = \tau_{Q'}$, i.e. $\tilde V$ is \emph{not} a unitary intertwiner between the two representations of $\cA_\bbR\otimes\cA_\bbR$, as $V$ is clearly not a unitary intertwiner between the irreducibles $\rho_{\alpha h}$ and $\rho_{-\alpha h}$, which is impossible unless $\alpha = 0$).
 \item $\cAd \tilde V(U_Q(\gamma_+\times \gamma_-)) = U_{Q'}(\gamma_+\times \gamma_-)$.
 \item $\tilde V \Omega_Q = \Omega_{Q'}$.
 \item $\cAd \tilde V (\psi^{\pmb{\alpha}}_Q) = \psi^{\pmb{\alpha}}_{Q'}$
 (to show this, we identify $\cH_{\alpha h} = \cH_0$ as Hilbert spaces,
 and note that $V$ acts on the same unitary on each of them and $\psi^{\pmb{\alpha}}_Q, \psi^{\pmb{\alpha}}_{Q'}$ do only the shift).
 \item For $I\times I$, where $I$ is the distinguished interval containing $\supp h$, we have
 \begin{align*}
  \cAd \tilde V(\cA_Q(I\times I)) &= \cAd \tilde V(\{\tau_Q(W(f_+)\otimes W(f_-)): \supp f_\pm \subset I, \psi^{\pmb{\alpha}}_Q\}'') \\
 &= \{\tau_{Q'}(W(f_+)\otimes W(-f_-)): \supp f_\pm \subset I, \psi^{\pmb{\alpha}}_{Q'}\}'' \\
 &= \{\tau_{Q'}(W(f_+)\otimes W(f_-)): \supp f_\pm \subset I, \psi^{\pmb{\alpha}}_{Q'}\}'' \\
 &= \cA_{Q'}(I\times I),
 \end{align*}
 where the equality is global on the local algebras over $I \times I$ but not elementwise.
 \item For any other double cone, we have the equality by covariance, hence $\cA_Q$ and $\cA_{Q'}$ are unitarily equivalent via $\tilde V$ as two-dimensional conformal nets \emph{per se}.
 \item Moreover, the two copies of the two-dimensional Heisenberg net (two subnets of $\cA_Q$ and $\cA_{Q'}$) are also unitarily equivalent via $\tilde V$, 
 globally on the local algebras over $I \times I$ hence for every double cone, indeed
 \begin{align*}
  \cAd \tilde V(\tau_Q(\cA_\bbR(I)\otimes \cA_\bbR(I)) &= \cAd \tilde V(\{\tau_Q(W(f_+)\otimes W(f_-)): \supp f_\pm \subset I\}'') \\
 &= \{\tau_{Q'}(W(f_+)\otimes W(-f_-)): \supp f_\pm \subset I\}'') \\
 &= \{\tau_{Q'}(W(f_+)\otimes W(f_-)): \supp f_\pm \subset I\}'') \\
 &= \tau_{Q'}(\cA_\bbR(I)\otimes\cA_\bbR(I)).
 \end{align*}
\end{itemize}
Hence $\cA_Q \supset \tau_Q(\cA_\bbR \otimes \cA_\bbR)$ and $\cA_{Q'} \supset \tau_{Q'}(\cA_\bbR \otimes \cA_\bbR)$ are \emph{not} unitarily equivalent
\textit{as extensions of nets}, as $\tau_Q$ and $\tau_{Q'}$ are not unitarily equivalent representations of $\cA_\bbR \otimes \cA_\bbR$.
However, two nets $\cA_Q, \cA_{Q'}$ are unitarily equivalent as two-dimensional conformal nets.

\begin{remark}
The braided autoequivalences (\cf \cite[Proposition 6.7(1)]{BKL15} in the completely rational and finite index context) corresponding to the extensions are, respectively, $\phi_Q(\sigma_{\alpha h}) = \sigma_{-\alpha h}$ and $\phi_{Q'}(\sigma_{\alpha h}) = \sigma_{\alpha h}$ (the identity autoequivalence). In particular, there is no natural transformation $\nu : \phi_Q \Rightarrow \phi_{Q'}$ (see Definition \ref{def:tensornattrafo}) as there is no unitary intertwiner between $\sigma_{-\alpha h} (= \bar \sigma_{\alpha h})$ and $\sigma_{\alpha h}$ for $\alpha \in \bbR$, unless $\alpha = 0$. 
\end{remark}

Let again $\cA_+\otimes \cA_- \subset \cB_1, \cB_2$ be two extensions realized on the Hilbert spaces $\cH_1, \cH_2$. 
Assume that there is a unitary operator $V: \cH_1 \to \cH_2$
such that $\cAd V(\cB_1(D)) = \cB_2(D)$ for any double cone $D$, $\cAd V (U_1(\gamma)) = U_2(\gamma)$
for any $\gamma \in \uDiffSone\times \uDiffSone$ and $V\Omega_1 = \Omega_2$.
By the second property, $\cAd V$ intertwines the actions of diffeomorphisms, thus the maximal chiral components are intertwined as well, in particular
$\cAd V (\tau_1(\cA_+(I)\otimes \bb1)) = \tau_2(\cA_+(I)\otimes \bb1),
\cAd V (\tau_1(\bb1\otimes \cA_-(I))) = \tau_2(\bb1\otimes \cA_-(I))$, where
$\tau_1, \tau_2$ are the representations of $\cA_+\otimes \cA_-$ on $\cH_1, \cH_2$, respectively.
By restricting the latter to their vacuum parts respectively denoted by $\tau_{1,0}, \tau_{2,0}$,
we have that $\cAd V (\tau_{1,0}(\cA_+(I)\otimes \bb1)) = \tau_{2,0}(\cA_+(I)\otimes\bb1)$
and $\cAd V (\tau_{1,0}(\bb1\otimes \cA_-(I))) = \tau_{2,0}(\bb1\otimes \cA_-(I))$.
If we identify $\tau_{1,0} = \tau_{2,0}$, this means that $\cAd V$ gives a vacuum-preserving automorphism of $\cA_+\otimes \cA_-$.
And this should intertwine all the other representations as above.

\section{Classification of two-dimensional extensions}\label{classification}
We classify two-dimensional extensions $\cA_{\frh_p} \subset \cB$ of Heisenberg nets.
More precisely, we consider the following situation:
let $\frh$ be an $N$-dimensional real Hilbert space, $p$ an orthogonal projection on $\frh$, and
$\cA_{\frh_p} := \cA_{p\frh}\otimes \cA_{\bar p\frh}$ the two-dimensional conformal with chiral components both Heisenberg nets as in Section \ref{sec:2d-chiral}.
Furthermore, we assume that $U(\gamma_+\times\iota) \in \cA_{p\frh}(I_+)$ if $\supp \gamma_+ \subset I_+$
and $U(\iota\times\gamma_-) \in \cA_{\bar p\frh}(I_-)$ if $\supp \gamma_- \subset I_-$.
In this case, we say that $\cB$ is a \textbf{conformal extension} of $\cA_{\frh_p}$.
This implies that, for (lightlike) translations $T_t$,
we have $U(T_t\times \iota), U(\iota\times T_t) \in \overline{\bigcup_{I_\pm \in \cI}\tau_\cB(\cA_{p\frh}(I_+)\otimes \cA_{\bar p\frh}(I_-))}$,
where $\tau_\cB$ is the representation of $\cA_{p\frh}\otimes\cA_{\bar p\frh}$ on the Hilbert space $\cH$ of the net $\cB$.
We wish to classify two-dimensional conformal extensions $\cA_{p\frh}\otimes \cA_{\bar p\frh} \subset \cB$, in the sense of Section \ref{1d2dconformalnet},
up to unitary equivalence (see Section \ref{unitaryequiv}).

In full generality, the problem looks rather difficult. Indeed, by \cite[Section 3.3]{Longo03}, \cite[Section 2.3]{Carpi04}, $\cH$ can be seen as a representation
of $\cA_{\frh_p}$ and assuming that it decomposes into irreducible representations, their unitary equivalence classes should form a subgroup of $(\bbR^N,+)$
with respect to the fusion rule. Thus we would need to know all possible
additive subgroups of $\bbR^N$, but we are unaware of a nice such parametrization.

Instead, if we require that such a subgroup forms a discrete subgroup of $(\bbR^N,+)$ (with the induced topology), then it is a (not necessarily $N$-dimensional)
lattice in $\bbR^N$ \cite[(4.2) Proposition]{Neukirch99AlgebraicNumberTheory}.
Under this assumption, we classify all such conformal extensions $\cA_{p\frh}\otimes \cA_{\bar p\frh} \subset \cB$. 
The result (see Theorem \ref{th:classification}) is that any such extension corresponds to
a (not necessarily $N$-dimensional) lattice which is even (thus necessarily integral, \ie $\bbZ$-valued)
with respect to $(\cdot|\cdot)$ defined as in Section \ref{sec:2d-twist},
and to any such even lattice $Q$, there is a unique extension $\cA_Q$ from Section \ref{2dlatticenet}.

The discreteness assumption has appeared also in \cite{BMT88} in their classification of relatively local extensions
of the $\rmU(1)$-current net. Here, for some basic parts, our arguments closely follow the one-dimensional case of \cite[Proposition 3.1]{BMT88}.

\begin{lemma}\label{lm:relativelocal}
 Let $\cA_+$, $\cA_-$ be conformal nets on $S^1$, $\cA_+\otimes\cA_-$ be the conformal net on $\bbR^{1+1}$ as in Section \ref{1d2dconformalnet}
 and $\cB$ be a conformal extension of $\cA_+\otimes \cA_-$ in the sense specified above.
 Let $I_+\times I_- \subset \bbR^{1+1}$ (in the sense of Section \ref{2dconformalnet}). Then $\cA_+(I_+')\otimes \bb1$ commutes with $\cB(I_+\times I_-)$.
\end{lemma}
\begin{proof}
 We take $I \in \cI$ such that $\overline I \subset I_+'$.
 There is $\gamma \in \uMob$ (acting on $\bbR$) that maps $I_+$ onto itself
 and $\gamma I\subset \bbR$. With such $\gamma$, it holds that $\cAd U(\gamma \times \iota)(\cB(I_+\times I_-)) = \cB(I_+\times I_-)$,
 while $\cAd U(\gamma \times \iota)(\cA(I)\otimes\bb1) = \cA_+(\gamma I)\otimes \bb1$.
 Hence, as $\cB(I_+\times I_-)$ and $\cA_+(\gamma I)\otimes \bb1$ commute,
 so do $\cB(I_+\times I_-)$ and $\cA_+(I)\otimes \bb1$.
 As $I \subset I_+$ was arbitrary under the condition $\overline I \subset I_+'$,
 we conclude that $\cB(I_+\times I_-)$ and $\cA_+(I_+')\otimes \bb1$ commute.
\end{proof}

Recall that a representation $\rho$ of a conformal net $\cA_\kappa$ on $S^1$ is said to be \textbf{factorial}
if $\rho(\cA_\kappa)'' := \overline{\bigcup_{I \in \cI}\rho_I(\cA_\kappa(I)))}$ is a factor.

\begin{proposition}\label{pr:tensorproduct2dsector}
 Assume that $\cA_\pm$ are type I conformal nets \cite[Section 3.2]{LT18},
 that is, all their factorial representations $\rho$ are irreducible, thus $\rho(\cA_\kappa)''$ is a type $I$ factor.
 Then any factorial representation (in particular, irreducible) of the two-dimensional conformal net $\cA_+\otimes \cA_-$ is of the form $\rho_+\otimes\rho_-$
 for some representations $\rho_+, \rho_-$ of $\cA_+, \cA_-$, respectively.
\end{proposition}
\begin{proof}
 A representation $\rho$ of $\cA_+\otimes \cA_-$ can be restricted to $\cA_+ \otimes \bb1$ and $\bb1 \otimes \cA_-$, respectively,
 giving representations of $\cA_+, \cA_-$ as conformal nets on $S^1$ (see Section \ref{1d2dconformalnet}).
 For any pair $I_+, I_-$ such that $|I_\pm| < \pi$, so that $I_+ \cup I_- \subset I$ for some $I \in \cI$,
 we see that the images of $\cA_+(I_+)$ and $\cA_-(I_-)$ under the respective restrictions of $\rho$ commute.
 By additivity, $\rho(\cA_+(I_+) \otimes \bb1)$ and $\rho(\bb1 \otimes \cA_-(I_-))$ commute for arbitrary $I_\pm \in \cI$.

 The center $\rho(\cA_+\otimes \bb1)'' \cap \rho(\cA_+\otimes \bb1)'$
 is included in the center
 \begin{align*}
  \rho(\cA_+\otimes \cA_-)'' \cap \rho(\cA_+\otimes \cA_-)' = \rho(\cA_+\otimes \cA_-)'' \cap \rho(\cA_+\otimes \bb1)' \cap \rho(\bb1\otimes \cA_-)',
 \end{align*}
 as we saw that $\rho(\cA_+\otimes \bb1)'' \cap \rho(\cA_+\otimes \bb1)' \subset \rho(\cA_+\otimes \bb1)''
 \subset \rho(\bb1\otimes \cA_-)'$.
 As $\rho$ is factorial, both centers must be trivial. That is, the restriction of $\rho$ to $\cA_+$
 is factorial. The same holds for $\cA_-$.

 By assumption, $\cA_+, \cA_-$ are of type I. Thus $\rho(\cA_+\otimes\bb1)''$ and $\rho(\bb1\otimes\cA_-)''$ are type I
 factors and, as they commute, the representation $\rho$ is of the tensor product form.

\end{proof}

Let $\cA_{\frh_p} \subset \cB$ be a two-dimensional conformal extension, with $\cB$ not necessarily local.
Assume that the restriction $\tau_\cB$ of the vacuum representation of $\cB$ on the Hilbert space $\cH$ to the subnet $\cA_{\frh_p}$
is a \emph{direct sum of irreducible representations} of $\cA_{\frh_p}$.
Combining Proposition \ref{pr:tensorproduct2dsector} and \cite[Example 3.9, Proposition 3.11]{LT18} with the classification of 
irreducible representations of $\cA_\bbR$ in \cite[Section 2B, 3B]{BMT88}, each such irreducible representation is of the form $\sigma_{p\pmb{\lambda}h}\otimes \sigma_{\bar p\pmb{\lambda}h}$
up to unitary equivalence. Thus we may and do assume that
$\cH = \bigoplus_{\pmb{\lambda} \in Q} \cH_{p\pmb{\lambda}h}\otimes \cH_{\bar p\pmb{\lambda}h}$
as the representation $\tau_{\cB} = \bigoplus_{\pmb{\lambda} \in Q} \sigma_{p\pmb{\lambda}h}\otimes \sigma_{\bar p\pmb{\lambda}h}$
of $\cA_{\frh_p} = \cA_{p\frh}\otimes \cA_{\bar p\frh}$,
and $Q \subset \frh = \bbR^N$ is a collection of points $\pmb{\lambda} \in Q$, possibly counted with multiplicity, for which $\sigma_{p\pmb{\lambda}h}\otimes \sigma_{\bar p\pmb{\lambda}h}$ appears in the decomposition.
Note that $\pmb 0 \in Q$, as $\cA_{\frh_p}$ and $\cB$ have the same vacuum vector (cyclic for $\cA_{\frh_p}$ on the subspace $\cH_{\pmb{0}}\otimes \cH_{\pmb{0}}$).
We take $h \in \cC^\infty(S^1)$ as in Section \ref{sec:Heisenberg-net} in such a way that $\supp h \subset I$, for some fixed $I\in \cI$.

The following is an adaptation of parts of \cite[Proposition 3.1]{BMT88} to the two-dimensional case
(locality of $\cB$ is not used here, as in the original proof).
\begin{lemma}\label{lm:BMT3.1'}
 Let $\cA_{\frh_p} \subset \cB$ be as above, and assume that the set of points in $Q$ is discrete in the induced topology from $\bbR^N$.
 Then  for each $\pmb{\alpha} \in Q$ and a double cone $O = I_+\times I_-$,
 there is $\phi^{\pmb{\alpha}}\in \cB(O)$ such that $\phi^{\pmb{\alpha}}\Omega \in \cH_{p\pmb{\alpha}h}\otimes \cH_{\bar p\pmb{\alpha}h}$.
 Moreover, each $\pmb{\lambda} \in Q$ comes with multiplicity $1$ (\ie $\sigma_{p\pmb{\alpha}h}\otimes \sigma_{\bar p\pmb{\alpha}h}$ appears with multiplicity 1 in $\tau_\cB$),
 and $Q \subset \bbR^N$ is an additive subgroup.
\end{lemma}
\begin{proof}
 By taking a basis on $p\frh, \bar p\frh$ we identify $\frh$ with $\bbR^N$ and
 define the operators $\frq_j$ on $\cH$, $j = 1, \cdots, N$, that act by multiplication with $\lambda_j$ on each direct summand $\cH_{p\pmb{\lambda}h}\otimes \cH_{\bar p\pmb{\lambda}h}$ of $\cH$, $\pmb{\lambda} \in Q$.

 For $\pmb{s} \in \bbR^N$, let $\pmb{s}\cdot \frq := \sum_j s_j \frq_j$ and take the unitary $\rme^{\rmi \pmb{s}\cdot \frq}$ 
 on $\cH$.
 We claim that it holds that $\cAd \rme^{\rmi s\cdot \frq}(\cB(O)) = \cB(O)$ for any double cone $O = I_+\times I_-$.
 Indeed, let us pick a component $j$, the $j$-th canonical basis vector $\pmb{e}_j = (0,\cdots, 0, \underset{j\text{-th}}{1}, 0, \cdots, 0)$
 and $f_1, f_2 \in C^\infty(S^1)$ such that $f_1 + f_2 = 1$, $\supp f_i \subset I_i, i=1,2$
 and $\{I_1, I_2\}$ is a cover of $S^1$, and $s\in\bbR$.
 
 In the vacuum direct summand $\cH_{p\pmb{0}h}\otimes \cH_{\bar p\pmb{0}h}$, recalling the notation from Section \ref{sec:Heisenberg-net}, it holds that
 $\rme^{\frac{1}{4\pi\rmi}\int s f_1'(t) sf_2(t) dt}\rme^{\rmi\pmb{e}_j(sf_1)}\rme^{\rmi\pmb{e}_j(sf_2)} = \rme^{\rmi\pmb{e}_j(s)} = \bb1$,
 where $s$ is the constant function, hence $\pmb{e}_j(s) = J_0$ which is the zero operator on $\cH_{p\pmb{0}h}\otimes \cH_{\bar p\pmb{0}h}$.
 Therefore, in the representation (automorphism) denoted for simplicity $\sigma_{\pmb{\lambda}h} := \sigma_{p\pmb{\lambda}h}\otimes \sigma_{\bar p\pmb{\lambda}h}$, for $s \in \bbR$, it holds that
 \begin{align*}
  &\rme^{\frac{1}{4\pi\rmi}\int sf_1'(t) sf_2(t)dt}\sigma_{\pmb{\lambda} h}(\rme^{\rmi\pmb{e}_j(sf_1)})\sigma_{\pmb{\lambda} h}(\rme^{\rmi\pmb{e}_j(sf_2)}) \\
  &= \rme^{\frac{1}{4\pi\rmi}\int sf_1'(t) sf_2(t)dt}\rme^{\frac{\rmi}{2\pi}\int sf_1(t) \lambda_j h(t)dt}\rme^{\rmi\pmb{e}_j(sf_1)}
  \rme^{\frac{\rmi}{2\pi}\int sf_2(t) \lambda_j h(t)dt}\rme^{\rmi\pmb{e}_j(sf_2)} \\
  &= \rme^{\frac{\rmi s \lambda_j}{2\pi}\int (f_1(t) + f_2(t))h(t)dt}\bb1 \\
  &=  \sigma_{\pmb{\lambda} h}(\rme^{\rmi s\lambda_j}\bb1).
 \end{align*}
 This shows that, on $\cH$, for $s_j \in \bbR$ it holds that
 \begin{align*}
  \rme^{\frac{1}{4\pi\rmi}\int s_j f_1'(t) s_jf_2(t)dt}\tau_{\cB}(\rme^{\rmi\pmb{e}_j(s_jf_1)})\tau_{\cB}(\rme^{\rmi\pmb{e}_j(s_jf_2)})
  &= \rme^{\rmi s_j \frq_j}.
 \end{align*}
 Note that this holds for any decomposition $1 = f_1 + f_2$ under the above conditions.

 As $\cB$ is an extension of $\cA_{\frh_p}$, for any $\pmb{s} \in \bbR^N$,
 it holds that $\cAd \rme^{\rmi \pmb{s}\cdot \frq}(\cB(I_+\times I_-)) \subset \cB(I_+ \times I_-)$.
 Indeed, we may assume that $\pmb{s} = s\pmb{e}_j$ for some $j$ in the  first $\dim p\frh$ entries (the other case is similar).
 Then we can write $e^{\rmi \pmb{s}\cdot \frq}$, where then $\pmb{s}\cdot \frq = s \frq_j$, as a product of two Weyl operators as above.
 We can take $f_1, f_2$ in such a way that $\supp f_2 \subset I_+'$,
 and $f_1 + f_2 = 1$ on $S^1$. By Lemma \ref{lm:relativelocal},
 $\cAd\rme^{\rmi \pmb{s}\cdot \frq}(\cB(I_+\times I_-)) = \cAd \tau_{\cB}(\rme^{\rmi s\pmb{e}_j(f_1)})(\cB(I_+\times I_-)) \subset \cB(\supp f_1 \times I_-)$.
 As the decomposition $f_1 + f_2 = 1$ is arbitrary under the condition that $\supp f_2 \subset I_+'$,
 $\supp f_1$ can be any closed interval containing $I_+$. In particular,
 $\cAd\rme^{\rmi \pmb{s}\cdot \frq}(\cB(I_+\times I_-)) = \cB(I_+ \times I_-)$.
 
 We claim that the set $Q$ forms an additive subgroup in $\bbR^N$. Let us fix $O = I_+\times I_-$.
 Now we have the situation where the locally compact group $\bbR^N$ has a unitary representation $V(\pmb{s}) := \rme^{\rmi \pmb{s}\cdot \frq}$,
 $\cAd V(\pmb{s})$ is an automorphism of $\cB(O)$, $V(\pmb{s})\Omega = \Omega$ and $\Omega$ is cyclic and separating for $\cB(O)$.
 As in \cite[Section 1.8]{Baumgaertel95}, we can define the Arveson spectrum of $\cAd V$, 
 and under these conditions,
 it coincides with the joint spectrum of $\frq_1, \ldots, \frq_N$. By assumption, the joint spectrum is discrete,
 and we can define the Arveson spectrum $\cB(O, \cAd V, \pmb{\lambda})$ of $\cAd V$ for one point $\pmb{\lambda}$,
 and by (an analogue of) \cite[Theorem 1.8.4]{Baumgaertel95}, we have
 \begin{align*}
  \cB(O, \cAd V, \pmb{\lambda}) = \{x \in \cB(O), x\Omega \in \cH_{\pmb \lambda h}\},
 \end{align*}
 where we denote for simplicity $\cH_{\pmb \lambda h} := \cH_{p\pmb{\lambda}h} \otimes \cH_{\bar p\pmb{\lambda}h}$.
 In particular, for each $\pmb{\alpha} \in Q$, we can choose $\phi^{\pmb{\alpha}} \in \cB(O, \cAd V, \pmb{\alpha})$, $\phi^{\pmb{\alpha}} \neq 0$, thus $\phi^{\pmb{\alpha}} \Omega \in \cH_{\pmb{\alpha}h}$.

 We claim that $x = \phi^{\pmb{\alpha}*}\phi^{\pmb{\alpha}}$ belongs to $\tau_{\cB}(\cA_{\frh_p}(O))$, $O = I_+\times I_-$.
 Indeed, it is invariant under $\cAd V(\pmb{s})$, and it can be restricted to $\cH_{\pmb 0 h}$.
 The restriction $x|_{\cH_{\pmb 0 h}}$ commutes with $\cA_{\frh_p}(I_+'\times I_-')$, thus by Haag duality on $\cH_{\pmb 0 h}$,
 there is an element $y$ of $\cA_{\frh_p}(I_+\times I_-)$ which coincides with the restriction of $x|_{\cH_{\pmb 0 h}}$
 to $\cH_{\pmb 0 h}$. By the separating property of $\Omega$ for $\cB(I_+\times I_-)$,
 $x = \tau_{\cB}(y)$, 
 proving the claim.
 
 Next, we claim that we can take $\phi^{\pmb{\alpha}}$ an isometry.
 Let $\phi^{\pmb{\alpha}} = v^{\pmb{\alpha}}(\phi^{\pmb{\alpha}*}\phi^{\pmb{\alpha}})^\frac12$ be the polar decomposition.
 We may assume that the phase $v^{\pmb{\alpha}}$ 
 is an isometry.
 Indeed, if it is only a partial isometry,
 then $v^{\pmb{\alpha}*} v^{\pmb{\alpha}}$ is a projection and it belongs to $\tau_{\cB}(\cA_{\frh_p}(I_+\times I_-))$ by the same argument as above.
 As $\cA_{\frh_p}(I_+\times I_-)$ is a type III factor, 
 we can find an isometry $u \in \cA_{\frh_p}(I_+\times I_-)$
 such that $\tau_{\cB}(uu^*) = v_{\pmb{\alpha}}^* v_{\pmb{\alpha}}$,
 and $v_{\pmb{\alpha}} \tau_{\cB}(u) \in \cB(O)$ satisfies the desired property.

 We show that $Q$ is an additive subgroup of $\bbR^N$.
 For this purpose, let $\pmb{\alpha}, \pmb{\beta} \in Q$ and take $\phi^{\pmb{\alpha}}, \phi^{\pmb{\beta}} \in \cB(O)$ isometries as above.
 As $\phi^{\pmb{\beta}*}\phi^{\pmb{\alpha}*}\phi^{\pmb{\alpha}}\phi^{\pmb{\beta}} = \bb1$,
 we have $\phi^{\pmb{\alpha}}\phi^{\pmb{\beta}} \Omega \neq 0$,
 and it is in $\cH_{(\pmb{\alpha} + \pmb{\beta})h}$.
 Next, $\Omega$ is separating for $\cB(I)$, $\phi^{\pmb{\alpha}*}\Omega \neq 0$,
 and it is in $\cH_{-\pmb{\alpha}h}$. Clearly, $\pmb 0 \in Q$.
 This shows that $Q$ is an additive subgroup of $\bbR^N$.

 Finally, we show that
 $\pmb{\lambda}$ has multiplicity $1$.
 To see this, let $\Psi, \Phi$, with $\Phi \neq 0$, be vectors that belong to $\cH_{\pmb{\lambda} h} \oplus \ldots \oplus \cH_{\pmb{\lambda} h}$,
 where the direct sum is taken over a single value $\pmb{\lambda} \in Q$ possibly with multiplicity
 (that is, $\pmb{\lambda} \in Q$ may be counted with multiplicity). Assume that
 $\<\Psi, \tilde\sigma_{\pmb{\lambda}h}(x)\Phi\> = 0$ for all $x \in \cA_{\frh_p}(I_+\times I_-)$ and $I_+,I_-\in\cI$,
 where $\tilde\sigma_{\pmb{\lambda}h}$ is the direct sum of copies of $\sigma_{\pmb{\lambda}h}$ over the single value $\pmb{\lambda}$ repeated in $Q$
 in the sense above.
 Then $\<\Psi, y\Phi\> = 0$ holds for all $y \in \overline{\bigcup_{I_\pm \in \cI}\tau_{\cB}(\cA_{\frh_p}(I_+\times I_-))}$,
 where the closure is taken in the weak operator topology on $\cH$.

 Let $\phi_1, \phi_2 \in \cB(O)$ with charge $\pmb{\lambda}$ as in the previous paragraph.
 As lightlike translations $U(T_t\times \iota)$ belong to $\overline{\bigcup_{I_\pm \in \cI}\tau_{\cB}(\cA_{\frh_p}(I_+\times I_-))}$
 as $\cB$ is a conformal extension of $\cA_{\frh_p}$ (see the beginning of this Section),
 and
 \begin{align*}
  \phi_1 U(T_t\times \iota)\phi_2^* = U(T_t\times \iota)\cdot U(T_t\times \iota)^*\phi_1 U(T_t\times \iota)\phi_2^*
  \in \overline{\bigcup_{I_\pm \in \cI}\tau_\cB(\cA_{\frh_p}(I_+\times I_-))},
 \end{align*}
 thus $\<\Psi, \phi_1 U(T_t\times \iota)\phi_2^*\Phi\> = 0$ for all $t$.
 By positivity of energy, this extends to the upper half plane, and then by Schwarz reflection principle, it is $0$ for all $t \in \bbR$.
 Apply this also for $U(\iota\times T_t)$, and by the uniqueness of $\Omega$, it must hold that
 $\<\Psi, \phi_1\Omega\>\<\Omega\phi_2^*, \Phi\> = 0$.
 As $\Omega$ is cyclic for $\cB(O)$ and $\phi_1, \phi_2$ were arbitrary operators with charge $\pmb{\lambda}$,
 we conclude that $\Psi = 0$.
\end{proof}

By Lemma \ref{lm:BMT3.1'}, $Q$ is an additive subgroup of $\bbR^N$.
If we assume that it is discrete, then by \cite[(4.2) Proposition]{Neukirch99AlgebraicNumberTheory}, it is a (not necessarily $N$-dimensional) lattice in $\frh$.
Recall that we defined simple (not twisted by the 2-cocycle) shift operators $\underline \psi^{\pmb{\alpha}}$ for any lattice $Q$ in \eqref{eq:simple-shift}.
For any $x \in \cA_{\frh_p}(I_+\times I_-)$, it holds that
$\underline \psi^{\pmb{\alpha}} \cdot \tau_{\cB} (x) \cdot \underline \psi^{\pmb{\alpha}*}
= \tau_{\cB}(\sigma_{-\pmb{\alpha}h}(x)) \in \tau_{\cB}(\cA_{\frh_p}(I_+\times I_-))$.

\begin{lemma}\label{lm:BMT3.1''}
 Let $Q$ be a lattice in $\bbR^N$ and let $\cB$ be a (local) conformal extensions of $\tau_\cB(\cA_{\frh_p})$ on $\cH_Q$ defined in \eqref{eq:H_Q},
 where $\tau_{\cB} = \bigoplus_{\pmb{\lambda} \in Q} \sigma_{p\pmb{\lambda}h}\otimes \sigma_{\bar p\pmb{\lambda}h}$ and $h\in\cC^\infty(S^1)$ has $\supp h \subset I$ as before.
 Then for each $\pmb{\alpha} \in Q$, there is $\psi^{\pmb{\alpha}}_\cB \in \cB(I\times I)$ such that
 $\psi^{\pmb{\alpha}*}_\cB \underline \psi^{\pmb{\alpha}} = \bigoplus_{\pmb{\lambda} \in Q} \epsilon_{\cB}(\pmb{\alpha}, \pmb{\lambda})$
 for some $\epsilon_{\cB}(\pmb{\alpha}, \pmb{\lambda}) \in S^1$.
 Moreover, $Q$ is an even lattice with respect to $(\cdot|\cdot)$.
\end{lemma}
\begin{proof}
 We take an isometry $\phi^{\pmb{\alpha}} \in \cB(I\times I, \cAd V, \pmb \lambda)$ as in Lemma \ref{lm:BMT3.1'},
 and we use the same notation introduced in its proof for $\sigma_{\pmb{\lambda}h}$ and $\cH_{\pmb{\lambda}h}$.
 By assumption, $\underline \psi^{\pmb{\alpha}} \phi^{\pmb{\alpha}*}$ preserves each $\cH_{\pmb{\lambda}h}$
 and commutes with $\sigma_{\pmb{\lambda}h}(\cA_{\frh_p}(I' \times I')))$,
 therefore, as each $\sigma_{\pmb{\lambda}h}$ is an automorphism, by Haag duality of $\cA_{\frh_p}$,
 it holds that $\underline \psi^{\pmb{\alpha}} \phi^{\pmb{\alpha}*} = \bigoplus_{\pmb{\lambda} \in Q} \sigma_{\pmb{\lambda}h}(W_{\pmb{\lambda}})$
 for some $W_{\pmb{\lambda}} \in \cA_{\frh_p}(I\times I)$ depending on $\pmb{\lambda} \in Q$.

 Let us consider $\pmb{0} \in Q$ and set $\psi^{\pmb{\alpha}}_{\cB} := \tau_{\cB}(W_{\pmb{0}})\phi^{\pmb{\alpha}}$.
 Then, denoted by $P_{\pmb{0}}$ the orthogonal projection onto $\cH_{\pmb{0}h}$, it holds that
 \begin{align*}
  P_{\pmb{0}}\psi^{\pmb{\alpha}}_{\cB}
  &= P_{\pmb{0}} \sigma_{\pmb{0}}(W_{\pmb{0}})\phi^{\pmb{\alpha}} \\
  &= P_{\pmb{0}} \underline \psi^{\pmb{\alpha}}
 \end{align*}
 therefore, for any $x \in \cA_{\frh_p}$ we have
 \begin{align*}
  P_{\pmb{0}}(\psi^{\pmb{\alpha}}_{\cB} \cdot \tau_{\cB}(x) \cdot \psi^{\pmb{\alpha}*}_{\cB})
  &= P_{\pmb{0}}(\underline \psi^{\pmb{\alpha}} \cdot \tau_{\cB}(x) \cdot \underline \psi^{\pmb{\alpha}*}).
 \end{align*}
 Note that, by taking $x \in \cA_{\frh_p}(\tilde O)$, where $\overline{I\times I} \subset \tilde O$,
 we have
 $\psi^{\pmb{\alpha}}_{\cB} \cdot \tau_{\cB}(x) \cdot \psi^{\pmb{\alpha}*}_{\cB},
 \underline \psi^{\pmb{\alpha}} \cdot \tau_{\cB}(x) \cdot \underline \psi^{\pmb{\alpha}*} \in \cA_{\frh_p}(\tilde O)$.
 By
 the separating property of the vacuum, we have that
 $\psi^{\pmb{\alpha}}_{\cB} \cdot \tau_{\cB}(x) \cdot \psi^{\pmb{\alpha}*}_{\cB}
 = \underline \psi^{\pmb{\alpha}} \cdot \tau_{\cB}(x) \cdot \underline \psi^{\pmb{\alpha}*}$,
 or equivalently,
 \begin{align*}
  \tau_{\cB}(x) \cdot \psi^{\pmb{\alpha}*}_{\cB}\underline \psi^{\pmb{\alpha}}
   = \psi^{\pmb{\alpha}*}_{\cB}\underline \psi^{\pmb{\alpha}} \cdot \tau_{\cB}(x).
 \end{align*}

 Moreover, if $x' \in \cA_{\frh_p}(I'\times I')$, then by relative locality and the fact that
 $\tau_{\cB}$ is localized in $I\times I$ (by the choice of localization of the function $h$), we have
 \begin{align*}
  \tau_{\cB}(x') \cdot \psi^{\pmb{\alpha}*}_{\cB}\underline \psi^{\pmb{\alpha}}
   = \psi^{\pmb{\alpha}*}_{\cB}\underline \psi^{\pmb{\alpha}} \cdot \tau_{\cB}(x').
 \end{align*}
 As $\cA_{\frh_p}(\tilde O)$ and $\cA_{\frh_p}(I'\times I')$ generate the whole net $\cA_{\frh_p}$,
 this shows that
 $\psi^{\pmb{\alpha}*}_{\cB}\underline \psi^{\pmb{\alpha}} = \bigoplus_{\pmb{\lambda} \in Q} \epsilon_{\cB}(\pmb{\alpha}, \pmb{\lambda})$
 for some $\epsilon_{\cB}(\pmb{\alpha}, \pmb{\lambda}) \in \bbC$.
 As $\psi^{\pmb{\alpha}}_{\cB}$ is an isometry and $\underline \psi^{\pmb{\alpha}}$ is unitary, we have $|\epsilon_{\cB}(\pmb{\alpha}, \pmb{\lambda})| = 1$,
 thus $\epsilon_{\cB}(\pmb{\alpha}, \pmb{\lambda}) \in S^1$.

 Recall that $\cB$ is local.
 Then, for any $\gamma_+, \gamma_- \in \uDiffSone\times\uDiffSone$ that map $I\times I$ to $I_+\times I_-$,
 where $I_+$ is in the past of $I$ and $I_-$ is in the future of $I$, $\cAd U(\gamma_+\times \gamma_-)(\psi^{\pmb{\alpha}}_{\cB})$ and $\psi^{\pmb{\alpha}}_{\cB}$ must commute. 
 On the other hand, as each $\epsilon_{\cB}(\pmb{\alpha}, \pmb{\lambda})$ clearly commutes with $U(\gamma_+\times \gamma_-)|_{\cH_{\pmb{\lambda}h}}$ for every $\pmb \lambda \in Q$,
 arguing as for \eqref{eq:field-cov}, we get
 \begin{align}
  \psi^{\pmb{\alpha}}_{\cB} \cdot \cAd U_Q(\gamma_+\times \gamma_-)(\psi^{\pmb{\alpha}}_{\cB})
  &= \tau_{\cB}\left(\sigma_{p\pmb{\alpha}h}((z_{p\pmb{\alpha}h}(\gamma_+))\otimes \sigma_{\bar p\pmb{\alpha}h}(z_{\bar p\pmb{\alpha}h}(\gamma_-))\right) \psi^{\pmb{\alpha}}_{\cB}\psi^{\pmb{\alpha}}_{\cB} \nonumber \\
  \cAd U_Q(\gamma_+\times \gamma_-)(\psi^{\pmb{\alpha}}_{\cB}) \cdot \psi^{\pmb{\alpha}}_{\cB}
  &= \tau_{\cB}\left(z_{p\pmb{\alpha}h}(\gamma_+)\otimes z_{\bar p\pmb{\alpha}h}(\gamma_-)\right) \psi^{\pmb{\alpha}}_{\cB} \psi^{\pmb{\alpha}}_{\cB}
  \label{eq:braiding-cB}
 \end{align}
 By the definition of braiding between the automorphisms of $\cA_{\frh_p}$, see \eqref{eq:braiding-z} and \eqref{eq:braidingNdim},
 \begin{align*}
  \tau_{\cB}(z_{p\pmb{\alpha}h}(\gamma_+)\otimes z_{\bar p\pmb{\alpha}h}(\gamma_-))^*
  \tau_{\cB}(\sigma_{p\pmb{\alpha}h}((z_{p\pmb{\alpha}h}(\gamma_+))\otimes \sigma_{\bar p\pmb{\alpha}h}(z_{\bar p\pmb{\alpha}h}(\gamma_-)))
  = \varepsilon_{p\pmb{\alpha}h, p\pmb{\alpha}h}^+ \varepsilon_{\bar p\pmb{\alpha}h, \bar p\pmb{\alpha}h}^- = \rme^{\rmi \pi (\alpha|\alpha)}.
 \end{align*}
 Thus, it must hold that $(\pmb{\alpha}|\pmb{\alpha}) \in 2\bbZ$ to have locality of $\cB$.

 For any pair $\pmb{\alpha}, \pmb{\beta} \in Q$, it holds that
 $(\pmb{\alpha}+\pmb{\beta}|\pmb{\alpha}+\pmb{\beta}) = (\pmb{\alpha}|\pmb{\alpha}) + 2(\pmb{\alpha}|\pmb{\beta}) + (\pmb{\beta}|\pmb{\beta}) \in 2\bbZ$,
 hence $(\pmb{\alpha}|\pmb{\beta}) \in \bbZ$. That is, $Q$ is an even lattice with respect to $(\cdot|\cdot)$.
\end{proof}

\begin{theorem}\label{th:classification}
 Let $\cA_{p\frh}\otimes\cA_{\bar p\frh} \subset \cB$ be a local conformal extension and assume that
 $\cA_{p\frh}\otimes\cA_{\bar p\frh}$ on $\cH$ is a direct sum
 of $\cH_{p\pmb{\lambda}h}\otimes \cH_{\bar p\pmb{\lambda}h}$, for $\pmb{\lambda} \in Q \subset \bbR^N$ possibly counted with multiplicity and corresponding to the representations $\sigma_{p\pmb{\lambda}h}\otimes\sigma_{\bar p\pmb{\lambda}h}$.
 Assume that  the set $Q$ is discrete in $\bbR^N$.
 Then each $\pmb{\lambda}$ appears with multiplicity 1, $Q$ is a (not necessarily $N$-dimensional) even lattice in $\bbR^N$ with respect to $(\cdot|\cdot)$ defined in Section \ref{sec:2d-twist}, and $\cB$ is unitarily equivalent to the two-dimensional lattice conformal net $\cA_Q$ constructed in Section \ref{2dlatticenet}.
\end{theorem}
\begin{proof}
 By assumption $\tau_{\cB} = \bigoplus_{\pmb{\lambda} \in Q} \sigma_{\pmb{\lambda}h}$ and $Q$ (disregarding multiplicities) is a discrete subset of $\bbR^N$ containing $\pmb 0$. By Lemma \ref{lm:BMT3.1'} and Lemma \ref{lm:BMT3.1''}, $Q$ is an even lattice and the multiplicity of each $\pmb \lambda \in Q$ in the decomposition of $\tau_\cB$ is 1. Recall the short-hand notation introduced earlier for $\sigma_{\pmb{\lambda}h}$ and $\cH_{\pmb{\lambda}h}$.
 To show that $\cB$ is unitarily equivalent to $\cA_Q$,
 for each $\pmb{\alpha} \in Q$ we take $\psi^{\pmb{\alpha}}_{\cB}$ as in Lemma \ref{lm:BMT3.1''},
 such that $\psi^{\pmb{\alpha}*}_\cB \underline \psi^{\pmb{\alpha}} = \bigoplus_{\pmb{\lambda} \in Q} \epsilon_{\cB}(\pmb{\alpha}, \pmb{\lambda})$
 and $\epsilon_{\cB}(\pmb{\alpha}, \pmb{\lambda}) \in S^1$.

 From this it follows that,
 for any $x \in \cA_{\frh_p}(O)$, $O$ a double cone in $\bbR^{1+1}$, we have
 $\psi^{\pmb{\alpha}}_\cB \tau_{\cB} (x) \psi^{\pmb{\alpha}*}_\cB
 = \tau_{\cB}(\sigma_{-\pmb{\alpha}h}(x)) \in \tau_{\cB}(\cA_{\frh_p}(O))$,
 just as with $\underline \psi^{\pmb{\alpha}}$.
 Then it follows that $\psi^{\pmb{\alpha}}_\cB \psi^{\pmb{\beta}}_\cB = Z(\pmb{\alpha}, \pmb{\beta})\psi^{\pmb{\alpha} + \pmb{\beta}}_\cB$,
 and \textit{a priori} $Z(\pmb{\alpha}, \pmb{\beta}) = \bigoplus_{\pmb{\gamma} \in Q} Z(\pmb{\alpha}, \pmb{\beta}, \pmb{\gamma})$
 with $Z(\pmb{\alpha}, \pmb{\beta}, \pmb{\gamma}) \in \bbC$, but by the separating property of the vacuum,
 actually $Z(\pmb{\alpha}, \pmb{\beta}, \pmb{\gamma})$ does not depend on $\pmb{\gamma}$.
 Moreover, $Z(\pmb{\alpha}, \pmb{\beta})$ satisfies the 2-cocycle equation.
 By repeating the calculation of \eqref{eq:braiding-cB} for possibly different $\pmb{\alpha}, \pmb{\beta}$,
 and imposing locality, we obtain
 $Z(\pmb{\alpha}, \pmb{\beta})Z(\pmb{\beta}, \pmb{\alpha})^{-1} = \rme^{\rmi(\pmb{\alpha}|\pmb{\beta})}$.

 Recall that, on the same Hilbert space $\cH_Q$, we defined in \eqref{eq:shift-field} the twisted shift operators $\psi^{\pmb{\alpha}}$, $\pmb \alpha \in Q$,
 which satisfy $\psi^{\pmb{\alpha}}\psi^{\pmb{\beta}} = \epsilon(\pmb{\alpha}, \pmb{\beta})\psi^{\pmb{\alpha} + \pmb{\beta}}$, and we also have
 $\epsilon(\pmb{\alpha}, \pmb{\beta})\epsilon(\pmb{\beta}, \pmb{\alpha})^{-1} = \rme^{\rmi(\pmb{\alpha}|\pmb{\beta})}$,
 see below \eqref{eq:twococycle}.
 
 The quotient of 2-cocycles is again a 2-cocycle, and $\epsilon(\pmb{\alpha}, \pmb{\beta})Z(\pmb{\alpha}, \pmb{\beta})^{-1}$
 satisfies
 \begin{align*}
  \epsilon(\pmb{\alpha}, \pmb{\beta})Z(\pmb{\alpha}, \pmb{\beta})^{-1}(\epsilon(\pmb{\beta}, \pmb{\alpha})Z(\pmb{\beta}, \pmb{\alpha})^{-1})^{-1}
  = 1.
 \end{align*}
 By \cite[Lemma 3.4.2]{Baumgaertel95}, it is 2-coboundary:
 $\epsilon(\pmb{\alpha}, \pmb{\beta})Z(\pmb{\alpha}, \pmb{\beta})^{-1} = \chi(\pmb{\alpha})\chi(\pmb{\beta})\chi(\pmb{\alpha}+\pmb{\beta})^{-1}$ for some function $\chi: Q\to S^1$.
 This means that $\psi^{\prime\pmb{\alpha}}_{\cB} := \chi(\pmb{\alpha})\psi^{\pmb{\alpha}}_{\cB} \in \cB(I\times I)$ satisfy
 \begin{align}\label{eq:prime-cocycle}
  \psi^{\prime\pmb{\alpha}}_{\cB}\psi^{\prime\pmb{\beta}}_{\cB} = \epsilon(\pmb{\alpha}, \pmb{\beta})\psi^{\prime\pmb{\alpha}+\pmb{\beta}}_{\cB}.
 \end{align}
 Note that $\psi^{\prime\pmb{0}}_{\cB} = \bb1 = \psi^{\pmb{0}}$.
 Furthermore, by \eqref{eq:prime-cocycle},
 \begin{align}\label{eq:psi-recursive}
  \psi^{\prime\pmb{\alpha}}_{\cB} = \epsilon(\pmb{\alpha}, \pmb{\beta})\psi^{\prime\pmb{\alpha}+\pmb{\lambda}}_{\cB}(\psi^{\prime\pmb{\lambda}}_{\cB})^*,
 \end{align}
 therefore,
 $\psi^{\prime\pmb{\alpha}}_{\cB}$ is completely determined by
 $\{\psi^{\prime\pmb{\lambda}}_{\cB}|_{\cH_{\pmb{0}h}}\}_{\pmb{\lambda} \in Q}$ and $\epsilon$.
 
 Now we can write $\psi^{\prime\pmb{\alpha}}_{\cB} = \bigoplus_{\pmb{\lambda} \in Q} \epsilon'_{\cB}(\pmb{\alpha}, \pmb{\lambda}) \psi^{\pmb{\alpha}}$
 for some $\epsilon'_{\cB}(\pmb{\alpha} \in \bbC$.
 It follows that  $\epsilon'_{\cB}(\pmb{0}, \pmb{0}) = 1$.
 Therefore, with the unitary $V := \bigoplus_{\pmb{\lambda}} \epsilon'_{\cB}(\pmb{\lambda}, \pmb{\lambda})$, it holds that
 \begin{align*}
  \cAd V(\psi^{\pmb{\lambda}})|_{\cH_{\pmb{0}h}}
  = \epsilon'_{\cB}(\pmb{\lambda}, \pmb{\lambda})\overline{\epsilon'_{\cB}(\pmb{0}, \pmb{0})}\psi^{\pmb{\lambda}}|_{\cH_{\pmb{0}h}}
  = \epsilon'_{\cB}(\pmb{\lambda}, \pmb{\lambda})\psi^{\pmb{\lambda}}|_{\cH_{\pmb{0}h}}
  = \psi^{\prime\pmb{\lambda}}_{\cB}|_{\cH_{\pmb{0}h}}
 \end{align*}
 As $\cAd V(\psi^{\pmb{\alpha}})|_{\cH_{\pmb{0}h}}$ has the same cocycle $\epsilon$,
 it follows from the remark below \eqref{eq:psi-recursive} that $\cAd V(\psi^{\pmb{\alpha}}) = \psi^{\prime\pmb{\alpha}}_{\cB}$.
 As $\cAd V$ commutes with $\tau_\cB(x) = \tau_Q(x)$ for any $x \in \cA_{\frh_p}(O)$,
 we conclude that two extensions $\cA_{\frh_p}\subset \cB$ and $\cA_{\frh_p}\subset \cA_Q$ are unitarily equivalent.
\end{proof}

\section{Two-dimensional Wightman fields}\label{wightman}
Let $\frh, p, Q, \epsilon$ as in Section \ref{sec:2d}.
Using the charge structures studied for conformal nets, we construct Wightman fields on the Hilbert space $\cH_Q$.
For this purpose, it is more convenient to take
\begin{align}\label{eq:K_Q}
 \cK_Q := \bigoplus_{\pmb{\lambda} \in Q} \cK_{p\pmb{\lambda}} \otimes \cK_{\bar p\pmb{\lambda}},
\end{align}
where $\cK_{p\pmb{\lambda}}, \cK_{\bar p\pmb{\lambda}}$
are the completion of $M_{p\frh}(1, p\pmb{\lambda}), M_{\bar p\frh}(1, \bar p\pmb{\lambda})$
including the actions of the operators $p\pmb{\alpha}(m), \bar p\pmb{\alpha}(m)$,
thus $\cK_{p\pmb{\lambda}} = \cH_{\frac{p\pmb{\lambda}}{2\pi}}$ (the constant function $\frac{p\pmb{\lambda}}{2\pi}$), see Section \ref{sec:Heisenberg-net}.
As we remarked there, 
$\cH_{p\pmb{\lambda}h} \cong \cH_{\frac{p\pmb{\lambda}}{2\pi}} = \cK_{p\pmb{\lambda}}$
as representations of $\cA_{p\frh}$, and similar for $\cA_{\bar p\frh}$.
In this sense, there is a natural unitary equivalence
\begin{align*}
 \cH_Q = \bigoplus_{\pmb{\lambda} \in Q} \cH_{p\pmb{\lambda}} \otimes \cH_{\bar p\pmb{\lambda}}
 \cong \bigoplus_{\pmb{\lambda} \in Q} \cK_{p\pmb{\lambda}} \otimes \cK_{\bar p\pmb{\lambda}} =\cK_Q
\end{align*}
intertwining the actions of $\cA_{\frh_p}$.

We use the following convention of formal series of vectors or operators, \cf \cite[Section 5.2]{AGT23Pointed}, \cite[Section 2.1]{AMT24OS}.
Let $z, \bar z$ be two independent formal variables\footnote{$\bar z$ is not the complex conjugate of $z$, although they
are natural notations when one considers correlation functions, see \cite{AMT24OS}.}.
For a family $\{A_{r,s}\}_{{r,s}\in\bbR}$ of vectors in $W$ (we often take $W = \End(V)$ for some vector space $V$),
we write $A(\underline z) = \sum_{r,s} A_{s,r} z^{-r-1} \bar z^{-s-1}$.
If $A_{r,s} = 0$ whenever $r-s \notin \bbZ$, we write $\sum_{r,s} A_{r,s} z^{-r-1} \bar z^{-s-1}$.
The elements $A_{r,s}$ are referred to as \textbf{Fourier components}.

\subsection{Vertex operators}\label{sec:vertex}

On each $M_{\frh}(1, \pmb{\lambda})$ of Section \ref{sec:Heisenberg-algebra},
we use the same symbol for the operators $\pmb \alpha(n)$, irrespectively of $\pmb \lambda$.
We define formal series (with a slight abuse of notations ``$\pmb{\alpha}(\cdot)$'')
\begin{align*}
 \pmb{\alpha}(z) &:= \sum_{n \in \bbZ} \pmb{\alpha}(n)z^{-n-1}, \\
 E^+(\pmb{\alpha}, z) &:= \exp\left(-\sum_{n \in \bbZ, n \ge 1} \frac{\pmb{\alpha}(n)}{n}z^{-n} \right), \\
 E^-(\pmb{\alpha}, z) &:= \exp\left(\sum_{n \in \bbZ, n \ge 1} \frac{\pmb{\alpha}(-n)}{n}z^{n} \right).
\end{align*}

By \cite[Chapter VI (1.2.2)]{Toledano-LaredoThesis}, it holds that
 \begin{align}\label{eq:comm-E}
  \left(1-\frac{z}w\right)^{-(\pmb{\alpha}, \pmb{\beta})_{\frh}}E^+(\pmb{\alpha}, w)E^-(\pmb{\beta}, z)
  &= E^-(\pmb{\beta}, z)E^+(\pmb{\alpha}, w),
 \end{align}
where $(1-u)^a = \sum_{n\in\bbZ, n \ge 0}\binom{a}{n}(-u)^n$ by definition, giving only integer powers of $u$.

Let $\widehat{\frh_p}$ be the infinite-dimensional Lie algebra generated by
$p\pmb{\alpha}(m), \bar p \pmb{\alpha}(n)$ and a central element $K$, where $\pmb{\alpha} \in \frh$, with the following commutation relations
\begin{align*}
 [p\pmb{\alpha}_1(m), p\pmb{\alpha}_2(n)] &= \delta_{m,-n}m\cdot (p\pmb{\alpha}_1,p\pmb{\alpha}_2)_\frh K, \\
 [\bar p\pmb{\alpha}_1(m), \bar p\pmb{\alpha}_2(n)] &= \delta_{m,-n}m\cdot (\bar p\pmb{\alpha}_1,\bar p\pmb{\alpha}_2)_\frh K, \\
 [p\pmb{\alpha}_1(m), \bar p\pmb{\alpha}_2(n)] &= 0.
\end{align*}

There are modules of this algebra parametrized by $\pmb{\lambda} = p\pmb{\lambda}\oplus \bar p\pmb{\lambda}$, where $\pmb{\lambda} \in \frh$ again.
Indeed, for each $\pmb{\lambda} \in \frh$, we denote $M_{\frh_p}(1, \pmb{\lambda}) = U(\frh_p)/J_{\pmb{\lambda}}$,
where $J_{\pmb{\lambda}}$ is the ideal generated by $p\pmb{\alpha}(m), \bar p\pmb{\alpha}(n), m, n > 0, p\pmb{\alpha}(0)-(p\pmb{\lambda}, p\pmb{\alpha})_\frh,
\bar p\pmb{\alpha}(0) - (\bar p\pmb{\lambda}, \bar p\pmb{\alpha})_\frh, K-1$.
As $\frh_p = p\frh \oplus \bar p\frh$, there is a natural isomorphism
$M_{\frh_p}(1, \pmb{\lambda}) \cong M_{p\frh}(1, p\pmb{\lambda}) \otimes M_{\bar p\frh}(1, \bar p\pmb{\lambda})$
see Section \ref{sec:Heisenberg-algebra}.

On each $M_{\frh_p}(1, \pmb{\lambda})$, we define formal series with two variables $z, \bar z$ by
\begin{align}
 \pmb{\alpha}(\underline z) &= \sum_{z \in \bbZ} (p\pmb{\alpha}(n)z^{-n-1} + \bar p \pmb{\alpha}(n)\bar z^{-n-1}), \nonumber \\
 E^+(\pmb{\alpha}, \underline z) &= \exp\left(-\sum_{z \ge 1} \left(\frac{p\pmb{\alpha}(n)}{n}z^{-n} + \frac{\bar p \pmb{\alpha}(n)}{n}\bar z^{-n}\right)\right), \label{eq:def-E} \\
 E^-(\pmb{\alpha}, \underline z) &= \exp\left(\sum_{z \ge 1} \left(\frac{p\pmb{\alpha}(-n)}{n}z^{n} + \frac{\bar p \pmb{\alpha}(-n)}{n}\bar z^{n}\right)\right).
 \nonumber
\end{align}
With respect to the tensor product structure $M_{\frh_p}(1, \pmb{\lambda}) \cong M_{p\frh}(1, p\pmb{\lambda}) \otimes M_{\bar p\frh}(1, \bar p\pmb{\lambda})$,
we can write
\begin{align}
 \pmb{\alpha}(\underline z) &= p\pmb{\alpha}(z)\otimes \bb1 + \bb1 \otimes \bar p\pmb{\alpha}(\bar z), \nonumber \\
 E^+(\pmb{\alpha}, \underline z) &= E^+(p\pmb{\alpha}, z) \otimes E^+(\bar p\pmb{\alpha}, \bar z), \label{eq:decomp-E} \\
 E^-(\pmb{\alpha}, \underline z) &= E^-(p\pmb{\alpha}, z) \otimes E^-(\bar p\pmb{\alpha}, \bar z). \nonumber\
\end{align}
With respect to this decomposition,
we can construct a pair of the Virasoro algebra representations by the Sugawara formula,
which we denote $L_0\otimes \bb1, \bb1\otimes L_0$, respectively (where we omitted the dependence on $\pmb{\lambda}$).

Let us write
$\underline z^{\pmb{\alpha}(0)} := z^{p\pmb{\alpha}(0)}\bar z^{\bar p\pmb{\alpha}(0)}$,
where $p\pmb{\alpha}(0)|_{\cK_{p\pmb{\lambda}}} = (p\pmb{\alpha}, p\pmb{\lambda})_{\frh}$ on each $\pmb{\lambda} \in Q$.
We introduce the pre-vertex operators as in \cite[Section 6.4]{AGT23Pointed}:
\begin{align}
 \underline{Y_{p\pmb{\alpha}}}(z) &= E^-(p\pmb{\alpha}, z) E^+(p\pmb{\alpha}, z) z^{p\pmb{\alpha}(0)} \nonumber \\
 \underline{Y_{\bar p\pmb{\alpha}}}(\bar z) &= E^-(\bar p\pmb{\alpha}, \bar z) E^+(\bar p\pmb{\alpha}, \bar z) \bar z^{\bar p\pmb{\alpha}(0)} \nonumber \\
 \underline{Y_{\pmb{\alpha}}}(\underline z) &= E^-(\pmb{\alpha}, \underline z) E^+(\pmb{\alpha}, \underline z) \underline z^{p\pmb{\alpha}(0)}
 \cong \underline{Y_{p\pmb{\alpha}}}(z)\otimes \underline{Y_{\bar p\pmb{\alpha}}}(\bar z) \label{eq:decomp-Y},
\end{align}
where the last tensor product refers again to
$M_{\frh_p}(1, \pmb{\lambda}) \cong M_{p\frh}(1, p\pmb{\lambda}) \otimes M_{\bar p\frh}(1, \bar p\pmb{\lambda})$.

\subsection{Formal Wightman fields}
On $\Psi = (\Psi_{\pmb{\lambda}})_{\pmb{\lambda} \in Q} \in \cK_Q$ (see \eqref{eq:K_Q}),
we take the twisted shift operator $c_{\pmb{\alpha}}$ defined by
\begin{align}\label{eq:def-c}
 (c_{\pmb{\alpha}} \Psi)_{\pmb{\alpha} + \pmb{\lambda}} = \epsilon(\pmb{\alpha}, \pmb{\lambda})\Psi_{\pmb{\lambda}}
\end{align}
(we use a notation different from \eqref{eq:shift-field}, in order to distinguish the specific equivalence classes $\cK_{\pmb{\lambda}}$
from $\cH_{\pmb{\lambda}h}$).
Then it holds that (\cf \cite[Section 6.4]{AGT23Pointed})
\begin{align}\label{eq:comm-c-z}
 c_{\pmb{\alpha}} c_{\pmb{\beta}} &= (-1)^{(\pmb{\alpha}|\pmb{\beta})}c_{\pmb{\beta}} c_{\pmb{\alpha}}, \nonumber \\
 c_{\pmb{\alpha}} \underline z^{\pmb{\beta}(0)} &= z^{-(p\pmb{\alpha}, p\pmb{\beta})_\frh} \bar z^{-(\bar p\pmb{\alpha}, \bar p\pmb{\beta})_\frh}
 \underline z^{\pmb{\beta}(0)}c_{\pmb{\alpha}}.
\end{align}

The formal series for our Wightman fields are given by
\begin{align}\label{eq:def-wightman}
 Y_{\pmb{\alpha}}(\underline z) := c_{\pmb{\alpha}}\underline{Y_{\pmb{\alpha}}}(\underline z).
\end{align}
There are only countably many non-zero Fourier components $Y_{\pmb{\alpha}, r, s}$ of
$Y_{\pmb{\alpha}}(\underline z) = \sum_{r,s} Y_{\pmb{\alpha}, r, s} z^{-r-1}\bar z^{-s-1}$
and when restricted to $M_{\frh_p}(1, \pmb{\lambda})$, its image is $M_{\frh_p}(1, \pmb{\alpha} + \pmb{\lambda})$
and only the components $Y_{\pmb{\alpha}, r, s}$ such that
$r + \frac{(p\pmb{\alpha}, p\pmb{\lambda})_{\frh}}2, s + \frac{(\bar p\pmb{\alpha}, \bar p\pmb{\lambda})_{\frh}}2 \in \bbZ$
are non-zero.

Compared with \cite[Section 6.4]{AGT23Pointed},
our $Y_{\pmb{\alpha}}$ has more tensor components in each chiral and antichiral components,
and $c_{\pmb{\alpha}}$ is twisted by the cocycle $\epsilon$.
Note that the cocycle $\epsilon$ takes values in $\{-1, 1\}$,
and we have the action of the pair of the Virasoro algebras $L_m\otimes \bb1, \bb1\otimes L_m$ on $\cK_Q$,
thus the commutation relations remain the same (\cf \cite[below (6.5)]{AGT23Pointed}):
\begin{align*}
 \pmb{\alpha}(m)c_{\pmb{\beta}} &= c_{\pmb{\beta}}(\pmb{\alpha}(m) + (\pmb{\alpha}, \pmb{\beta})_{\frh}\delta_{m,0}), \\
 L_m\otimes \bb1 \cdot c_{\pmb{\alpha}} &= c_{\pmb{\alpha}}(L_m\otimes \bb1 + p\pmb{\alpha}(m)) + \frac{(p\pmb{\alpha}, p\pmb{\alpha})_{\frh}}2\delta_{m,0}, \\
 \bb1\otimes L_m \cdot c_{\pmb{\alpha}} &= c_{\pmb{\alpha}}(\bb1\otimes L_m + \bar p\pmb{\alpha}(m)) + \frac{(\bar p\pmb{\alpha}, \bar p\pmb{\alpha})_{\frh}}2\delta_{m,0}.
\end{align*}
From this, it is analogous to \cite[Section 6.4]{AGT23Pointed} to show that $Y_{\pmb{\alpha}}$ are primary
with the scaling dimensions $(\frac{(p\pmb{\alpha}, p\pmb{\alpha})_{\frh}}2, \frac{(\bar p\pmb{\alpha}, \bar p\pmb{\alpha})_{\frh}}2)$:
\begin{align}
 [L_m\otimes \bb1, Y_{\pmb{\alpha}}(\underline z)] = \partial_z Y_{\pmb{\alpha}}(\underline z)z^{m+1} + \frac{(p\pmb{\alpha}, p\pmb{\alpha})_{\frh}}2
 (m+1)Y_{\pmb{\alpha}}(\underline z), \label{eq:primary+}\\
 [\bb1\otimes L_m, Y_{\pmb{\alpha}}(\underline z)] = \partial_{\bar z} Y_{\pmb{\alpha}}(\underline z)\bar z^{m+1}
 + \frac{(\bar p\pmb{\alpha}, \bar p\pmb{\alpha})_{\frh}}2 (m+1)Y_{\pmb{\alpha}}(\underline z). \label{eq:primary-}
\end{align}

Using \eqref{eq:def-wightman}\eqref{eq:decomp-Y}\eqref{eq:decomp-E}\eqref{eq:comm-E}\eqref{eq:comm-c-z},
we have (\cf \cite[(6.4)]{AGT23Pointed})
\begin{align}\label{eq:comm-Y}
 &\left(1-\frac{z}w\right)^{-(p\pmb{\alpha}, p\pmb{\beta})_{\frh}}
 \left(1-\frac{\bar z}{\bar w}\right)^{-(\bar p\pmb{\alpha}, \bar p\pmb{\beta})_{\frh}}
 w^{-(p\pmb{\alpha}, p\pmb{\beta})_{\frh}} \bar w^{-(\bar p\pmb{\alpha}, \bar p\pmb{\beta})_{\frh}}
 Y_{\pmb{\alpha}}(\underline w)Y_{\pmb{\beta}}(\underline z) \nonumber \\
 &= (-1)^{(\pmb{\alpha}|\pmb{\beta})}\cdot \left(1-\frac w{z}\right)^{-(p\pmb{\alpha}, p\pmb{\beta})_{\frh}}
 \left(1-\frac{\bar w}{\bar z}\right)^{-(\bar p\pmb{\alpha}, \bar p\pmb{\beta})_{\frh}}
 z^{-(p\pmb{\alpha}, p\pmb{\beta})_{\frh}} \bar z^{-(\bar p\pmb{\alpha}, \bar p\pmb{\beta})_{\frh}}
 Y_{\pmb{\beta}}(\underline z) Y_{\pmb{\alpha}}(\underline w)
\end{align}

\subsection{Wightman fields}
We say that a formal series $A(z)$ or $A(\bar z)$ with coefficients in $\End(M_{p\frh}(1, p\pmb{\lambda}))$ or $\End(M_{\bar p\frh}(1, \bar p\pmb{\lambda}))$
(respectively $B(\underline z)$ with coefficients in
$\End(\bigoplus_{\pmb{\lambda} \in Q} M_{p \frh}(1, p\pmb{\lambda})\otimes M_{\bar p\frh}(1, \bar p\pmb{\lambda}))$)
satisfies polynomial energy bounds
if $\|A_s\Psi\| \le C(|s|+1)^p \|(L_0 + \bb1)^q \Psi\|$ (respectively $\|B_{s,t}\Psi\| \le C(|s| + |t| + 1)^p \|(L_0\otimes \bb1 + \bb1\otimes L_0)^q \Psi\|)$
for some $C, p, q > 0$.

By \cite[Proposition VI.1.2.1]{Toledano-LaredoThesis} \cite[Theorem A.2]{Gui19}, $\underline{Y_{p\pmb{\alpha}}}(z), \underline{Y_{\bar p\pmb{\alpha}}}(\bar z)$
satisfies polynomial energy bounds. Thus by the arguments of \cite[Lemma 5.5]{AGT23Pointed}
$\underline{Y_{\pmb{\alpha}}}(\underline z)$ satisfies polynomial energy bounds as well.

Let us write
\begin{align}\label{eq:Dlambda}
 \scD_{\pmb{\lambda}} = C^\infty(L_0|_{\cK_{p\pmb{\lambda}}}\otimes \bb1 + \bb1\otimes L_0|_{\cK_{\bar p\pmb{\lambda}}})
 \subset \cK_{p\pmb{\lambda}}\otimes \cK_{\bar p\pmb{\lambda}}\cong \cK_{\pmb{\lambda}}
\end{align}
and $\scD_Q = \bigoplus_{\pmb{\lambda} \in Q} \scD_{\pmb{\lambda}}$ (the algebraic direct sum).
By polynomial energy bounds, the Fourier components of $Y_{\pmb{\alpha}}(\underline z)$,
which are first defined on $M_{p\frh}(1, p\pmb{\lambda}) \otimes M_{\bar p\frh}(1, \bar p\pmb{\lambda})$,
extend to $\scD_{\pmb{\lambda}}$, thus to $\scD_Q$.
Moreover, for a test function $f$ on $(-\pi, \pi)\times (-\pi, \pi)$, define
\begin{align*}
 f_{r,s} = \frac1{(2\pi)^2}\int_{-\pi}^{\pi}\int_{-\pi}^{\pi} f(\theta_+, \theta_-)\rme^{-\rmi (r\theta_+ + s\theta_-)}d\theta_+d\theta_-.
\end{align*}
Then the sum
\begin{align*}
 Y_{\pmb{\alpha}}(f) := \sum_{r,s \in \bbR} Y_{\pmb{\alpha}, r, s} f_{r,s}
\end{align*}
converges on each $\scD_{\pmb{\lambda}}$ of \eqref{eq:Dlambda}, thus on $\scD_Q$.

Recall that we defined conformal Wightman fields on the Einstein cylinder in \cite[Section 5.1]{AGT23Pointed},
by requiring \textbf{locality}, \textbf{diffeomorphism covariance}, \textbf{positivity of energy},
\textbf{vacuum and the Reeh-Schlieder property} and \textbf{polynomial energy bounds}
for a family of such fields.
Among them, the cyclicity of vacuum is required for a family of fields and not for single fields (see \cite[(2dW4)]{AGT23Pointed}).
For simplicity, we include the fields $\pmb{\alpha}(f), \pmb{\alpha} \in \frh_p$.

\begin{theorem}
 The family $Y_{\pmb{\alpha}}(f), \pmb{\alpha} \in Q$ with $\pmb{\alpha}(f), \pmb{\alpha} \in \frh_p$,
 together with $U_Q = \bigoplus_{\pmb{\lambda} \in Q} U_{p\pmb{\lambda}}\otimes U_{\bar p\pmb{\lambda}}$ (\cf Section \ref{2dlatticenet})
 and $\Omega_{\pmb{0}}$,
 satisfies the above axioms for two-dimensional conformal Wightman fields.
\end{theorem}
\begin{proof}
 As most of the proofs are analogous to \cite[Theorem 5.9, Theorem 5.7]{AGT23Pointed}, we will be brief
 except for locality.
 We have checked polynomial energy bounds of all these fields.
 Positivity of energy follows by construction, \cf Theorem \ref{th:2dnetQ}.
 Moreover, as we have (2dCN\ref{2dcn:cylinder}), by covariance all the fields extend to the Einstein cylinder.
 As we include chiral and antichiral fields $\pmb{\alpha}(f)$, they span $\cK_{\pmb{0}}$ from $\Omega_Q$,
 and as $Y_{\pmb{\alpha}}(f)$ maps $\cK_{\pmb{\lambda}}$ to $\cK_{\pmb{\alpha} + \pmb{\lambda}}$,
 their image exhausts $\cK_Q$.

 To verify diffeomorphism covariance, it is enough to consider diffeomorphisms of the form
 $\gamma_+\times \iota$ or $\iota \times \gamma_-$. For this, we apply the arguments of \cite[Proposition 6.4]{CKLW18},
 \cite[Lemma 5.4]{AGT23Pointed}: from \eqref{eq:primary+}
 the commutation relations with the scaling dimension $\frac{(p\pmb{\alpha}, p\pmb{\alpha})_{\frh}}2$ follow:
 \begin{align*}
  \rmi[L_1(f_1)\otimes\bb1, Y_{\pmb{\alpha}}(f_2)] = Y_{\pmb{\alpha}}(\textstyle{\frac{(p\pmb{\alpha}, p\pmb{\alpha})_{\frh}}2}f_1'f_2 - f_1\partial_{\theta_+}f_2),
 \end{align*}
 where $f_1 \in C^\infty_0(S^1)$, $f_2$ is a compactly supported test function on $(-\pi, \pi)\times (-\pi, \pi)$.
 From this, the covariance with respect to $U(\gamma_+\times\iota)$ follows exactly as in \cite[Lemma 5.4]{AGT23Pointed}.
 The covariance with respect to $U(\iota\times\gamma_-)$ is similar.
 The covariance of the chiral and antichiral fields are known, \eg \cite[Proposition 6.4]{CKLW18}.

 Locality between $\pmb{\alpha}(f), \pmb{\beta}(g)$
 and between $\pmb{\alpha}(f), Y_{\pmb{\beta}}(g)$ is analogous to \cite{AGT23Pointed}.
 In order to show locality between $Y_{\pmb{\alpha}}(f)$ and $Y_{\pmb{\beta}}(g)$,
 we take technical elements of \cite[Lemma 5.3]{AGT23Pointed}:
 On $\bigoplus_{\pmb{\lambda} \in Q} \cK_{p\pmb{\lambda}}$, with $c'_{p\pmb{\alpha}}$ the simple shift operator by $p\pmb{\alpha}$
 (without cocycle), we have the commutation relations between formal series, just as we did for $Y_{\pmb{\alpha}}(\underline z)$
 (even more easily)
 \begin{align*}
   \left(1-\frac{z}w\right)^{-(p\pmb{\alpha}, p\pmb{\beta})_{\frh}}
   w^{-(p\pmb{\alpha}, p\pmb{\beta})_{\frh}}
  c'_{p\pmb{\alpha}}\underline{Y_{p\pmb{\alpha}}}(w)\cdot c'_{p\pmb{\beta}}\underline{Y_{p\pmb{\beta}}}(z)
  &= \left(1-\frac{w}z\right)^{-(p\pmb{\alpha}, p\pmb{\beta})_{\frh}}
   z^{-(p\pmb{\alpha}, p\pmb{\beta})_{\frh}}
  c'_{p\pmb{\beta}}\underline{Y_{p\pmb{\beta}}}(z)\cdot c'_{p\pmb{\alpha}}\underline{Y_{p\pmb{\alpha}}}(w).
 \end{align*}
 By the same arguments as in \cite[Lemma 5.3]{AGT23Pointed},
 this translates into the commutation relations between operators:
 Let $f_+, g_+ \in C^\infty_0((-\pi, \pi))$ such that $\supp f_+ < \supp g_+$, then
 \begin{align*}
  c'_{p\pmb{\alpha}}\underline{Y_{p\pmb{\alpha}}}(f_+)c'_{p\pmb{\beta}}\underline{Y_{p\pmb{\beta}}}(g_+)
  &= \lim_{\mathrm{Im}\, \zeta > 0, \zeta\to -1}\zeta^{-(p\pmb{\alpha}, p\pmb{\beta})_{\frh}}
  \cdot
  c'_{p\pmb{\beta}}\underline{Y_{p\pmb{\beta}}}(g_+)\underline{Y_{p\pmb{\alpha}}}(f_+),
 \end{align*}
 where $\zeta^s$ is defined as the analytic continuation of positive roots on $\bbR_{>0}$ to $\bbC\setminus \bbR_{\le 0}$.
 These relations hold even when restricted to $\cK_{\pmb{p\lambda}}$,
 and thus to $\cK_{\pmb{p\lambda}}\otimes \cK_{\bar p\pmb{\lambda}}$.
 Similarly, if $\supp f_- > \supp g_-$,
 \begin{align*}
  c'_{\bar p\pmb{\alpha}}\underline{Y_{\bar p\pmb{\alpha}}}(f_-)c'_{\bar p\pmb{\beta}}\underline{Y_{\bar p\pmb{\alpha}}}(g_-)
  &= \lim_{\mathrm{Im}\, \zeta < 0, \zeta\to -1}\zeta^{-(\bar p\pmb{\alpha}, \bar p\pmb{\beta})_{\frh}}
  \cdot
  c'_{p\pmb{\beta}}\underline{Y_{\bar p\pmb{\beta}}}(g_-)c'_{p\pmb{\alpha}}\underline{Y_{\bar p\pmb{\alpha}}}(f_-).
 \end{align*}

 On $\cK_{\pmb{\lambda}}$, we have
 $Y_{\pmb{\alpha}}(\underline z) = c_{\pmb{\alpha}} \underline{Y_{p\pmb{\alpha}}}(z)\otimes \underline{Y_{\bar p\pmb{\alpha}}}(\bar z)$,
 where $c_{\pmb{\alpha}}$ is the shift $c'_{p\pmb{\alpha}} \otimes c'_{p\pmb{\alpha}}: \cK_{\pmb{\lambda}} \cong\cK_{\pmb{p\lambda}}\otimes \cK_{\bar p\pmb{\lambda}}\to \cK_{\pmb{\alpha} + \pmb{\lambda}}\cong
 \cK_{p(\pmb{\alpha} + \pmb{\lambda})}\otimes \cK_{\bar p(\pmb{\alpha} + \pmb{\lambda})}$
 followed by the scalar $\epsilon(\pmb{\alpha}, \pmb{\lambda})$.

 Recall that $c_{\pmb{\alpha}}c_{\pmb{\beta}} = (-1)^{(\pmb{\alpha}|\pmb{\beta})}c_{\pmb{\beta}}c_{\pmb{\alpha}}$.
 Altogether,
 for test functions $f,g \in C^\infty_0((-\pi, \pi)\times (-\pi, \pi))$ with compact supports
 such that $\supp f$ is on the left of $\supp g$,
 \eqref{eq:comm-Y}
 implies on each $\cK_{\pmb{\lambda}}$ that
 \begin{align*}
  Y_{\pmb{\alpha}}(f)Y_{\pmb{\beta}}(g)
  &= (-1)^{(\pmb{\alpha}|\pmb{\beta})}\cdot
  \lim_{\mathrm{Im}\, \zeta > 0, \zeta\to -1}\zeta^{-(p\pmb{\alpha}, p\pmb{\beta})_{\frh}}\cdot
  \lim_{\mathrm{Im}\, \zeta < 0, \zeta\to -1}\zeta^{-(\bar p\pmb{\alpha}, \bar p\pmb{\beta})_{\frh}}\cdot
  Y_{\pmb{\beta}}(g) Y_{\pmb{\alpha}}(f) \\
  &= (-1)^{(\pmb{\alpha}|\pmb{\beta})}\cdot
  \lim_{\mathrm{Im}\, \zeta > 0, \zeta\to -1}\zeta^{-(p\pmb{\alpha}, p\pmb{\beta})_{\frh}}\cdot
  \lim_{\mathrm{Im}\, \zeta > 0, \zeta\to -1}\zeta^{(\bar p\pmb{\alpha}, \bar p\pmb{\beta})_{\frh}}\cdot
  Y_{\pmb{\beta}}(g) Y_{\pmb{\alpha}}(f) \\
  &= (-1)^{(\pmb{\alpha}|\pmb{\beta})}\cdot
  (-1)^{-(\pmb{\alpha}|\pmb{\beta})}
  Y_{\pmb{\beta}}(g) Y_{\pmb{\alpha}}(f) \\
  &= Y_{\pmb{\beta}}(g) Y_{\pmb{\alpha}}(f),
 \end{align*}
 where in the second equality we calculate the limit on the circle
 $S^1 \setminus \{-1\}$ and used the fact that the limit $\zeta \to -1, \mathrm{Im}\, \zeta < 0$
 amounts to the limit $\zeta^{-1} \to -1, \mathrm{Im}\, \zeta > 0$.

\end{proof}

\begin{example}
 Let us take the case of Example \ref{ex:RR2net}.
 The restriction of $\epsilon(\pmb{\alpha}, \pmb{\lambda})^{-1}Y_{\pmb{\alpha}}(\underline z)$ to $M_{\frh_p}(1, \pmb{\lambda})\subset \cK_{\pmb{\lambda}}$
 has the form
 \begin{align*}
  \epsilon(\pmb{\alpha}, \pmb{\lambda})^{-1}Y_{\pmb{\alpha}}(\underline z)
  = Y_{p\pmb{\alpha}}(z) \otimes Y_{\bar p\pmb{\alpha}}(\bar z)
 \end{align*}
 and the tensor component $Y_{p\pmb{\alpha}}(z)$ can be identified (up to a change of notations) with
 $\cY_{p\pmb{\alpha}}(c_{p\pmb{\alpha}}, x)$ of \cite[Theorem A.2(a)]{Gui19}.
 Thus its Fourier components are bounded uniformly in $\pmb{\lambda}$, say by $C$,
 if $(p\pmb{\alpha}, p\pmb{\alpha})_\frh \le 1$ (similar for $Y_{\bar p\pmb{\alpha}}(\bar z)$).

 If we have both $(p\pmb{\alpha}, p\pmb{\alpha})_\frh \le 1, (\bar p\pmb{\alpha}, \bar p\pmb{\alpha})_\frh \le 1$,
 then  the smeared two-dimensional field $Y_{\pmb{\alpha}}(f)$ is bounded (\cf \cite{Rehren97}).
 This holds if $\pmb{\alpha} = \frac{R}{\sqrt 2}\oplus \frac{R}{\sqrt 2}$ with $R^2 \le 2$
 or $\pmb{\alpha} = \frac{R^{-1}}{\sqrt 2}\oplus \frac{R^{-1}}{\sqrt 2}$ with $R^{-2} \le 2$.
 Therefore, if $\frac1{\sqrt 2} \le R \le \sqrt 2$,
 they together generate the two-dimensional conformal net.
 By the uniqueness result, Theorem \ref{th:classification}, the generated net must be $\cA_Q$ of Theorem \ref{th:2dnetQ}.
\end{example}


\subsubsection*{Acknowledgements}
We thank Sebastiano Carpi and Yasuyuki Kawahigashi for discussions about braided equivalences,
Yuto Moriwaki for informing us of his results in \cite{Moriwaki23}
and Mizuki Oikawa for discussions regarding \cite{BKL15}.

M.S.A.\! is a Humboldt Research Fellow supported by the Alexander von Humboldt Foundation
and is partially supported by INdAM--GNAMPA Project ``Ricostruzione di segnali, tramite operatori e frame, in presenza di rumore'' CUP E53C23001670001 and GNAMPA--INdAM.
L.G.\! and Y.T.\! acknowledge support from the GNAMPA--INdAM project \lq\lq Operator algebras
and infinite quantum systems", CUP E53C23001670001, project GNAMPA 2024 \lq\lq Probabilit\`a Quantistica e Applicazioni", CUP E53C23001670001, 
from the MIUR Excellence Department Project MatMod@TOV awarded to the Department of Mathematics, University of Rome ``Tor Vergata'', CUP E83C23000330006
and from the Tor Vergata University of Rome funding OANGQM, CUP E83C25000580005.

{\small
\def\polhk#1{\setbox0=\hbox{#1}{\ooalign{\hidewidth
  \lower1.5ex\hbox{`}\hidewidth\crcr\unhbox0}}} \def\cprime{$'$}

}

\end{document}